\documentclass[12pt,letterpaper]{article}
\usepackage{amssymb,amsmath,amsthm, enumerate, float,tocbibind}
\usepackage[margin=1in]{geometry}
\usepackage{mathtools}
\usepackage[colorlinks=true,linkcolor=blue,filecolor=blue,urlcolor=blue,citecolor=blue]{hyperref}
\usepackage{color}
\usepackage{comment}

\usepackage[final]{graphicx}
\usepackage[hang,small,bf]{caption}
\usepackage{bbm}

\usepackage[style=authoryear, uniquename = false, doi = false, url = false, maxnames =100]{biblatex}
\usepackage{setspace}

\newtheorem{assumption}{Assumption}
\newtheorem{proposition}{Proposition}
\newtheorem{remark}{Remark}
\newtheorem{lemma}{Lemma}
\newtheorem{dfn}{Definition}
\newtheorem{theorem}[lemma]{Theorem}
\newtheorem{corollary}[lemma]{Corollary}

\makeatletter
  \renewcommand{\theequation}{%
  \thesection.\arabic{equation}}
  \@addtoreset{equation}{section}
\makeatother
\allowdisplaybreaks

\newcommand{\argmin}{\mathop{\rm arg~min}\limits}
\newcommand\eqd{\overset{d}{=}}

\title{A Uniform Confidence Band for the Marginal Treatment Effect Function\footnote{We received valuable comments from Michal Kolesár (Editor), an associate editor, two anonymous reviewers, the participants at the Kansai Econometric Workshop, Tsinghua, AMES2024V, LMU-Todai workshop, PCIC2025, and SETA2025. We acknowledge financial support from the Japan Society for the Promotion of Science under KAKENHI Grant Nos. 22K20154, 23H00804, 23K25501 and 23KJ0713. All remaining errors are ours.}}

\author{Toshiki Tsuda,\thanks{Department of Economics, Yale University. 28 Hillhouse Ave. New Haven, CT 06511. U.S.A. Email: \href{toshiki.tsuda@yale.edu}{toshiki.tsuda@yale.edu}}\ \  Yanchun Jin,\thanks{Independent researcher, \href{chnjyc@gmail.com}{chnjyc@gmail.com}}\ \ and Ryo Okui\thanks{Graduate School of Economics and Faculty of Economics, University of Tokyo, 7-3-1 Hongo, Bunkyou-ku, Tokyo, 113-0033, Japan. Email:\href{mailto:okuiryo@e.u-tokyo.ac.jp}{okuiryo@e.u-tokyo.ac.jp}}}

\date{\today}

\begin{document}
	
	\maketitle
	
	\begin{abstract}

This paper presents a method for constructing uniform confidence bands for the marginal treatment effect (MTE) function. The shape of the MTE function provides insight into how the unobserved propensity to receive treatment relates to the treatment effect. Our approach visualizes the statistical uncertainty of an estimated function, facilitating inferences about the function's shape. The proposed method is computationally inexpensive and requires only minimal information: sample size, standard errors, kernel function, and bandwidth. 
We derive a Gaussian approximation for a local quadratic estimator and consider the approximation of the distribution of its supremum in polynomial order. Monte Carlo simulations demonstrate that our bands provide the desired coverage and are less conservative than those based on the Gumbel approximation. An empirical application based on the rural electrification program is included.

		\bigskip
		
		\noindent \textit{Keywords}: uniform confidence band; marginal treatment effect; instrumental variables, empirical process; Gaussian approximation.
		
		\bigskip 
		
		\noindent \textit{JEL Classification}: C12, C14, C21, C26.
		
	\end{abstract}

\begin{refsection}

\section{Introduction}
 The marginal treatment effect (MTE) provides information about how the likelihood of receiving treatment relates to treatment effects \parencite{bjorklund1987estimation, heckman1999local, heckman2005structural}. We consider settings in which unobserved factors influencing selection into treatment are correlated with the potential outcomes and instrumental variables (IVs) are available. IVs affect the selection into treatment without directly affecting the outcomes. MTE represents the benefits to individuals on the cusp of deciding whether or not to participate in treatment. 
 The shape of the MTE function provides insight into how the unobservable propensity is related to the treatment effect. 
 For instance, when the MTE function decreases, those with a high propensity to receive treatment experience a relatively large treatment effect compared with those with a low propensity. 

 MTE functions have been examined in various empirical applications. For instance, \textcite{carneiro2011marginalretrurn} finds that the MTE function for college attendance is decreasing. This result suggests that individuals who benefit from college education are more likely to attend college. The estimated MTE function for the effect of early childcare by \textcite{YamaguchiAsaiKambayashi18} is increasing. It suggests that ``children who would benefit most from childcare are less likely to attend, implying inefficient allocation'' \parencite[][p 56]{YamaguchiAsaiKambayashi18}. 
 Other empirical examples of MTE include \textcite{doyle2007child},
 \textcite{kasahara2016does}, 
\textcite{brinch2017beyond},
and
\textcite{arnold2022measuring}.
\textcite{cornelissen2016from} provide an empirically oriented review.
\textcite[][page 3]{shea2023ivmte} provide a wide list of empirical articles that use MTE.

 This paper proposes a method for constructing uniform confidence bands for the MTE function. A uniform confidence band covers the true function with a prespecified probability. These bands are used to make statistical inferences on the shape of the MTE function. Uniform confidence bands are particularly useful because they enable the visualization of the statistical uncertainty underlying the estimated MTE function. All the above empirical articles present pointwise confidence intervals using the bootstrap method. A pointwise confidence interval quantifies the statistical uncertainty at a point of the MTE function, not the function itself. Pointwise procedures are not valid when evaluating the statistical uncertainty in the shape of the MTE function. For this purpose, uniform confidence bands are needed. 

Our proposed method is straightforward to implement. The uniform confidence band takes a familiar form: "estimated function $\pm$ critical value $\times$ standard error." This format makes our confidence bands easy to visualize. The critical values can be obtained analytically,  and their computation requires only the sample size, the kernel function, and the bandwidth used to estimate the MTE function. Our approach does not rely on computationally intensive procedures, such as bootstrapping, making it computationally efficient.

Our construction of uniform confidence bands is based on an asymptotic approximation of the supremum of the estimated MTE function over its domain of interest.
We follow \textcite{heckman2005structural}  for the setting. The semiparametric estimation of the finite-dimensional parameters is based on \textcite{carneiro2009estimating}. The local quadratic estimator estimates the MTE function as in \textcite{heckman2006understanding}.
We approximate the supremum of a normalized version of the estimated function by the supremum of the Gaussian process using the result of \textcite{chernozhukov2014gaussian}. Then, we further approximate the Gaussian field by a stationary Gaussian field by the result of \textcite{ghosal2000testing}. 
We derive critical values using an analytic approximation of the distribution of the supremum of a stationary Gaussian process based on the arguments of \textcite{piterbarg1996asymptotic}. The theoretical analysis follows steps similar to those of \textcite{lee2017doubly}. However, our derivation requires non-trivial technical analysis because the MTE function is estimated as the first derivative of the nonparametric regression function, and the nonparametric estimation involves a generated regressor. 

As a key theoretical contribution, we provide a rigorous treatment of the generated regressor problem that arises when using the estimated propensity score in the MTE estimator. Because the nonparametric regression uses the estimated rather than the true propensity score as the regressor, the estimation error in the first step must be carefully controlled. We address this issue through two complementary arguments. First, we show that, under our regularity conditions, the first-step propensity score estimator is sufficiently well behaved so that the bias induced by its estimation is asymptotically negligible. Second, rather than treating the estimated propensity score as a conventional generated covariate, we exploit that it is an estimated function of the instruments and treat it as an infinite-dimensional nuisance component. This perspective allows us to analyze the estimator within a nonparametric framework with an infinite-dimensional nuisance parameter, namely, the estimator is viewed as an empirical process indexed by the propensity score. As a result, we establish that the nonparametric regression based on the estimated propensity score admits a Gaussian process approximation, which provides the theoretical foundation for constructing valid uniform confidence bands for the MTE function.

This new approximation is more precise than the existing Gumbel-based approximation. It is well known that the supremum of a Gaussian process asymptotically follows the Gumbel distribution \parencite{cramer67}. This approximation is used for statistical inference based on kernel estimation on a nonparametric function by \textcite{bickel1973global, johnston1982probabilities, hardle1989asymptotic}. However, the Gumbel approximation is only accurate up to logarithmic order. In contrast, our approximation is valid up to polynomial order, yielding more accurate results. Moreover, it also offers tighter confidence bands than those derived from the Gumbel approximation.

We examine finite sample properties of the proposed uniform confidence band through Monte Carlo simulations. The results demonstrate that our confidence bands cover the true function with probability similar to the prespecified confidence level. Moreover, the simulation results show the superiority over alternative methods. Pointwise confidence bands yield much lower coverage, indicating that they are inadequate for making inferences about the shape of the MTE function. We also find that, although uniform confidence bands are wider than pointwise confidence bands, they remain sufficiently informative. The critical value for a 95\% confidence band is 1.96 for the pointwise procedure. It is around 3.03 for our procedure in our simulation designs. 
The uniform confidence bands based on the Gumbel approximation are much broader and may not be informative. The 95\% critical value from the Gumbel approximation is around 4.08. 

We illustrate the procedure by revisiting an empirical study by \textcite{dinkelman2011effects} about the rural electrification program, using the continuously distributed average land gradient as the instrument. The estimated MTE function displays a mild downward trend over much of the propensity-score range.

The remainder of the paper is organized as follows. In the following subsection, we discuss the relation to the literature. 
In Section \ref{framework}, we describe the econometric model and discuss the identification of the MTE function.
In Section \ref{estimation}, we present a semiparametric estimation method for the MTE function.
In Section \ref{confidenceband}, we construct the uniform confidence band for the MTE function.
Section \ref{asymptotic} provides the asymptotic justification of the proposed procedure. Section \ref{Monte_Carlo} demonstrates our uniform confidence band's finite sample performances compared with other methods via Monte Carlo simulations. Section \ref{empirical_application} illustrates our method with an empirical application. Section \ref{conclusion} concludes. The Online Supplement includes technical appendices and additional simulation and empirical results.

 \subsection{Related literature}

This study contributes to the literature on the MTE function for causal inference.
As discussed in the introduction, the shape of the MTE function is important in empirical applications. Many empirical papers include pointwise confidence bands when plotting the estimated MTE function \parencite[See, for example,][]{carneiro2009estimating}. This paper presents a straightforward method for computing uniform confidence bands, thereby enabling statistical inference about the shape of the MTE function. We also note that the shape of the MTE function is important not only for policy evaluation but also for optimal policy assignment, as discussed in \textcite{chen2022personalized} and \textcite{liu2022policy}.

This paper contributes to the literature on uniform confidence bands for functions estimated nonparametrically. Recent works on uniform confidence bands for nonparametric models include \textcite{horowitz2012uniform}, \textcite{chernozhukov2014gaussian}, \textcite{chernozhukov2014anti}, \textcite{belloni2015new}, \textcite{horowitz2017nonparametric}, \textcite{chen2018optimal}, and \textcite{chen2025adaptive}, among others.
As mentioned above, our setting involves a generated regressor, which introduces an additional source of estimation error that is absent in the existing literature.

The additional difference is that we provide an analytical procedure for uniform inference, whereas the above methods are primarily computationally intensive, such as bootstrapping. A bootstrap procedure should be possible by extending the methods in the aforementioned papers; in our Monte Carlo simulations, a bootstrap-based method (whose validity we do not establish) performs comparably to ours. We nonetheless focus on our straightforward, computationally inexpensive approach.

In the causal inference literature, a uniform confidence band is primarily considered for conditional average treatment effect functions. See, for example, \textcite{lee2017doubly}, \textcite{semenova2021debiased}, \textcite{fan2022estimation}, \textcite{baybutt2023doubly}, and \textcite{Imai03042026}.
We largely follow the approach taken by \textcite{lee2017doubly}. There are two critical differences. First, our target function is not a conditional mean function, but rather its derivative, which complicates the theoretical analysis. Second, the argument of our target function is the propensity score, which must be estimated; we therefore must account for the estimation error in the propensity score. This second difference poses technical challenges and constitutes our technical contribution.

\subsection{Notation} 
Let $: = $ denote ``equals by definition,'' and let a.s. denote ``almost surely.'' Let $\mathbbm{1}\{\cdot\}$ denote the indicator function. For a random variable $X$, $f_{X}(\cdot)$ denotes the probability density function of $X$. $\|f\|_{\infty}$ denotes $\sup_{x\in\mathcal{X}}|f(x)|$, where $\mathcal{X}$ is the support of $f$. For an arbitrary set $T$, let $\ell^{\infty}(T)$ denote the class of functions $f$ that maps from $T$ to $\mathbb{R}$ such that $\|f\|_{\infty}<\infty$. Unless otherwise stated, $c>0$ refers to universal constants whose values may change from place to place.

\section{Framework}\label{framework}
This section presents the econometric framework and defines and identifies the MTE function.
Our discussion is based on the standard potential outcome framework with selection on unobservables. In particular, we follow  \textcite{heckman2005structural}.

Let $Y_1$ be the potential outcome under treatment, and $Y_0$ be the outcome without treatment. The treatment $D$ is a binary variable. We set $D=1$ when the treatment is received, and in this case, $Y_1$ is observed. Conversely, when the treatment is not received, $ D=0$, and $Y_0$ is observed. 
Thus, the observed outcome is 
\begin{align*}
Y=DY_1+(1-D)Y_0.
\end{align*}

The potential outcomes are determined by observable covariates $X$ and unobservable factors $U_1$ and $U_0$. 
We write:
\begin{align*}
Y_1=\mu_1(X, U_1), \text{ and } Y_0=\mu_0(X, U_0),
\end{align*}
where $\mu_1$ and $\mu_0$ are unknown functions. 

We assume the individual selects the treatment status according to the value of a vector of instruments $Z$ and an unobserved characteristic $U_D$: 
\begin{align*}
&D^*=\mu_D(Z)-U_D\,\
&D=1 \textit{ if } D^*\geq 0;\ D=0 \textit{ otherwise},
\end{align*}
where $\mu_D$ is an unknown function.  We assume that, without loss of generality, $Z$ includes all elements of $X$. 

We allow for the possibility that selection into the treatment may be correlated with $(U_1, U_0)$ even after conditioning on $X$. This means that $ (U_1, U_0) $ and $ U_D$ may be correlated. We also assume that the instrument vector $Z$ is independent of the unobservable variables $(U_1, U_0, U_D)$. At least one element of $Z$ is not an element of $X$. This excluded instrument affects whether a unit receives the treatment, but does not directly affect the potential outcomes.

The selection process can be summarized using the propensity score, defined as the probability of receiving treatment based on the exogenous variables. Let $F_{U_D}$ denote the distribution function of $U_D$. We assume that $F_{U_D}$ is strictly increasing. Consequently, it holds $V:= F_{U_D} (U_D) \sim U(0,1)$. Let $P(Z) = F_{U_D} (\mu_D (Z) )$. Note that $D= \mathbbm{1} \{ P(Z) \geq V \} $.
It holds that $ E (D | Z) = P (D=1 | Z) = P (Z) $. Thus, $P(Z)$ is the propensity score.

The MTE is the mean effect of the treatment conditional on observed characteristics $X$ at a particular value of $V$:
\begin{align*}
MTE(v, x):= E[Y_1-Y_0| V=v, X=x].
\end{align*}
The MTE function is a function of $v$ for a given value of $x$, and captures how the unobserved propensity $V$ relates to the treatment effect. A small value of $V$ corresponds to a high propensity to receive the treatment, and vice versa. Thus, a decreasing MTE function indicates that those with a higher propensity to receive the treatment exhibit larger treatment effects.

For the identification, following \textcite{heckman2005structural}, we assume 
\begin{align*}
 (U_0, U_1, V) \perp\!\!\perp Z|X,
\end{align*}
$ E[Y_1]<\infty$ and $ E[Y_0]< \infty.$
Then, MTE can be identified by
\begin{align*}
MTE(p,x)= \frac{\partial E[Y|P(Z)=p, X=x]}{\partial p}
\end{align*}
over the support of distribution of $P(Z)$ conditional on $X=x$.

We make the following assumptions to simplify the MTE formula, thereby facilitating estimation and statistical inference.
First, the potential outcomes are additive in $X$ and $(U_1, U_0)$. That is, we write 
\begin{align*}
Y_1=\mu_1(X, U_1) = \mu_1 (X) + U_1, \text{ and }
Y_0=\mu_0(X, U_0) = \mu_0 (X) + U_0,
\end{align*}
where $\mu_j (X) = E[ Y_j | X]$ for $j=0,1$.

\begin{assumption}\label{ass::semipara}\
    \begin{enumerate}
        \item \label{ass::semipara_indep_X} 
$E[Y_j|X,V]=\mu_j(X)+E[U_j|V]$ holds where $\mu_j (X) = E[ Y_j | X]$ for $j=0,1$. 
        \item \label{ass::semipara_linear} $\mu_{j}(X)=\beta_{j}^\prime X$ for $j=0,1$.
    \end{enumerate}
\end{assumption}
Assumption \ref{ass::semipara}.\ref{ass::semipara_indep_X} states that the conditional expectation of unobservables does not depend on covariates once we condition on $V$. \textcite{brinch2017beyond} employ this assumption to simplify the formulation of the MTE function and its statistical analysis. Assumption \ref{ass::semipara}.\ref{ass::semipara_indep_X} is weaker than the independence between the covariates and all unobservables, which is a standard assumption in applied work such as \textcite{carneiro2009estimating} and \textcite{carneiro2011marginalretrurn}.  Assumption \ref{ass::semipara}.\ref{ass::semipara_linear} imposes the linearity in the regression of a potential outcome on the covariates. Taken together, Assumption \ref{ass::semipara} implies the MTE has a partially linear structure, i.e.
\begin{equation*}
    E[Y_1-Y_0|X,V]=(\beta_1-\beta_0)^{\prime}X+E[U_1-U_0|V]
\end{equation*}
and this structure enables us to obtain a tractable estimator and confidence-band procedure.

Under these assumptions, the MTE function is separable in $X$ and $V$. A straightforward calculation yields
\begin{align}
E[Y|X,P(Z)]  =\mu_{0}(X)+(\mu_{1}(X)-\mu_{0}(X))P(Z)+\Lambda (P(Z), X),  \notag
\end{align}
where we define
\begin{align*}
    \Lambda (p,x):=\int^{p}_{0}E[U_{1}-U_{0}|V=v,X=x]dv.
\end{align*}
Under Assumption \ref{ass::semipara}, $\Lambda(p,x)$ does not depend on $x$, so we write it as $\Lambda(p)$. 
We thus have 
\begin{align}
E[Y|X,P(Z)]  =\beta_0' X+ (\beta_1 - \beta_0)'X P(Z)+\Lambda (P(Z)) .
  \label{conditional_Y_estimation}  
\end{align}
The MTE function is 
\begin{align*}
    MTE (p, x) = (\beta_1 - \beta_0)' x + \left.\frac{\partial \Lambda (P(Z))}{\partial P(Z)}\right|_{P(Z)=p}.
\end{align*}
In this paper, we consider this characterization of the MTE function and propose a method for constructing uniform confidence bands for it.

\section{Estimation}\label{estimation}
This section provides a semiparametric estimation method of the MTE. Our estimation strategy combines the semiparametric estimation approach of \textcite{carneiro2009estimating} with the local quadratic estimation method of \textcite{heckman2006understanding}. We assume an independent and identically distributed sample,  $(Y_i,X_i,Z_i,D_i), i=1,\dots, n$, of $(Y, X, Z, D)$ is available, where $X_i\in\mathbb{R}^d$ and $Z_{i}\in\mathbb{R}^m$. We also assume that $Z_{i}$ has enough variation to generate a propensity score $P(Z_{i})$ continuously from $0$ to $1$ conditional on $X_{i}=x$ for all values of $x$ in the region we want our confidence band to cover.

We first estimate the propensity score $P(Z_i)$. Note that it is a binary choice problem where $D$ is the binary outcome and $Z_i$ is the vector of regressors. We often consider parametric models for $P(Z_i)$ and may use the logit or probit models. Let $\hat P(Z_i)$ be the estimated propensity score for $i$.

We consider a modified version of the partially polynomial estimator considered in \textcite{carneiro2009estimating} for the estimation of $\beta_{0}$ and $\beta_{1}$. Let $\hat \beta_0$ and $\hat \beta_1 $ denote the estimators. Our theory is agnostic about the choice of coefficient estimator, provided its convergence rate is sufficiently fast. We suggest applying the partially linear estimator to the regressions of $D_iY_i$ and $(1-D_i)Y_i$. 
A straightforward calculation gives
\begin{align*}
    E[D_{i}Y_{i}|X_{i},P(Z_{i})]=&\beta_1' X_i P(Z_{i})+ \Lambda_{1}(P(Z_{i})), \\
    E[(1-D_{i})Y_{i}|X_{i},P(Z_{i})]=&\beta_0' X_i (1-P(Z_{i})) + \Lambda_{0}(P(Z_{i})),
\end{align*}
where
\begin{align*}
    \Lambda_{1}(P(Z_{i})):=&E[D_{i}U_{1i}|P(Z_{i})], \\
    \Lambda_{0}(P(Z_{i})):=&E[(1-D_{i})U_{0i}|P(Z_{i})].
\end{align*}
The estimation of each of these two partially linear regression models is carried out in two steps. In the first step, we use the local linear estimator for the regression of $(D_i Y_i, X_i \hat P(Z_i))$ on $\hat P(Z_i)$. In the second step, we regress the residual of $ D_iY_i$ on the residual of $X_i\hat P(Z_{i})$ to obtain $\hat \beta_1$. The estimator $\hat \beta_0$ is obtained by applying the same procedure to the regression model for $(1-D_i) Y_i$.\footnote{The original procedure in \textcite{carneiro2009estimating} applies the partially linear estimator to the regression of $Y_i$ on $X_i$ and $\hat P(Z_i)$ separately for the observations with $D_i=0$ and $D_i =1$. Our modification enables the use of all observations to estimate each equation. Alternatively, we may estimate $\beta_0$ and $\beta_1 - \beta_0$ simultaneously from the regression model of $E[ Y\mid X, P(Z)]$ as in \textcite{heckman2006understanding}. Extending our procedure to these alternative approaches is straightforward.}

Lastly, we estimate the derivative of $\Lambda (\cdot)$. 
Let $\tilde{Y}_{i}=Y_{i}- \beta_{0}^\prime X_{i}-(\beta_{1}-\beta_{0})^\prime X_{i}P(Z_{i})$.
From (\ref{conditional_Y_estimation}), we obtain $\tilde{Y}_{i}=\Lambda (P(Z_{i}))+\varepsilon_{i}$. We thus have 
\begin{align*}
    E[\tilde{Y}_{i}|P(Z_{i})]=\Lambda (P(Z_{i})).
\end{align*}
This equation motivates us to consider a local quadratic estimator to estimate the derivative of $\Lambda(\cdot)$. Let $\hat{\tilde{Y}}_i$ denote $Y_{i}-\hat{\beta}_{0}^\prime X_{i}-(\hat \beta_{1}- \hat \beta_{0})^\prime X_{i}\hat{P}(Z_{i})$.
For each $p$, the local quadratic estimator can be obtained by solving
\begin{align*}\label{min}
\argmin_{(\theta_0,\theta_1,\theta_{2})}\sum_{i=1}^n\left[\big(\hat{\tilde{Y}}_i-\theta_0-\theta_1(\hat{P}(Z_{i})-p)-\theta_2(\hat{P}(Z_{i})-p)^{2}\big)^2K\left(\frac{\hat{P}(Z_{i})-p}{h_n}\right)\right],
\end{align*}
where $K(\cdot)$ is a kernel function of $\mathbb{R}$ and $h_n$ is a sequence of bandwidths. $(\theta_0, \theta_1, \theta_2)$ correspond to the conditional mean, first derivative, and second derivative, respectively. Let $\hat \theta_1 (p)$ be the estimate of $\theta_1$ at $p$.
We use the local quadratic estimator because it provides desirable properties for estimating the first derivative of the nonparametric function. \textcite{carneiro2009estimating} and \textcite{carneiro2011marginalretrurn} also use the local quadratic estimator. We note that the bandwidth $h_n$ must be smaller than the mean-squared-error optimal bandwidth; this undersmoothing centers the asymptotic Gaussian approximation around the true function, and its exact rate is discussed in Section \ref{asymptotic}.

The resulting estimate of $MTE(p,x)$ is 
\begin{equation*}
    \widehat{MTE}(p,x)=(\hat \beta_{1} - \hat \beta_{0})^\prime x+\hat{\theta}_1(p).
\end{equation*}

\section{Uniform confidence band}\label{confidenceband}
This section explains how to compute the critical values for our uniform confidence bands. Our bands have a usual form of ``estimated function $\pm$ critical value  $\times $ standard error.'' Unlike the conventional Gaussian (pointwise) critical value, ours is based on a refined approximation of the distribution of the supremum of a Gaussian process; its theoretical justification is given in the next section.

The critical value for our uniform confidence solves the following equation: 
\begin{align*}
F_{n,1}(c)\geq 1-\alpha,
\end{align*}
where $F_{n,1}(t)=\exp\big(-2e^{-t-t^2/2\ell_n^2}\big)$ and $\ell_n$ is the largest solution to the following equation:
	\begin{align*}
	(b_{0}-a_{0})\cdot h_n^{-1}\sqrt{\lambda}(2\pi)^{-1}\exp(-\ell_n^2/2)=1,
	\end{align*}
	with 
	\begin{align*}
	\lambda:=-\frac{\int g(u)g^{\prime\prime}(u)du}{\int g^2(u)du}\quad \text{and} \quad g(u):=uK(u),
	\end{align*}
 where $(a_0,b_0)$ is defined as the region over which we construct uniform confidence bands for $p$. The explicit expression of $\ell_n$ is:\footnote{We can provide an explicit formula because our nonparametric component, $\hat \theta_1 (p)$, is unidimensional, unlike \textcite{lee2017doubly} who allow multi-dimensional functions.} 
  	\begin{align*}
	\ell_n=\sqrt{2\log\left(\frac{\sqrt{\lambda}(b_0-a_0)}{2\pi h_n}\right)}
	\end{align*}
Note that $h_n$ must be small enough for $l_n$ to be a real number. 
$F_{n,1} (t)$ provides a refined approximation of the supremum of a Gaussian process, which is an asymptotic approximation of the estimated MTE function with undersmoothing.
Solving the equation, we have
\begin{equation*}
    c_{1-\alpha}=(\ell_{n}^{2}-2\log\{\log[(1-\alpha)^{-1/2}]\})^{1/2}.
\end{equation*}

Our two-sided uniform confidence band has the form
\begin{align*}
\widehat{MTE}(p,x) - c_{1-\alpha} \frac{\hat{s}(p)}{\sqrt{nh_n^3}} \leq MTE(p,x) \leq \widehat{MTE}(p,x) + c_{1-\alpha} \frac{\hat{s}(p)}{\sqrt{nh_n^3}},
\end{align*}
where $\hat{s}(p)/\sqrt{nh_n^3}$ is a standard error of $\widehat{MTE(p,x)}$.
Our confidence band has a usual form of confidence interval (i.e., estimate $\pm$ critical value $\times$ standard error). However, the critical value differs from the Gaussian pointwise critical value, which is fixed. Because the critical value increases at the rate of $(\log (n))^{1/2}$, the resulting band is wider and more conservative than a pointwise band while ensuring uniform coverage.

For $\hat{s}(p)$, in the simulations and the empirical example, we use
  	\begin{align*}
	   \hat{s}^2(p)=\frac{\nu_{2}(K)}{nh_n\hat{f}_P^2(p)\kappa_2^2(K)}\sum_{i=1}^{n}\hat{\varepsilon}_i^2K\left(\frac{\hat{P}(Z_{i})-p}{h_n}\right),
		\end{align*}
		where $\hat{f}_{P}(p)$ is the kernel density estimator of $f_P$, $\hat{\varepsilon}_i=\hat{\tilde{Y}}_i-\hat{\theta}_0(\hat{P}(Z_{i})) $, $\nu_{2}(K) = \int u^2K^2(u) du$ and  $\kappa_2(K)=\int u^2K(u)du$.

We would like to emphasize that calculating our confidence band is straightforward. To determine the critical value $c_{1 - \alpha}$, we need to know the kernel function $K$, the bandwidth $h_n$, and the range of the propensity score $ (a_0, b_0) $. Once we establish the critical value, we can compute the confidence band using the standard errors, the bandwidth, and the sample size.

\begin{remark}
\label{rem:gumbel}
The conventional approach based on the Gumbel approximation is based on the approximation given by 
$ F_{\mathrm{Gumbel}}(t)
=
\exp\bigl(-2e^{-t}\bigr)$.
The $(1-\alpha)$ critical value based on this approximation is 
$
c^{\mathrm{Gumbel}}_{1-\alpha}
=
\ell_n
-
\frac{1}{\ell_n}
\log\!\left\{
\log[(1-\alpha)^{-1/2}]
\right\},
$
where $\ell_n$ is the normalizing sequence defined above.
Compared with this critical value, ours includes a logarithmic adjustment term that increases with sample size. The resulting critical value is smaller than the Gumbel critical value, yielding narrower confidence bands.
\end{remark}

\section{Asymptotic Theory}\label{asymptotic}
In this section, we provide an asymptotic justification for our proposed procedure. We begin by introducing the relevant notations and definitions. Next, we outline a set of assumptions that will be used in our asymptotic theory. The distribution of the estimated MTE function is approximated by a Gaussian process. Finally, we justify our critical value based on an approximation of the supremum of this Gaussian process.

First, we introduce some definitions about the region of the uniform confidence band.
\begin{dfn}[Definition of the region of the uniform confidence band]\label{1}
	$\mathcal{X}:= \Pi_{j=1}^d[a_j,b_j]$, where $a_j<b_j$, $j=1,\dots, d$, and $\mathcal{X}$ is a strict subset of the support of $X$. 
	Let $\mathcal{P}:= [a_0,b_0]$, where $0 < a_0 < b_0 < 1$ and $\mathbb{P}\subset (0,1)$ denote an open interval including 	$\mathcal{P}$. 
\end{dfn}

Definition \ref{1} specifies the region over which we construct uniform confidence bands for $X$ and $p$. The requirement that the region for $X$ be a Cartesian product of closed intervals is typically not restrictive, since researchers commonly hold $X$ fixed at a specific value such as the sample mean. Additionally, the region of $p$ must be a closed interval, excluding neighborhoods of the boundary values 0 and 1. It is also important in practice because the estimated MTE function tends to behave poorly near these extreme values.

Second, we introduce some notation. Let $m(p)$ denote $E[\tilde{Y}_{i}|P( Z_{i})=p]$. Also let $\varepsilon_{i}=\tilde{Y}_{i}-m(P( Z_{i}))$.
Let $\sigma^{2}(p)$ denote $E[\varepsilon^{2}_{i}|P(Z_{i})=p]$. 
Let $s_n^2(p)$ denote the  first-order approximated version of the variance of $\widehat{MTE}(p,x)$ that is  not feasible in practice:
\begin{align*}
s_n^2(p)=\frac{1}{h_n^{3}f_P^2(p)\kappa^2_2(K)}E\left[\varepsilon_i^2K^2\left(\frac{P(Z_{i})-p}{h_n}\right)(P(Z_{i})-p)^2\right].
\end{align*}
We note that $\widehat{\beta_{1}-\beta_{0}}$ is estimated at the parametric rate, while $\hat \theta_1$ is nonparametrically estimated and its convergence rate is slower than $\widehat{\beta_{1}-\beta_{0}}$. Therefore, the asymptotic variance of $\widehat{MTE}(p,x)$ depends only on the asymptotic variance of $\hat \theta_1 $. The $s_n^2(p)$ formula corresponds to the variance of $\hat \theta_1$.

We impose the following conditions to establish the asymptotic theory. Many of these assumptions are similar to those used for pointwise inference. The main difference is that, because uniform inference involves a range of propensity-score values, the regularity conditions must be imposed uniformly over the evaluation region rather than at a specific point. Additional assumptions are needed to justify the Gaussian-process approximation used to construct the confidence band and to control the effect of estimating the propensity score in the first stage.

\begin{assumption}
    \label{2}
	The distribution of $P(Z_{i})$ has a bounded Lebesgue density $f_P(\cdot)$ on $\mathbb{P}$ and is bounded away from zero. Furthermore, $f_P(\cdot)$ is twice continuously differentiable in $\mathbb{P}$ and the second derivative of $f_P(\cdot)$ is uniformly bounded on $\mathbb{P}$.
\end{assumption}

\begin{assumption}
	\label{3}
	$\sigma^{2}(p)$ is continuous on $p \in \mathcal{P}$, and   $\sup_{p\in (0,1)}E[|\varepsilon_i|^4|P(Z_{i})=p]<\infty$ holds. Furthermore, $p\mapsto \sigma^{2}(p)f_P(p)$ is Lipschitz continuous. 
\end{assumption}

\begin{assumption}
     \label{4}
	$m(p)$ is three times continuously differentiable on $\mathbb{P}$ with uniformly bounded derivatives.

\end{assumption}

Assumptions \ref{2}-\ref{4} are standard in the literature on nonparametric estimation. To establish uniform convergence results for local quadratic estimators, we follow the approach of \textcite{chernozhukov2014gaussian} and assume the existence of conditional fourth moments of error terms. 
The Lipschitz continuity condition is essential for approximating the supremum of a Gaussian process with the supremum of a stationary Gaussian process.

Assumption \ref{2} is not merely a technical convenience. When the propensity score approaches the boundary values 0 or 1, many treatment effect parameters may become irregularly identified, leading to nonstandard asymptotic behavior. See \textcite{khan2010irregular} and \textcite{heiler2021valid} for discussions of irregular identification.

\begin{assumption}
    \label{5}
	Let $K(\cdot)$ be a kernel function on $\mathbb{R}$ compactly supported, bounded, symmetric around zero, and six times differentiable.

\end{assumption}

We assume that the kernel function is six times differentiable, as stated in Assumption \ref{5}. This assumption is important for establishing the asymptotic theory for the maximum of a Gaussian homogeneous field, particularly because we follow Theorem 4.3 from \textcite{piterbarg1996asymptotic}. This assumption is also essential in analytically deriving the limit distribution of the supremum of a stationary Gaussian process. A practical implication is that it excludes specific kernels, such as the uniform kernel. 

\begin{assumption}
    	\label{7}
    \begin{enumerate}
      \item \label{7-1} $\max_{1\leq i\leq n}|\hat{P}(Z_i)-P(Z_i)|=o_p(h_n^{3})$ holds.
      \item \label{7-2} There exists a VC-type class $\mathcal{M}$ with the envelope 2 such that $\lim_{n\rightarrow\infty}\Pr(\hat{P}\in\mathcal{M})=1$ holds. 
  \end{enumerate}

\end{assumption}

\begin{assumption}
    	
  \label{8}
$h_n=Cn^{-\eta}$, where $C$ and $\eta$ are positive constant such that  $\frac{1}{7}< \eta <\frac{1}{6}$ and   $h_n<\sqrt{\lambda}(b_0-a_0)/2\pi$ hold with $	\lambda=-\int g(u)g^{\prime\prime}(u)du/\int g^2(u)du$ and $g(u)=uK(u)$.

\end{assumption}

Assumption \ref{7} outlines the conditions that the predicted propensity score, $\hat{P}(Z)$, must satisfy. In practice, $\hat{P}(Z)$ can be estimated using parametric models like probit or logit. In these cases, the estimation error is typically of the order $O_p(n^{-1/2})$. This convergence rate can meet the requirements outlined in Assumption \ref{7}.\ref{7-1}, provided that the sequence of $h_n$ satisfies the conditions in Assumption \ref{8}.
Semiparametric estimators are also options. However, appropriate higher-order kernels must be utilized to achieve a sufficiently fast convergence rate. For example, for a single-index model, we require at least a 4th-order kernel. 
The definition of the VC-type class related to Assumption \ref{7}.\ref{7-2} is provided in Appendix \ref{Appendix_VC}. This condition is mild and is satisfied by many commonly used methods for estimating propensity scores.

The bandwidth conditions in Assumption \ref{8} ensure that the estimated MTE function is asymptotically unbiased: the required rate is faster than the mean-squared-error optimal rate, so we undersmooth.  
Furthermore, this condition, along with Assumption \ref{7} and the Lipschitz continuity of the kernel, allows us to asymptotically disregard the bias that arises from propensity score estimation in the first step.

\begin{assumption}
    \label{as-sn}
   \begin{enumerate}
   \item \label{9}
	$\inf_{n\geq 1} \inf_{p\in \mathcal{P}}s_n(p)>0$ and $s_n(p)$ is continuous for each $n\geq 1$. 

    \item \label{10}
	An estimator of $s_n^2(p)$, $\hat{s}^2(p)$, exists such that 
	\begin{align*}
	\sup_{p\in \mathcal{P}}\big|\hat{s}^2(p)- s_n^2(p) \big|=O_p(n^{-c}).
	\end{align*}
\end{enumerate}

\end{assumption}

Assumption \ref{as-sn} addresses the asymptotic variances of the MTE function. The conditions of strict positivity and continuity in Assumption  \ref{as-sn}.\ref{9} are standard requirements. The rate condition in Assumption   \ref{as-sn}.\ref{10} is necessary to prevent the inverse of the standard error from diverging to infinity. This rate can be achieved using the estimator provided at the end of the previous section.\footnote{Alternatively, the formula for $s_n^2(p)$ suggests another estimator:
		\begin{align*}
\tilde{s}^2(p) = \frac{1}{nh_n^3\hat{f}_P^2(p)\kappa_2^2(K)} \sum_{i=1}^{n} \hat{\varepsilon}_i^2 K^2\left(\frac{\hat{P}(Z_{i})-p}{h_n}\right)(\hat{P}(Z_{i})-p)^2.
		\end{align*}
Simulations indicate that the estimator presented in the previous section behaves more stably than this alternative estimator.}

\begin{assumption}
 \label{covariate}
   \begin{enumerate}
  \item \label{11-1}There exists an estimator $\widehat{\beta}_{0}$ and $\widehat{\beta_{1}-\beta_{0}}$ such that
 \begin{equation*}
     \widehat{\beta}_{0}-\beta_{0}=O_{p}(n^{-1/2}),\quad \widehat{\beta_{1}-\beta_{0}}-(\beta_{1}-\beta_{0})=O_{p}(n^{-1/2}).
 \end{equation*}
      \item \label{11-2} For each $\ell\in\{1,\cdots,d\}$, $E[X_{i,\ell}|P(Z_{i})=p]$ is continuously differentiable in $(0,1)$, and the derivative of $E[X_{i,\ell}|P(Z_{i})=p]$ is uniformly bounded on $(0,1)$.
      \item \label{11-3} For each $\ell\in\{1,\cdots,d\}$, $\sup_{p\in (0,1)}E[|X_{i,\ell}|^2|P(Z_{i})=p]<\infty$ holds.
  \end{enumerate}

\end{assumption}

Assumption \ref{covariate} pertains to the impact of preliminary estimation. 
The semiparametric coefficient estimator presented in \textcite{carneiro2009estimating} satisfies Assumption  \ref{covariate}.\ref{11-1}. For more detailed assumptions required to achieve root-$n$ consistent estimators, refer to \textcite{mammen2016semiparametric}. 
Assumptions  \ref{covariate}.\ref{11-2} and \ref{covariate}.\ref{11-3} indicate that the covariates must exhibit sufficient variation, even after conditioning on the propensity score.

\begin{remark}
Note that we can also consider a more flexible estimation of the conditional expectation of potential outcomes. In particular, the main arguments would continue to hold as long as two conditions are satisfied: (i) the estimators of the conditional expectations, $\mu_d(X)$, converge at a polynomial rate faster than $\sqrt{n h_n^{3}}$, and (ii) the generated regressor $\hat{\tilde{Y}}_i$ converges to $\tilde{Y}_i$ sufficiently fast so that the first-stage estimation error remains asymptotically negligible for the construction of the confidence bands.
\end{remark}

The following lemma establishes a linear expansion of the semiparametric estimator.
\begin{lemma} \label{lemma1}
Let Assumptions \ref{ass::semipara}-\ref{covariate} hold. Then,
\begin{align*}
\sup_{(p,x)\in \mathcal{P}\times \mathcal{X}}\sqrt{nh_n^{3}}  &\left| \frac{\widehat{MTE}(p,x)-MTE(p,x)}{\hat{s}(p)}-\right.\\
& \left.\frac{1}{nh_n^{3}f_P(p)\kappa_2(K)s_n(p)}\sum_{i=1}^n\varepsilon_iK\left(\frac{P(Z_i)-p}{h_n}\right)(P(Z_i)-p) \right|\\
=O_P(n^{-c})&,
\end{align*}
for some positive constant $c>0$.
\end{lemma} 

This lemma demonstrates that preliminary estimation of propensity scores and coefficients in the partial linear model does not asymptotically affect the estimation error of the MTE function. The proof addresses several potential sources of estimation error.
First, we consider the difference in the dependent variables in the local quadratic estimation.
The difference between the true dependent variable $\tilde{Y}$ and the feasible version $\hat{\tilde{Y}}_{i}$  
is asymptotically negligible under Assumptions \ref{7}.\ref{7-1} and \ref{covariate}.\ref{11-1}, noting that Assumptions \ref{7}.\ref{7-1} and \ref{8} together ensure sufficiently fast convergence of the estimated propensity score.
Second, we replace the asymptotic variance $s_n$ with its estimator $\hat{s}$, which is handled by Assumption \ref{as-sn}.
Third, and most importantly, the nonparametric estimator relies on the generated regressor $\hat{P}(Z_i)$ rather than on the true propensity score. 
We therefore treat it as an infinite-dimensional nuisance component depending on $Z$ rather than as a conventional generated covariate. Assumption \ref{7} and \ref{8} ensure that the estimation error in 
$\hat{P}(Z_i)$  is sufficiently well behaved so that the MTE estimator can be analyzed as 
an empirical process indexed by the propensity score.

Define
\begin{align*}
T_n(p)=\frac{1}{nh_n^{3}}\sum_{i=1}^n\varepsilon_i K\left(\frac{P(Z_i)-p}{h_n}\right)(P(Z_i)-p)
\end{align*}
and 
\begin{align*}
c_n(p)=\left\{\frac{1}{h_n^{3}} E\Big[ \varepsilon_i^2K^2\left(\frac{P(Z_i)-p}{h_n}\right)(P(Z_i)-p)^2\Big]  \right\}^{-1/2}.
\end{align*}
By Lemma \ref{lemma1}, we have
\begin{align*}
\sup_{(p,x)\in \mathcal{P}\times \mathcal{X}}\sqrt{nh_n^{3}}\left| \frac{\widehat{MTE}(p,x)-MTE(p,x)}{\hat{s}(p)}-c_n(p)T_n(p) \right|=O_p(n^{-c}).
\end{align*}
We approximate the supremum of the empirical process $c_n(p)\sqrt{nh_n^{3}}(T_n(p)-E[T_n(p)])$ by the supremum of a Gaussian process using the result of \textcite{chernozhukov2014gaussian}. Define
\begin{align*}
W_n=\sup_{p\in \mathcal{P}}c_n(p)\sqrt{nh_n^{3}}\left[T_n(p)-E[T_n(p)]) \right].
\end{align*}

\begin{lemma}\label{lemma2}
	Let Assumptions \ref{ass::semipara}-\ref{covariate} hold. Then, for every $n\geq 1$, there is a tight Gaussian random variable $\tilde{B}_{n,1}$ in $\ell^{\infty} (\mathcal{P})$ with mean zero and 
	covariance function
	\begin{align*}
	&E[ \tilde{B}_{n,1}(p) \tilde{B}_{n,1}(\check{p})] \notag \\
 =&
	h_n^{-3}c_n(p)c_n(\check{p})E \left[\varepsilon_{i}^2K\left(\frac{P(Z_{i})-p}{h_n}\right) K\left(\frac{P(Z_{i})-\check{p}}{h_n}\right)(P(Z_{i})-p)^2(P(Z_{i})-\check{p})^2 \right],
	\end{align*}
	and there is a sequence $\tilde{W}_{n,1}$ of random variables such that $\tilde{W}_{n,1}$ is equal to $\sup_{p\in \mathcal{P}} \tilde{B}_{n,1}(p)$ in distribution and as $n \to \infty$:
	\begin{align*}
	|W_n-\tilde{W}_{n,1}|=O_P\{(nh_n^{3})^{-1/6}\log n +(nh_n^{3})^{-1/4}\log^{5/4}n+(nh_n^{6})^{-1/4}\log^{3/2}n       \}.
	\end{align*}
\end{lemma}
Next, we show that $\tilde{W}_{n,1}$ can be further approximated by the supremum of a stationary Gaussian process.
\begin{lemma}\label{lemma3}
	Let Assumptions \ref{ass::semipara}-\ref{covariate} hold.  For sufficiently large $n$, there is a tight Gaussian process $\tilde{B}_{n,2}$ in $\ell^{\infty} (\mathcal{P})$ with mean zero and 
	covariance function
		\begin{align*}
	E[ \tilde{B}_{n,2}(p) \tilde{B}_{n,2}(\check{p})]=\mathbf{\rho}(p-\check{p})
	\end{align*}
	for $p, \check{p} \in h_{n}^{-1}\mathcal{P}$, where we define
	\begin{align*}
	\rho(p)=\frac{\int uK(u)(u-p)K(u-p)du}{\int u^{2}K^2(u)du}.
	\end{align*}
	And there is a sequence $\tilde{W}_{n,2}$ of random variables such that $\tilde{W}_{n,2}$ is equal to $\sup_{p\in \mathcal{P}} \tilde{B}_{n,2}(h_{n}^{-1}p)$ in distribution and as $n \to \infty$:
	\begin{align*}
	|\tilde{W}_{n,1}-\tilde{W}_{n,2}|=O_P\left\{h_n^{1/2}\sqrt{\log h_n^{-(3/2
 )}} \right\}.
	\end{align*}
\end{lemma}

 From Lemma \ref{lemma1} to Lemma \ref{lemma3} with symmetry, there exist some positive constant $c$ such that
\begin{equation*}
     \sup_{(p,x)\in \mathcal{P}\times\mathcal{X}}\sqrt{nh_n^{3}}   \left|\frac{\widehat{MTE}(p,x)-MTE(p,x)}{\hat{s}(p)}\right|-\sup_{p\in\mathcal{P}}|\tilde{B}_{n,2}(h^{-1}_{n}p)|=o_p(n^{-c})
\end{equation*}
holds. The supremum of the studentized estimated MTE function can be approximated by the supremum of a stationary Gaussian process. Moreover, the approximation error is of a polynomial order.

Combining this with known results on the distribution of the supremum of a stationary Gaussian process, we obtain the following main theorem:
\begin{theorem}\label{theorem}
	Suppose Assumptions \ref{ass::semipara}-\ref{covariate} hold.  Then, uniformly in $t$ over any finite interval,  the following result holds:
	\begin{align}\nonumber
	&\Pr\Big(\ell_n\Big[ \sup_{(p,x)\in \mathcal{P}\times \mathcal{X}}\sqrt{nh_n^3}\left|\frac{\widehat{MTE}(p,x)-MTE(p,x)}{\hat{s}(p)} \right|-\ell_n   \Big]<t \Big)
	\label{thmeq} \\
	=& \exp\big(-2e^{-t-t^2/2\ell_n^2} \big)+o(1),
	\end{align}
	as $n\to \infty$, where $\ell_n$ is the largest solution to the following equation:
	\begin{align*}
	(b_{0}-a_{0})\cdot h_n^{-1}\sqrt{\lambda}(2\pi)^{-1}\exp(-\ell_n^2/2)=1,
	\end{align*}
	where 
	\begin{align*}
	\lambda:=-\frac{\int g(u)g^{\prime\prime}(u)du}{\int g^2(u)du}\quad \text{and} \quad g(u):=uK(u).
	\end{align*}
\end{theorem}

Theorem \ref{theorem} provides theoretical justification for our confidence band. For any $(p,x)\in\mathcal{P}\times \mathcal{X}$, we have
\begin{align*}
&\Pr\left(\sqrt{nh_n^3}\left|\frac{\widehat{MTE}(p,x)-MTE(p,x)}{\hat{s}(p)} \right|> c_{1-\alpha}\right) \\
    \leq &\Pr\left(\sup_{(p,x)\in \mathcal{P}\times \mathcal{X}}\sqrt{nh_n^3}\left|\frac{\widehat{MTE}(p,x)-MTE(p,x)}{\hat{s}(p)} \right|> c_{1-\alpha}\right) \\
    \leq & \alpha + o(1).
\end{align*}

\begin{remark}
The result in Theorem \ref{theorem} is derived from a refined approximation of the distribution of the supremum of a stationary Gaussian process $\tilde{W}_{n,2}$, as presented by \textcite{piterbarg1996asymptotic}. While the commonly used Gumbel approximation $\Pr\Big(l_n \left(|\tilde{W}_{n,2}| - l_n \right) < t \Big) = \exp\big(-2e^{-t}\big) + o(1)$ has a logarithmic rate of convergence, \textcite{piterbarg1996asymptotic} demonstrates a more accurate approximation: $\Pr\Big(l_n \left(|\tilde{W}_{n,2}| - l_n \right) < t \Big) = \exp\big(-2e^{-t - t^2/2 \ell_n^2}\big) + o(n^{-c})$, where the approximation error is of polynomial order. The correction term $-t^2/2 \ell_n^2$ reduces the critical value relative to the Gumbel distribution, resulting in uniform confidence bands that achieve better coverage while being less conservative than Gumbel-based methods.
\end{remark}

\section{Monte Carlo Experiments}\label{Monte_Carlo}

In this section, we present the results of Monte Carlo experiments. To reduce researcher discretion and facilitate comparisons with the existing literature, we base our simulation design on the framework used in the illustration of the MTE in Figures 1A and 1B of \textcite{heckman2005structural}.\footnote{
    We note that alternative Monte Carlo designs have also been considered in the recent MTE literature, including \textcite{heiler2022efficient} and \textcite{sloczynski2025abadie}. We choose the design of \textcite{heckman2005structural} because it provides a simple and transparent benchmark that directly illustrates the key features of the MTE framework while satisfying the support conditions required by our theory.}

\subsection{Data generating process}

We consider settings with one exogenous (included) covariate and one excluded instrument. 
 We generate $(X_{1},Z_{1})$ as follows: $ (X_{1},Z_{1})^{\prime}:=N(\mathbf{0}_{2},I_{2})$ 
where we denote $\mathbf{0}_{2}$ and $I_{2}$ as a $2\times1$ zero vector and a $2\times2$ identity matrix, respectively. 
$X=X_1$ is a covariate that affects potential outcomes, and $Z=(X_1, Z_1)$ is the vector of instruments that affect the propensity score. 
The potential outcomes are generated by
\begin{align*}
Y_1=0.67+0.2+0.5X_{1}+U_{0} , \text{ and } Y_0=0.67+0.2X_{1}+U_{1}.
\end{align*}
The treatment $D$ is generated by $ D=\mathbbm{1}\{Z_{1}>U_D\}$.
The vector of unobservable components $(U_{1},U_{0},U_D)^{\prime}$ is generated from $N(\mathbf{0}_{3},\Sigma)$, where $\Sigma$ depends on the design. We use the following three types of covariance matrices.
\begin{align*}
 \Sigma_{1} =& \begin{pmatrix}
                \frac{9}{400} & -\frac{9}{400}& \frac{3}{20} \\
                -\frac{9}{400} & \frac{9}{400} & -\frac{3}{20} \\
                \frac{3}{20}   & -\frac{3}{20}  & 1 
            \end{pmatrix}, 
            \quad
            \Sigma_{2} =& \begin{pmatrix}
               \frac{9}{400}  & -\frac{9}{400}  & -\frac{3}{20} \\
                -\frac{9}{400}  & \frac{9}{400} & \frac{3}{20}\\
                -\frac{3}{20}  & \frac{3}{20} & 1 
            \end{pmatrix},
            \quad  
            \Sigma_{3} =& \begin{pmatrix}
                \frac{9}{400}  & -\frac{9}{400}  & 0 \\
                -\frac{9}{400}  & \frac{9}{400}  & 0 \\
                0  & 0 & 1 
            \end{pmatrix}.
\end{align*}
In the designs $\Sigma_1$ and $\Sigma_2$, the unobservable that affects selection into the treatment ($U_D$) is correlated with $U_1$ and $U_0$. Thus, there exists endogenous selection. In contrast, in design $\Sigma_{3}$, the treatment status is independent of the potential outcomes.

We now derive the MTE function in this setting.
By the definition of normal distribution, the conditional distribution of $(U_1, U_0)$ given $V=\Phi (U_D)$, where  $\Phi(\cdot)$ is the cumulative distribution function of the standard normal distribution, can be analytically calculated. Thus, the MTE function has a closed-form representation:
\begin{align*}
    MTE(p, x_1)&=0.2+0.3 x_1+ (Cov(U_{1},U_D)-Cov(U_{0},U_D))\times \Phi^{-1}(p) .
\end{align*}
We focus on $MTE(p,E(X_1)) = MTE(p,0) = 0.2+(Cov(U_{1},U_D)-Cov(U_{0},U_D))\times \Phi^{-1}(p)  $.
 The shape of the MTE function varies across the types of $\Sigma$. $\Sigma_1$ yields an increasing MTE function, while the MTE function under $\Sigma_2$ is decreasing. In the case of $\Sigma_3$, the MTE function is constant. 

The number of Monte Carlo and bootstrap replications is 1000 and 500, respectively. The sample size is $n=3000$.

\subsection{Estimation Procedure}

We estimate the MTE function using the following steps.
\begin{enumerate}
    \item We estimate the coefficients in $\mu_D$ using the maximum likelihood estimator (probit) and construct the estimated propensity score, $P(Z)$. The MTE function can only be estimated at the intersection of the supports of the propensity scores for individuals who received the treatment and those who did not receive it. Therefore, we exclude all observations that fall outside this intersection when we estimate the MTE function. 

    \item The partially linear estimator is used to estimate regression coefficients. The estimation method is described in Section \ref{estimation}. The bandwidths for conditional expectations are set as the one minimizing the average squared error through the cross-validation method. 
Let $\hat{\beta}_{d}$ denote the coefficient estimator on $X_1$ in the regression for $Y_d$.
    \item 
    We estimate the $\Lambda (p)=\int^{p}_{0}E[U_{1}-U_{0}|V=v]dv$ and its first derivative through the local quadratic estimator with a Gaussian kernel. 
    The bandwidth is determined by adjusting the rule-of-thumb bandwidth $ h_{n} $ as described by \textcite{fan1996local}. First, we compute the standard rule-of-thumb bandwidth, and then we modify its order from $ n^{-1/7} $ to $ n^{-2/13} $, so that the bandwidth meets the requirements of Assumption \ref{8}.
    \item $MTE(p, E(X_1), E(X_2))$ is estimated as $(\hat \beta_{1} - \hat \beta_{0}) \bar X_1  + \hat \Lambda'(p)$, where $\hat \Lambda'(p)$ is the local quadratic estimator of the first derivative of $\Lambda (p)$. 
\end{enumerate}

\subsection{Methods in comparison}

We compare four confidence bands --- pointwise, Gumbel, bootstrap, and our method --- on the interval $[0.15,0.85]$. The pointwise method uses the usual critical values of 1.64 (90\%) and 1.96 (95\%). 
The Gumbel method is explained in Remark \ref{rem:gumbel}. 
For the bootstrap method, we employ a Gaussian multiplier bootstrap that perturbs the local polynomial estimator through independent random weights and approximates the distribution of the supremum of the studentized estimation error process.\footnote{We do not examine the theoretical properties of the bootstrap method. We expect its theoretical validity to follow from an extension of \textcite{chernozhukov2014gaussian} combined with our Lemma \ref{lemma1}.}

\subsection{Results}
The results of the Monte Carlo simulations are summarized in Table \ref{monte_result_1}. It specifically presents the empirical coverage frequencies and mean values of the critical values observed in the Monte Carlo experiments. 

 \begin{table}[thb]
        \centering
    \caption{Results of Monte Carlo simulations}
           \label{monte_result_1}
        \begin{tabular}{ccccc}
        \hline
             & CP (90\%) & Mean Crit. Val. (90\%) & CP (95\%) &  Mean Crit. Val. (95\%) \\
             \hline 
    \multicolumn{5}{c}{Design: $\Sigma = \Sigma_1$. The mean bandwidth is 0.054} \\
            Pointwise & 0.212 & 1.65 & 0.427  & 1.96 \\
            Gumbel &  0.987  & 3.54 & 0.998 & 4.08 \\
            Bootstrap & 0.911  & 2.89 & 0.963 & 3.15 \\ 
            Ours  & 0.889 & 2.78 & 0.951 & 3.03 \\ 
            \hline
    \multicolumn{5}{c}{Design: $\Sigma = \Sigma_2$. The mean bandwidth is 0.055} \\
            Pointwise & 0.223  & 1.65  & 0.430 & 1.96 \\
            Gumbel & 0.985  & 3.54  & 0.999 & 4.08 \\
            Bootstrap & 0.905 & 2.89 & 0.957 & 3.16\\ 
            Ours  & 0.889 & 2.78  & 0.942 & 3.03 \\ 
            \hline
    \multicolumn{5}{c}{Design: $\Sigma = \Sigma_3$.  The mean bandwidth is 0.056} \\
            Pointwise &  0.206 & 1.65  & 0.434 & 1.96 \\
            Gumbel & 0.990 & 3.55  & 0.999 & 4.09 \\
           Bootstrap & 0.885 & 2.88  & 0.957 & 3.15\\ 
            Ours  &0.855 &2.77 & 0.934 & 3.02 \\ 
            \hline \\
        \end{tabular}

\begin{minipage}{\linewidth}

    Note: This table shows the empirical coverage and critical values for pointwise, Gumbel, bootstrap, and our method. The reported critical values and mean bandwidths are rounded to two and three decimal places, respectively.
\end{minipage}

\end{table}

The confidence band constructed by our method contains the true MTE function with a nominal coverage probability, and the maximum error is 2.5\% in all cases except cases with $\Sigma =\Sigma_3$ (constant MTE).
The empirical coverage probabilities of the pointwise confidence band are much lower than the nominal levels. This is not surprising, because the pointwise band is designed to cover the function at a single point rather than over a range, making it unsuitable for assessing the uncertainty of the estimated function.
The confidence bands based on the Gumbel distribution attain the nominal coverage probabilities. However, this method is conservative and may not be informative.

In cases with $\Sigma =\Sigma_3$, where the slope of the true MTE function is zero, all methods exhibit slightly lower precision relative to cases with $\Sigma= \Sigma_1$ and $\Sigma_2$. One possible explanation is that the estimation procedure is based on a Taylor approximation. When the true first derivative is close to zero, the approximation becomes more sensitive to sampling fluctuations, leading to less stable finite-sample performance.

The 95\% critical value given by our method is approximately 3.03. As expected, it is larger than the conventional pointwise approach (1.96). However, this value remains below the critical value from the Gumbel distribution, which is around 4.08. Interestingly, the critical values do not vary significantly across different designs.

We next compare our method with the bootstrap procedure. The empirical coverage probabilities and critical values of the two methods are similar at both the 90\% and 95\% confidence levels across all designs. For example, the 95\% critical values from the bootstrap procedure are around 3.15, while ours are around 3.03. Although the bootstrap method attains slightly better empirical coverage probabilities in some cases, the differences between the two procedures are small. Overall, these results indicate that our method achieves statistical performance comparable to the bootstrap approach. Importantly, our method is roughly ten times faster than the bootstrap method (Table \ref{monte_result_time} in Appendix \ref{sec_monte_time}).

\section{Empirical Application}\label{empirical_application}

We illustrate the use of our proposed method through an empirical application that revisits \textcite{dinkelman2011effects}, which studies the rural electrification program. Using the MTE framework, we analyze how the effect of electrification on the female employment change varies across communities with different propensities to receive electricity access. We then construct our uniform confidence band for the MTE function and compare it with alternative inference methods.

We use the data from \textcite{dinkelman2011effects}, which examines the labor market effects of rural electrification in South Africa. The outcome variable $Y_i$ is the change in the female employment share between 1996 and 2001. The treatment variable $D_i$ is an indicator for whether a community was connected to the electricity grid during the study period. Following \textcite{dinkelman2011effects}, the vector $X_i$ includes a rich set of community characteristics.\footnote{Specifically, it includes household density, poverty rates, demographic composition, educational attainment, and geographic measures of infrastructure access; they are measured prior to electrification using information from the 1996 Census. In addition, we include changes in access to water and sanitation services.}

For the potential outcomes, we consider the partially linear specification:
\[
Y_{ji}=X_i'\beta_j+U_{ji},
\qquad j\in{0,1},
\]
where $X_i$ denotes the vector of baseline covariates and $U_{ji}$ is an unobserved error term. The MTE framework allows the treatment effect to vary across communities by assuming that $(U_{1i}, U_{0i})$ depend on the latent selection variable.

To identify the propensity score, we use the geographic instrument proposed by \textcite{dinkelman2011effects}. Specifically, the instrument is the average land gradient of a community, which affects the cost of extending the electricity network but is plausibly unrelated to future labor market outcomes except through electrification. We estimate the propensity score using a probit model of the form:
\[
D_i=\mathbbm{1}\{X_i'\gamma+\delta Z_i\ge V_i\},
\]
where $Z_i$ denotes the average land gradient and $V_i\sim N(0,1)$. Observations with missing values are excluded from the analysis.

\begin{figure}[thb]
  \centering
        \includegraphics[scale=0.45]{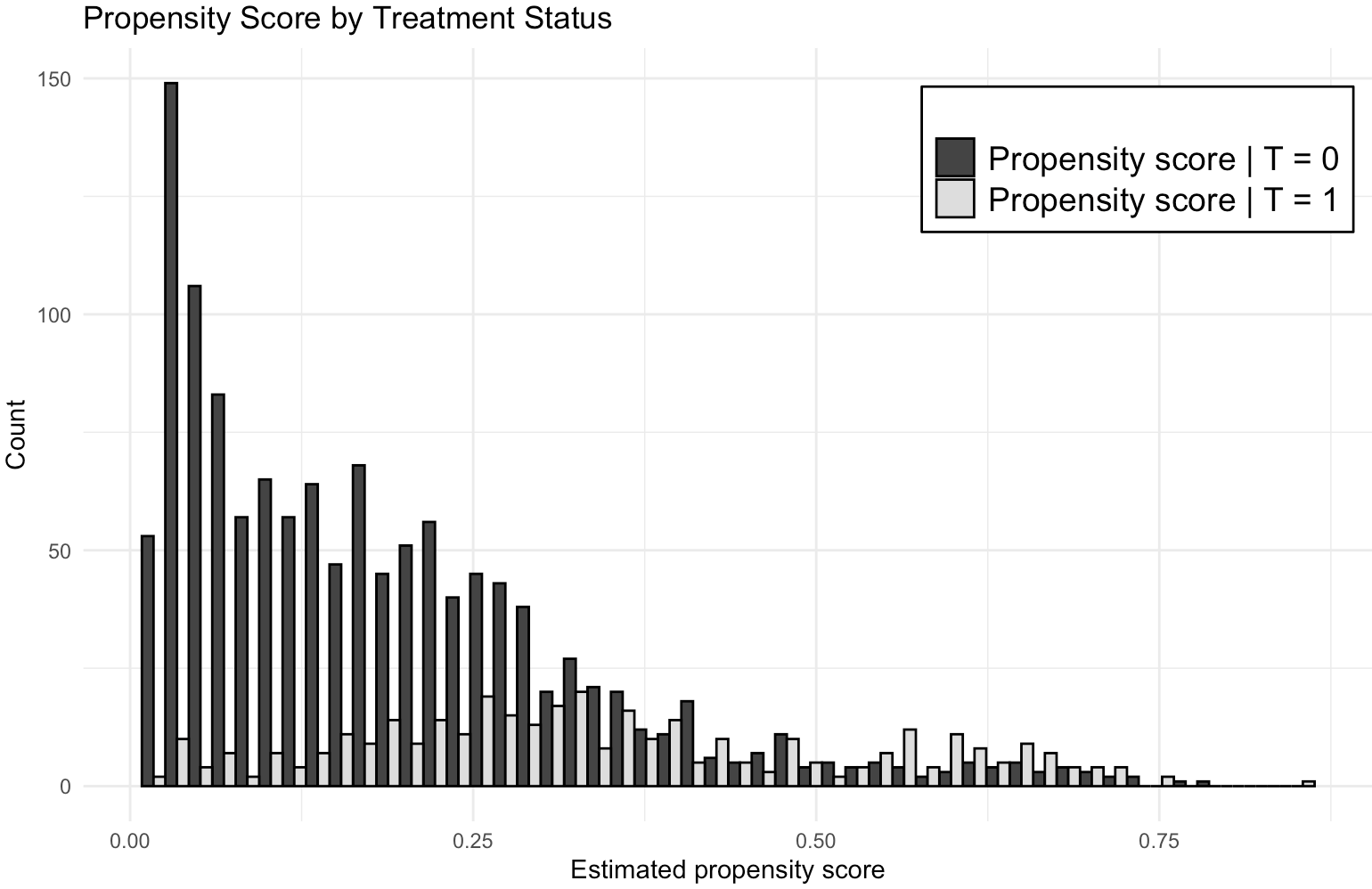}
        \caption{This figure presents histograms of the estimated propensity scores for the treated and untreated communities. The propensity scores are estimated using the probit model. 
        }
        \label{support of propensity score}
\end{figure}  

Figure \ref{support of propensity score} presents histograms of the estimated propensity scores for the treated and untreated communities. The average land gradient is continuously distributed across communities, and consequently, the estimated propensity scores exhibit substantial variation and span a broad range. 
The figure also shows considerable overlap between the treated and untreated groups across most of the propensity score distribution; Appendix \ref{sec:add_emp} confirms that substantial overlap remains even after conditioning on covariates.

We estimate the MTE function using the local quadratic estimator. The bandwidth $h_n$ is selected using the rule-of-thumb procedure based on \textcite{fan1996local}, and we employ the Gaussian kernel throughout the analysis. To improve estimation accuracy near the boundaries, the local polynomial estimator is computed over the support region $[0.02,0.858]$, which corresponds to the intersection of the estimated propensity-score supports for the treated and untreated communities. This region contains 1,647 observations, and the estimated propensity score takes a distinct value for each observation.

For statistical inference, we focus on the interval $[0.05,0.50]$. Although the MTE is estimable over a wider range, the precision of nonparametric estimation deteriorates near the boundaries of the propensity score support due to limited local information. Restricting attention to the interior region allows us to conduct inference where overlap between treated and untreated communities remains substantial, and the estimator is most reliable.

\begin{figure}[thb]
  \centering
        \includegraphics[scale=0.4]{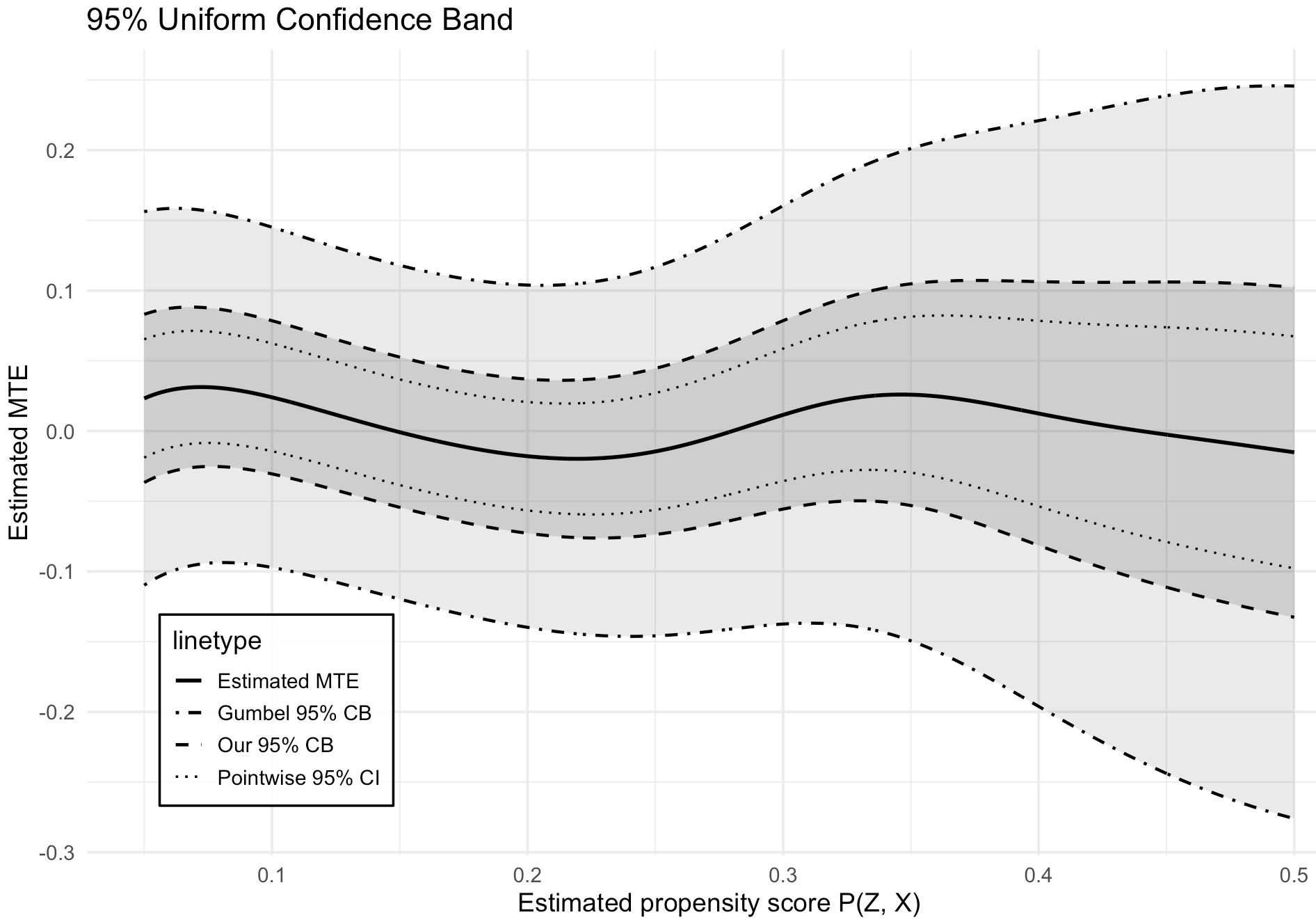}
        \caption{This figure shows the MTE with three types of confidence bands. The MTE estimates the causal effects of the rural electrification program. The calculated bandwidth is $0.070$.}
        \label{applied_result}
\end{figure}   

As shown in Figure \ref{applied_result}, the estimated MTE exhibits a mild downward trend over much of the propensity score range. However, the uniform confidence band indicates considerable uncertainty regarding the shape of the treatment effect function: the data do not provide sufficient evidence to establish monotonicity of the MTE, and several plausible treatment-effect patterns remain consistent with the data. Compared with the pointwise confidence interval, our uniform confidence band is necessarily wider, reflecting the additional uncertainty associated with simultaneous inference over the entire propensity score range. Nevertheless, it remains substantially narrower and more informative than the band based on the Gumbel approximation, which is sufficiently wide to accommodate a broad class of treatment effect functions while still being valid for uniform inference. Overall, our procedure provides a useful compromise between statistical validity and informativeness in empirical applications.

\section{Conclusion}\label{conclusion}
 This paper proposes a method to construct a uniform confidence band for the marginal treatment effect function. To estimate the MTE, we assume that the potential outcomes are linear in the covariates and propose a semiparametric estimator. Our uniform confidence band relies on the approximation of the supremum of a Gaussian process combined with the Gaussian approximation of empirical processes. Empirical researchers are recommended to add our easy-to-implement uniform confidence bands when reporting MTE function estimation results.

 Several avenues for future research emerge from this study. While our paper focuses on local quadratic estimation of the MTE function, alternative approaches such as sieve or series estimation, as employed by \textcite{hoshino2022estimating}, warrant exploration. Developing methods for uniform confidence bands for the MTE function estimated via sieve methods presents an intriguing challenge, as it requires mathematical techniques distinct from those employed here. Relatedly, it is also a standard practice to use parametric polynomial models in MTE function estimation, as demonstrated by \textcite{brinch2017beyond} and \textcite{sasaki2024welfare}. Extending methodologies to accommodate such parametric approaches could benefit practitioners. We may also try to formally establish the validity of bootstrap procedures, extending the analysis of \textcite{chernozhukov2014anti}. 
We may also be able to improve the rate at which we approximate errors. \textcite{CattaneoYu2025} present more accurate methods for establishing the rate of Gaussian approximations. However, extending these findings to construct confidence intervals and determine the rate of coverage errors remains a significant challenge. We leave these extensions to future research. 

\paragraph{Conflict of Interest statement:} The authors report there are no competing interests to declare.

\paragraph{Use of Generative AI Statement:} During the preparation of this work, the authors used Grammarly to improve language. After utilizing this tool, the authors reviewed and edited the content as needed and take full responsibility for the final content.

\paragraph{Data Availability Statement:} The data used in the empirical application are available on the replication package of \textcite{dinkelman2011effects} (\url{https://www.openicpsr.org/openicpsr/project/112474/version/V1/view}).

\printbibliography[heading=subbibliography]

\end{refsection}

\newpage

\renewcommand{\thepage}{S.\arabic{page}} 
\setcounter{page}{1}

\section*{Online Supplemental Materials for ``A Uniform Confidence Band for the Marginal Treatment Effect Function'' by Tsuda, Jin, and Okui}

\begin{refsection}

\appendix

\renewcommand{\theequation}{A.\arabic{equation}} 
\def\thesection{\Alph{section}}
\setcounter{equation}{0}

\section{Definition of VC-type classes}
\label{Appendix_VC}

This appendix defines VC-type classes used in Assumption \ref{7-2}. We first introduce various notations and then give the definition.

For any function $f$, let $\|f\|_{Q,2}$ denote $(\int |f|^2d Q)^{1/2}$. We use the notation $e_{Q}$ as the $L^{2}(Q)$-semimetric, i.e. $e_{Q}(f,g):=\|f-g\|_{Q,2}$. Let $(T,d)$ denote the semimetric space. For $\varepsilon>0$, we define an $\varepsilon$-net of $T$, $T_{\varepsilon}$, as a subset of $T$ such that for every $t\in T$ there exists a point 
$t_{\varepsilon}\in T_{\varepsilon}$ with $d(t, t_{\varepsilon}) < \varepsilon$. The $\varepsilon$-covering number $N(T,d,\varepsilon)$ is the infimum of the cardinality of $T_{\varepsilon}$, namely, $N(T,d,\varepsilon):=\inf\{\text{Card}(T_{\varepsilon}) : T_{\varepsilon}  \text{ is an }  \varepsilon\text{-net of T}\}$.  Unless otherwise stated, $c>0$ refers to universal constants whose values may change from place to place. 

We now define VC-type classes. 
\begin{dfn}[Definition of VC-type classes]\label{VC-type-def}
		Let $\mathcal{F}$ be a class of measurable functions on a measurable space $(S,\mathcal{S})$, to which a measurable envelope $F$ is attached. We say that $\mathcal{F}$ is VC-type with envelope $F$ if there are constants $A$, $v>0$ such that $\sup_QN(\mathcal{F},e_{Q}, \epsilon\|F\|_{Q,2} )\leq (A/\epsilon)^v$ for all $0<\epsilon\leq 1$, where the supremum is taken over all finitely discrete probability measure on $(S,\mathcal{S})$.
\end{dfn}

\section{Proof of Lemma \ref{lemma1}}\label{proof_unif_conv}

\begin{proof}[Proof of Lemma \ref{lemma1}]
 Let $\Gamma(\hat{V})$, $\Omega(\hat{V})$, $\widehat{\tilde{\mathbf{Y}}}$, $\tilde{\mathbf{Y}}$ and $e_{2}$ denote
\begin{align*}
\Gamma(\hat{V})=&\begin{bmatrix}
1 & (\hat{P}(Z_{1})-p) & (\hat{P}(Z_{1})-p)^2 \\
\vdots & \vdots & \vdots \\
1 & (\hat{P}(Z_{n})-p) & (\hat{P}(Z_{n})-p)^2
\end{bmatrix},& \\
\Omega(\hat{V})=&diag\left[K\left(\frac{\hat{P}(Z_{1})-p}{h_{n}}\right), \cdots, K\left(\frac{\hat{P}(Z_{n})-p}{h_{n}}\right)\right],& \\
\widehat{\tilde{\mathbf{Y}}}=&\begin{bmatrix}
\hat{\tilde{Y}}_{1} \\
\vdots \\
\hat{\tilde{Y}}_{n}
\end{bmatrix}, \quad
\tilde{\mathbf{Y}}=\begin{bmatrix}
\tilde{Y}_{1} \\
\vdots \\
\tilde{Y}_{n}
\end{bmatrix}, \quad  e_{2}=\begin{bmatrix}
0\\
1 \\
0
\end{bmatrix},&
\end{align*}
respectively.
We define $\hat{\tilde{Y}}_{i}$ and $\tilde{Y}_{i}$ as $Y_{i}-\hat{\beta}_{0}^\prime X_{i}+(\widehat{\beta_{1}-\beta_{0}})^\prime X_{i}\hat{P}(Z_{i})$ and $Y_{i}-\beta_{0}^\prime X_{i}+(\beta_{1}-\beta_{0})^\prime X_{i}P(Z_{i})$, respectively. Under the model assumptions, we obtain
\begin{align*}
&\widehat{MTE}(p,x)-MTE(p,x) \\
=&\left((\widehat{\beta_{1}-\beta_{0}})^{\prime}x+e^{\prime}_{2}\big[\Gamma^{\prime}(\hat{V})\Omega(\hat{V})\Gamma(\hat{V})\big]^{-1}\Gamma^{\prime}(\hat{V})\Omega(\hat{V})\widehat{\tilde{\mathbf{Y}}}\right)-\left((\beta_{1}-\beta_{0})^{\prime}x+\frac{\partial E[\tilde{Y}|P(Z)=p]}{\partial P(Z)}\right)  \\
=&((\widehat{\beta_{1}-\beta_{0}})-(\beta_{1}-\beta_{0}))^{\prime}x+e^{\prime}_{2}\big[\Gamma^{\prime}(\hat{V})\Omega(\hat{V})\Gamma(\hat{V})\big]^{-1}\Gamma^{\prime}(\hat{V})\Omega(\hat{V})\widehat{\tilde{\mathbf{Y}}}-\frac{\partial E[\tilde{Y}|P(Z)=p]}{\partial P(Z)} \\
     =&((\widehat{\beta_{1}-\beta_{0}})-(\beta_{1}-\beta_{0}))^{\prime}x \\
     +&e^{\prime}_{2}\big[\Gamma^{\prime}(\hat{V})\Omega(\hat{V})\Gamma(\hat{V})\big]^{-1}\Gamma^{\prime}(\hat{V})\Omega(\hat{V})(\widehat{\tilde{\mathbf{Y}}}-\tilde{\mathbf{Y}}) \\
     +&e^{\prime}_{2}\big[\Gamma^{\prime}(\hat{V})\Omega(\hat{V})\Gamma(\hat{V})\big]^{-1}\Gamma^{\prime}(\hat{V})\Omega(\hat{V})\tilde{\mathbf{Y}}-\frac{\partial E[\tilde{Y}|P(Z)=p]}{\partial P(Z)} .
\end{align*}
 For any $i$, from \eqref{conditional_Y_estimation}, it follows that
\begin{align*}
    \tilde{Y}_{i}=E[\tilde{Y}_{i}|P(Z)=P(Z_{i})]+\varepsilon_{i},
\end{align*}
where $E[\varepsilon_{i}|P(Z)]=0$. A straightforward application of the Taylor expansion yields
\begin{align*}
    E[\tilde{Y}_{i}|P(Z)=P(Z_{i})]&=E[\tilde{Y}_{i}|P(Z)=p] \\
    &+\frac{\partial E[\tilde{Y}_{i}|P(Z)=p]}{\partial P(Z)}(P(Z_{i})-\hat{P}(Z_{i})) \\
    &+\frac{\partial E[\tilde{Y}_{i}|P(Z)=p]}{\partial P(Z)}(\hat{P}(Z_{i})-p) \\
     &+\frac{\partial^2 E[\tilde{Y}_{i}|P(Z)=p]}{\partial (P(Z))^2}(\hat{P}(Z_{i})-p)^2 \\
     &+2\frac{\partial^2 E[\tilde{Y}_{i}|P(Z)=p]}{\partial (P(Z))^2}(\hat{P}(Z_{i})-p)(P(Z_{i})-\hat{P}(Z_{i})) \\
     &+\frac{\partial^2 E[\tilde{Y}_{i}|P(Z)=p]}{\partial (P(Z))^2}(P(Z_{i})-\hat{P}(Z_{i}))^2 \\
    &+r_{i},
\end{align*}
where $r_{i}$ is a Taylor series reminder term, i.e. for $\lambda_{i}\in(0,1)$,
\begin{equation*}
    r_{i}=\frac{\partial^3 E[\tilde{Y}_{i}|P(Z)=p+\lambda_{i}(P(Z_{i})-p)]}{\partial (P(Z))^3}(P(Z_{i})-p)^3.
\end{equation*}
For each $i$, we set $P_{i}$ as
\begin{align*}
    P_{i}=&\frac{\partial E[\tilde{Y}_{i}|P(Z)=p]}{\partial P(Z)}(P(Z_{i})-\hat{P}(Z_{i}))+2\frac{\partial^2 E[\tilde{Y}_{i}|P(Z)=p]}{\partial (P(Z))^2}(\hat{P}(Z_{i})-p)(P(Z_{i})-\hat{P}(Z_{i})) \\
    +&\frac{\partial^2 E[\tilde{Y}_{i}|P(Z)=p]}{\partial (P(Z))^2}(P(Z_{i})-\hat{P}(Z_{i}))^2 .
\end{align*}
From the above definition, we obtain
\begin{align*}
    \tilde{\mathbf{Y}}=\Gamma(\hat{V})\begin{bmatrix}
E[\tilde{Y}_{i}|P(Z)=p] \\
\frac{\partial E[\tilde{Y}_{i}|P(Z)=p]}{\partial P(Z)} \\
\frac{\partial^2 E[\tilde{Y}_{i}|P(Z)=p]}{\partial (P(Z))^2}
\end{bmatrix}
+P+r+\varepsilon,
\end{align*}
where we define 
\begin{equation*}
\tilde{\mathbf{Y}}=\begin{bmatrix}
\tilde{Y}_{1} \\
\vdots  \\
\tilde{Y}_{n}
\end{bmatrix}, 
\quad
P=
\begin{bmatrix}
P_1 \\
\vdots \\
P_n
\end{bmatrix},
\quad
    r=
\begin{bmatrix}
r_{1} \\
\vdots \\
r_{n}
\end{bmatrix},
\ \text{and} \ \varepsilon=
\begin{bmatrix}
\varepsilon_{1} \\
\vdots \\
\varepsilon_{n}
\end{bmatrix},
\end{equation*}
respectively.

Following the approach used in the proof of Theorem 4 in \textcite{carneiro2009estimating}, we establish
\begin{equation}\label{CL_0}
\begin{split}
&\widehat{MTE}(p,x)-MTE(p,x)  \\
     =&((\widehat{\beta_{1}-\beta_{0}})-(\beta_{1}-\beta_{0}))^{\prime}x  \\
     +&e_2^{\prime}\big[\Gamma^{\prime}(\hat{V})\Omega(\hat{V})\Gamma(\hat{V})\big]^{-1}\Gamma^{\prime}(\hat{V})\Omega(\hat{V})(\widehat{\tilde{\mathbf{Y}}}-\tilde{\mathbf{Y}}) \\     
     +&e_2^{\prime}\big[\Gamma^{\prime}(\hat{V})\Omega(\hat{V})\Gamma(\hat{V})\big]^{-1}\Gamma^{\prime}(\hat{V})\Omega(\hat{V})\tilde{\mathbf{Y}}-e_2^{\prime}\big[\Gamma^{\prime}(\hat{V})\Omega(\hat{V})\Gamma(\hat{V})\big]^{-1}\Gamma^{\prime}(\hat{V})\Omega(\hat{V})(\tilde{\mathbf{Y}}-P-r-\varepsilon) \\
      =&((\widehat{\beta_{1}-\beta_{0}})-(\beta_{1}-\beta_{0}))^{\prime}x \\
     +&e_2^{\prime}\big[\Gamma^{\prime}(\hat{V})\Omega(\hat{V})\Gamma(\hat{V})\big]^{-1}\Gamma^{\prime}(\hat{V})\Omega(\hat{V})(\widehat{\tilde{\mathbf{Y}}}-\tilde{\mathbf{Y}}) \\  
+&e_2^{\prime}\big[\Gamma^{\prime}(\hat{V})\Omega(\hat{V})\Gamma(\hat{V})\big]^{-1}\Gamma^{\prime}(\hat{V})\Omega(\hat{V})P  \\  
     +&e_2^{\prime}\big[\Gamma^{\prime}(\hat{V})\Omega(\hat{V})\Gamma(\hat{V})\big]^{-1}\Gamma^{\prime}(\hat{V})\Omega(\hat{V})\varepsilon+e_2^{\prime}\big[\Gamma^{\prime}(\hat{V})\Omega(\hat{V})\Gamma(\hat{V})\big]^{-1}\Gamma^{\prime}(\hat{V})\Omega(\hat{V})r. 
\end{split}
\end{equation}
Through the results for nonparametric estimation with a generated regressor, we achieve the following convergence rate:
\begin{align}
     e_2^{\prime}\big[\Gamma^{\prime}(\hat{V})\Omega(\hat{V})\Gamma(\hat{V})\big]^{-1}\Gamma^{\prime}(\hat{V})\Omega(\hat{V})(\widehat{\tilde{\mathbf{Y}}}-\tilde{\mathbf{Y}})&=O_{p}(h_{n}^{2}), \label{CL_S-1} \\   
     e_2^{\prime}\big[\Gamma^{\prime}(\hat{V})\Omega(\hat{V})\Gamma(\hat{V})\big]^{-1}\Gamma^{\prime}(\hat{V})\Omega(\hat{V})P &=O_{p}(h_{n}^{2}), \label{CL_S-2} \\  
     e_2^{\prime}\big[\Gamma^{\prime}(\hat{V})\Omega(\hat{V})\Gamma(\hat{V})\big]^{-1}\Gamma^{\prime}(\hat{V})\Omega(\hat{V})r&=O_{p}(h_{n}^{2}), \label{CL_2} \\
e_2^{\prime}\big[\Gamma^{\prime}(\hat{V})\Omega(\hat{V})\Gamma(\hat{V})\big]^{-1}\Gamma^{\prime}(\hat{V})\Omega(\hat{V})\varepsilon&=e_2^{\prime}\big[\Gamma^{\prime}(V)\Omega(V)\Gamma(V)\big]^{-1}\Gamma^{\prime}(V)\Omega(V)\varepsilon \notag \\
&+o_{p}\left(n^{-c}\sqrt{\frac{\log n}{nh^{3}_{n}}}\right). \label{CL_1}
\end{align}
Proofs of all the convergence results are given in Appendix \ref{proof_proof_unif_conv}. From Assumptions \ref{9} and \ref{10}, we know that there exists $\tilde{s}(p)=\lambda s_n(p)+(1-\lambda)\hat{s}(p), \lambda\in (0,1)$ such that
	$\big[s_n^2(p)\big]^{-1/2}=\big[\hat{s}^2(p)\big]^{-1/2}-\frac{1}{2}\big[\tilde{s}^2(p)\big]^{-3/2}(s_n^2(p)-\hat{s}^2(p) ) $, which implies 
	\begin{equation}\label{sd-est}
	    \sup_{p\in \mathcal{P}}\left|\frac{1}{s_n(p)}-\frac{1}{\hat{s}(p)}\right|=O_p(n^{-c}).
	\end{equation}
It holds from \eqref{CL_0} to \eqref{sd-est} that
\begin{align*}
&\sup_{(p,x)\in \mathcal{P}\times\mathcal{X}}\sqrt{nh_n^{3}}&& \left| \frac{\widehat{MTE}(p,x)-MTE(p,x)}{\hat{s}(p)}\right.\\
&&&-\left.
\frac{1}{nh_{n}^{3}f_{P}(p)\mu_{2}(K)}\sum_{i=1}^{n}\varepsilon_{i}K\left(\frac{P(Z_{i})-p}{h_{n}}\right)(P(Z_{i})-p) \right|\\
\leq &\sup_{(p,x)\in \mathcal{P}\times\mathcal{X}}\sqrt{nh_n^{3}}&& 
 \left| \frac{\widehat{MTE}(p,x)-MTE(p,x)}{\hat{s}(p)}-\frac{\widehat{MTE}(p,x)-MTE(p,x)}{s_n(p)} \right|\\
+&\sup_{(p,x)\in \mathcal{P}\times\mathcal{X}}\sqrt{nh_n^{3}}&&\left|\frac{\widehat{MTE}(p,x)-MTE(p,x)}{s_n(p)}\right.\\
&&&-\left.\frac{1}{nh_{n}^{3}f_{P}(p)\mu_{2}(K)s_n(p)}\sum_{i=1}^{n}\varepsilon_{i}K\left(\frac{P(Z_{i})-p}{h_{n}}\right)(P(Z_{i})-p) \right|\\
=&O_p(n^{-c}).&&
\end{align*}
\end{proof}

\renewcommand{\theequation}{C.\arabic{equation}} 
\setcounter{equation}{0}

\section{Proof of Lemma \ref{lemma2}}\label{proof_emp_gauss}

\begin{proof}[Proof of Lemma \ref{lemma2}]
	
	We closely follow the proof of Proposition 3.2 of \textcite{chernozhukov2014gaussian}.
	Let $\mathbb{V}=(0,1)$.	For given $p \in \mathcal{P}$ and $h_{n}>0$, define 
	\begin{align*}
	f_{p,h_{n}}(\varepsilon,t):=c_{n}(p)K\left(\frac{t-p}{h_{n}}\right)\varepsilon (t-p), (\varepsilon,t)\in \mathbb{\varepsilon} \times\mathbb{V},
	\end{align*} 
	where we define $c_{n}(p)$ as follows:
	\begin{equation*}
	    c_{n}(p):=\left\{\frac{1}{h_n^{3}} E\left[ \varepsilon_i^2K^2\left(\frac{P(Z_i)-p}{h_n}\right)(P(Z_i)-p)^2\right]  \right\}^{-1/2}.
	\end{equation*}
	We consider the class of functions $\mathcal{F}_n=\{f_{p,h_n}-E[f_{p,h_n}(\varepsilon_i,P(Z_{i}))]: p\in \mathcal{P} \}$.
Let $Z_n=\sup_{f\in \mathcal{F}_n}\mathbb{G}_nf$, where $\mathbb{G}_nf$ denotes the empirical process indexed by $\mathcal{F}_n$.
	
Under Assumptions \ref{2}-\ref{covariate}, the following conditions are satisfied:
	
	\begin{enumerate}
		\item 
		$\sup_{p\in (0,1)} E[|\varepsilon_{i}|^4|P(Z_{i})=p]<\infty$ holds.
		\item 
		$K(\cdot)$ is a bounded and continuous kernel function on $\mathbb{R}$ , with the function class $\mathcal{K}:= \{t\mapsto (ht +v)K(ht +v): h>0, v\in \mathbb{R}\}$ of VC type and an envelope $ \|K\|_\infty$. 
		\item 
		The distribution of $P(Z_{i})$ has a bounded Lebesgue density on its support.
		\item 
		$h_n \to 0$ and $\log(1/h_n)=O(\log n)$ as $n\to \infty$.
		\item The constant
		$C_{\mathcal{V}}:= \sup_{n\ \geq1}\sup_{p\in \mathcal{P}} |c_n(p)|$ is finite, and  for every fixed $n\geq 1$ and for every $p_m\in \mathcal{P} \to p\in \mathcal{P}$ pointwise, $c_n(p_m) $ converges to $c_n(p)$.		
	\end{enumerate}

Note that $|f_{p,h}-E[f_{p,h}(\varepsilon_i,P(Z_{i}))]| \leq C_{\mathcal{V}} \|K\|_\infty (|\varepsilon|+E[|\varepsilon_i|])$. We thus use
\[
F(\varepsilon, p) \coloneqq C_{\mathcal{V}} \|K\|_\infty (|\varepsilon|+E[|\varepsilon_i|]) \]
as an envelope of $\mathcal{F}_n$. 
As in equation (31) of \textcite{chernozhukov2014gaussian}, we can show that there exist constants $A, r>0$ such that  
\begin{align*}
\sup_QN(\mathcal{F}_n,e_{Q}, u \| F \|_{Q,2}) \leq (A/u)^r, \ 0< \forall u <1, \forall n\geq 1.
\end{align*}
Hence, for every $n\geq 1$, it follows from Lemma 2.1 of \textcite{chernozhukov2014gaussian}  that $\mathcal{F}_n$ is pre-Gaussian and there exists a tight Gaussian random variable $G_n$ in $\ell^\infty(\mathcal{F}_n)$ with mean zero and covariance function 
\begin{align*}
E[G_n(f)G_n(\check{f})]=Cov (f(\varepsilon_i, P(Z_{i})),\check{f}(\varepsilon_i, P(Z_{i}))), \  f,\check{f}\in \mathcal{F}_n.
\end{align*}

Let $M=\sup_{p\in (0,1)} E[\varepsilon_{i}^4|P(Z_{i})=p]$. A straightforward calculation gives
\begin{align*}
E[|f_{p,h_n}(\varepsilon_i, P(Z_{i}))-E[f_{p,h_n}(\varepsilon_i,P(Z_{i}))]|^2 ]&\lesssim E[|f_{p,h_n}(\varepsilon_i, P(Z_{i}))|^2]\\
&\leq (1+M)C_{\mathcal{V}}^2\|f_{P}\|_\infty h_n^{3}\int_{\mathbb{R}}|K(t)t|^2dt,\\
E[|f_{p,h_n}(\varepsilon_i, P(Z_{i}))-E[f_{p,h_n}(\varepsilon_i,P(Z_{i}))]|^3 ]&\lesssim C E[|f_{p,h_n}(\varepsilon_i, P(Z_{i}))|^3]\\
&\leq (1+M)C_{\mathcal{V}}^3\|f_{P}\|_\infty h_n^{4}\int_{\mathbb{R}}|K(t)t|^3dt,\\
E[|f_{p,h_n}(\varepsilon_i, P(Z_{i}))-E[f_{p,h_n}(\varepsilon_i,P(Z_{i}))]|^4 ]&\lesssim C E[|f_{p,h_n}(\varepsilon_i, P(Z_{i}))|^4]\\
&\leq (1+M)C_{\mathcal{V}}^4\|f_{P}\|_\infty h_n^{5}\int_{\mathbb{R}}|K(t)t|^4dt.
\end{align*}
Then, by Corollary 2.2 of \textcite{chernozhukov2014gaussian} with parameters $\gamma=\gamma_n=(\log n)^{-1}$, $b=O(1)$ and $\sigma=\sigma_n=h_n^{3/2}$, there exists a sequence $\tilde{Z}_n$ of random variables such that $\tilde{Z} \eqd \sup_{f\in \mathcal{F}_n}G_nf$ and as $n \to \infty$,
\begin{align*}
|Z_n-\tilde{Z}_n|=O_P(n^{-1/6}h_n\log n +n^{-1/4}h_{n}^{3/4}\log^{5/4}n+n^{-1/4}\log^{3/2}n).
\end{align*}
Define $\tilde{B}_{n,1}(p)=h_n^{-3/2}G_n(f_{p,h_n})$, $p\in \mathcal{P}$ and $\tilde{W_n}=h_n^{-3/2}\tilde{Z_n}$. Then $B_n(p)$ is the desired Gaussian process, and since $W_n=h_n^{-3/2}Z_n$, we have $\tilde{W_n}\eqd \sup_{p\in \mathcal{P}}\tilde{B}_n(p)$ and 
\begin{align*}
&|W_n-\tilde{W_n}|=h_n^{-3/2}|Z_n-\tilde{Z_n}|\\
=&O_P\{(nh_n^{3})^{-1/6}\log n +(nh_n^{3})^{-1/4}\log^{5/4}n+(nh_{n}^{6})^{-1/4}\log^{3/2}n       \}.
\end{align*}
\end{proof}

\renewcommand{\theequation}{D.\arabic{equation}} 
\setcounter{equation}{0}

\section{Proof of Lemma \ref{lemma3}}\label{proof_gauss_stationary}

\begin{proof}[Proof of Lemma \ref{lemma3}]

	This lemma follows the approach of Lemma 3.4 in \textcite{ghosal2000testing}.
Let:
\begin{align*}
\phi_{n,p}(\varepsilon_i,P(Z_{i}))= &\Big[E\big[\varepsilon^{2} K^{2}\left(\frac{P(Z)-p}{h_n}\right) (P(Z)-p)^{2} \big]\Big]^{-1/2} \varepsilon_i K\left(\frac{P(Z_{i})-p}{h_n}\right)(P(Z_i)-p), \\
\psi_{n,p}(\varepsilon_i,P(Z_{i}))=&\Big[h_n^{3}E[\varepsilon^{2}|P(Z)=P(Z_{i})]f_P(P(Z_{i}))\int K^{2}(u)u^{2}du   \Big]^{-1/2} \\
\times &\varepsilon_iK\left(\frac{P(Z_{i})-p}{h_n}\right)(P(Z_i)-p).
\end{align*}
The remaining proof steps follow the approach in Lemma 5 of \textcite{lee2017doublysup}.
We can interpret Gaussian processes $\tilde{B}_{n,1}$ and $\tilde{B}_{n,2}$ as Brownian bridges $B_n(\phi_{n,p})$ and $B_n(\psi_{n,p})$, respectively, in the sense that $E(B_n(g))=0$ and $E[B_n(g)B_n(\check{g})]=\mathrm{Cov} (g,\check{g})$ for $g=\phi_{n,p}$, $\check{g}=\phi_{n,\check{p}}$.
Define $\delta_n(p):=B_n(\phi_{n,p})-B_n(\psi_{n,p})$. The process, $\delta_n(p)$, is a mean zero Gaussian process with
\begin{align*}
E[\delta_n(p)\delta_n(\check{p})]=E\big[(\phi_{n,p}(\varepsilon,P(Z))-\psi_{n,p}(\varepsilon,P(Z)))(\phi_{n,\check{p}}(\varepsilon,P(Z))-\psi_{n,\check{p}}(\varepsilon, P(Z))\big].
\end{align*}
For the proof of Lemma \ref{lemma3}, we use the following result with the proof provided in Appendix \ref{proof_proof_gauss_stationary}.
\begin{lemma}\label{lemma_3_aux}
    Let Assumptions \ref{2}-\ref{covariate} hold. The supremum of the $L_{2}$ diameter of $\delta_{n}(p)$ converges to zero at the same rate as $h_{n}$, namely  $\sup_{p\in\mathcal{P}}E[(\delta_{n}(p)^{2})]=O(h_{n})$. Furthermore, for any $p,p^{\prime}\in \mathcal{P}$, $e_{Q}(\delta_{n}(p),\delta_{n}(p^{\prime}))$ is bounded by $C\sqrt{|p-p^{\prime}|/h_{n}^{3}} $ where $C>0$ is a universal constant.
\end{lemma}

Using Lemma \ref{lemma_3_aux}, we obtain
\begin{align*}
    N(\{\delta_n(p):p\in\mathcal{P}\},  e_{Q},\varepsilon)\leq N\left(\mathcal{P},|\cdot|,\frac{\varepsilon^{2}h_{n}^{3}}{C^2},\right)\lesssim \frac{1}{\varepsilon^{2}h_{n}^{3}}.
\end{align*}
Applying Corollary 2.2.8 of van der Vaart and Wellner (1996), we conclude that
\begin{align*}
    E\left(\sup_{p\in\mathcal{P}}|\delta_{n}(p)|\right)&\lesssim \int^{\infty}_{0}\sqrt{\log  N(\{\delta_n(p):p\in\mathcal{P}\}, e_{Q}} ,\varepsilon) d\varepsilon \\
    &\lesssim \int^{O(h_{n}^{1/2})}_{0}\sqrt{\log \frac{1}{h_{n}^{3/2}\varepsilon}} d\varepsilon \\
    &= O(\sqrt{h_{n}\log h_{n}^{-3/2}}).
\end{align*}

\end{proof}

\renewcommand{\theequation}{E.\arabic{equation}} 
\setcounter{equation}{0}

\section{Proof of Theorem \ref{theorem}}\label{proof_main}

\begin{proof}[Proof of  Theorem \ref{theorem}]
Through a straightforward calculation, for any $t\in\mathbb{R}$, we have
\begin{align*}
    &\left|\Pr\left(\ell_{n} \left[\sup_{(p,x)\in \mathcal{P}\times\mathcal{X}}\sqrt{nh_n^{3}}  \left|\frac{\widehat{MTE}(p,x)-MTE(p,x)}{\hat{s}(p)}\right|-\ell_{n}\right]<t\right)\right. \\
    &\left.-\Pr\left(\ell_{n}\left[\sup_{p\in\mathcal{P}}|\tilde{B}_{n,2}(h^{-1}_{n}p)|-\ell_{n}\right]<t\right)\right| \\
    \leq &\Pr\left(\ell_{n} \left|\sup_{(p,x)\in \mathcal{P}\times\mathcal{X}}\sqrt{nh_n^{3}}   \left|\frac{\widehat{MTE}(p,x)-MTE(p,x)}{\hat{s}(p)}\right|-\sup_{p\in\mathcal{P}}|\tilde{B}_{n,2}(h^{-1}_{n}p)|\right|\geq\varepsilon_{n}\right) \\
    &+\max\left\{\Pr\left(\ell_{n}\left[\sup_{p\in\mathcal{P}}|\tilde{B}_{n,2}(h^{-1}_{n}p)|-\ell_{n}\right]<t+\varepsilon_{n}\right)-\Pr\left(\ell_{n} \left[\sup_{p\in\mathcal{P}}|\tilde{B}_{n,2}(h^{-1}_{n}p)|-\ell_{n}\right]<t\right),\right. \\
    &\left. \Pr\left(\ell_{n} \left[\sup_{p\in\mathcal{P}}|\tilde{B}_{n,2}(h^{-1}_{n}p)|-\ell_{n}\right]<t\right)-\Pr\left(\ell_{n}\left[\sup_{p\in\mathcal{P}}|\tilde{B}_{n,2}(h^{-1}_{n}p)|-\ell_{n}\right]<t-\varepsilon_{n}\right)\right\},
\end{align*}
where $\varepsilon_{n}$ is a sequence converging to zero, i.e. $\varepsilon_{n}\rightarrow0$ as $n$ goes toward infinity. We state the following lemma to achieve the polynomial convergence to the distribution function, with the proof given in Appendix \ref{proof_proof_main}. \begin{lemma}\label{theorem4_aux}
 Let Assumptions \ref{2}-\ref{covariate} hold. There exist real positive numbers $d$ such that
\begin{align*}
    &\Pr\left(\ell_{n} \left|\sup_{(p,x)\in \mathcal{P}\times\mathcal{X}}\sqrt{nh_n^{3}}   \left|\frac{\widehat{MTE}(p,x)-MTE(p,x)}{\hat{s}(p)}\right|-\sup_{p\in\mathcal{P}}|\tilde{B}_{n,2}(h^{-1}_{n}p)|\right|\geq n^{-d}\right) \\
    =&o(1), 
\end{align*}
and that 
\begin{align*}
       &\max\left\{\Pr\left(\ell_{n}\left[\sup_{p\in\mathcal{P}}|\tilde{B}_{n,2}(h^{-1}_{n}p)|-\ell_{n}\right]<t+n^{-d}\right)-\Pr\left(\ell_{n} \left[\sup_{p\in\mathcal{P}}|\tilde{B}_{n,2}(h^{-1}_{n}p)|-\ell_{n}\right]<t\right),\right. \\
    &\left. \Pr\left(\ell_{n} \left[\sup_{p\in\mathcal{P}}|\tilde{B}_{n,2}(h^{-1}_{n}p)|-\ell_{n}\right]<t\right)-\Pr\left(\ell_{n}\left[\sup_{p\in\mathcal{P}}|\tilde{B}_{n,2}(h^{-1}_{n}p)|-\ell_{n}\right]<t-n^{-d}\right)\right\} \\
    =&o(1).
    \end{align*}
\end{lemma}

By Lemma \ref{theorem4_aux}, uniformly in $t$ over any finite interval, we obtain
\begin{align*}
   &\Pr\left(\ell_{n} \left[\sup_{(p,x)\in \mathcal{P}\times\mathcal{X}}\sqrt{nh_n^{3}}  \left|\frac{\widehat{MTE}(v)-MTE(v)}{\hat{s}(p)}\right|-\ell_{n}\right]<t\right) \\
   =&\Pr\left(\ell_{n}\left[\sup_{p\in\mathcal{P}}|\tilde{B}_{n,2}(h^{-1}_{n}p)|-\ell_{n}\right]<t\right)+o(1).
\end{align*}
We note that $\tilde{B}_{n,2}$ defined in Lemma \ref{lemma3} is a homogeneous Gaussian field with mean zero and the covariance function $\rho(p)$. The imposed condition on the kernel function ensures the covariance function $\rho(p)$ has finite support and is six times differentiable, which implies that the Gaussian process $\tilde{B}_{n,2}$ is three times differentiable in the mean square sense. For more details, see Chapter 4 of \textcite{rasmussen2006gaussian}. Using Theorem 14.3 of \textcite{piterbarg1996asymptotic}, we obtain
\begin{align*}
	&\Pr\Big(\ell_n\Big[ \sup_{(p,x)\in \mathcal{P}\times \mathcal{X}}\sqrt{nh_n^3} \left|\frac{\widehat{MTE}(p,x)-MTE(p,x)}{\hat{s}(p)}\right| -\ell_n   \Big]<t \Big) \\
	\label{thmeq}
	=& \exp\big(-2e^{-t-t^2/2\ell_n^2} \big)+o(1),
	\end{align*}
	as $n\to \infty$, where $\ell_n$ is the largest solution to the following equation:
	\begin{align*}
	(b_{0}-a_{0})\cdot h_n^{-1}\sqrt{\lambda}(2\pi)^{-1}\exp(-\ell_n^2/2)=1,
	\end{align*}
	where we define
    \begin{align*}
	\lambda:=-\frac{\int g(u)g^{\prime\prime}(u)du}{\int g^2(u)du}\quad \text{and} \quad g(u):=uK(u).
	\end{align*}
 
\end{proof}

\renewcommand{\theequation}{F.\arabic{equation}} 
\setcounter{equation}{0}

\section{Preliminaries}\label{main-proposition}
\textcite{Dony_Einmahl_Mason_2016} give results about the uniform convergence of kernel functions. We introduce their assumptions with slight modifications.

Let $\mathcal{G}$ denote a class of measurable real valued functions $g$ of $u\in\mathbb{R}^{d}$. Assume the following
\begin{assumption}[Assumptions of \textcite{Dony_Einmahl_Mason_2016}]\label{Assumption-g}
$\ $
\begin{enumerate}
    \item[(G.i)] $\sup_{g\in\mathcal{G}}\|g\|_{\infty}=:\kappa<\infty$,
    \item[(G.ii)] $\sup_{g\in\mathcal{G}}\int_{\mathbb{R}^{d}}g^{2}(x)dx=:L<\infty$.
\end{enumerate}
\end{assumption}
Let $ \mathcal{F}_{\mathcal{G}}$ denote the class of functions of $s\in\mathbb{R}$ formed from $\mathcal{G}$. We define $\mathcal{F}_{\mathcal{G}}$ as follows:
\begin{equation*}
    \mathcal{F}_{\mathcal{G}}:=\{g(z-s\lambda):\lambda\geq1, z\in\mathbb{R}\ \text{and}\ g\in\mathcal{G}\}.
\end{equation*}
Before stating the convergence result for kernel-type functions, we introduce some preferable properties of the class of functions. 
\begin{dfn}
    The class of functions $\mathcal{F}$ is pointwise measurable if there exists a countable subset $\mathcal{G}$ such that for every $f\in\mathcal{F}$, there exists a sequence $g_{m}\in\mathcal{G}$ with $g_{m}(x)\rightarrow f(x)$ for every $x\in S$. 
\end{dfn}
Then, assume $\mathcal{F}_{\mathcal{G}}$ satisfies the following properties:
\begin{assumption}[Modified assumptions of \textcite{Dony_Einmahl_Mason_2016} ]\label{Assumption-f}
$\ $
\begin{enumerate}
    \item[(F.i)] The $ \mathcal{F}_{\mathcal{G}}$ is of  VC-type with envelope $\kappa$.
    \item[(F.ii)] $\mathcal{F}_{\mathcal{G}}$ is pointwise measurable.
\end{enumerate}
\end{assumption}

\begin{remark}\label{def-VC}
\textcite{Dony_Einmahl_Mason_2016} also requires the $ \mathcal{F}_{\mathcal{G}}$ be of  VC-type. However, they use a stronger definition of VC-type: $ \mathcal{F}_{\mathcal{G}}$ is of VC-type with envelope $F_{\mathcal{G}}$ if there exist finite constants $A, v>0$ such that 
\begin{equation*}
    \sup_{Q}N( \mathcal{F}_{\mathcal{G}},e_{Q},\varepsilon\|F_{\mathcal{G}}\|_{Q,2})\leq (A/\varepsilon)^{\nu} 
\end{equation*}
    where the supremum is taken over all probability measures $Q$ for which $0 <\|F\|_{Q,2}<\infty$. Therefore, we need to show that the result of \textcite{Dony_Einmahl_Mason_2016} holds with the weaker condition introduced in Definition \ref{VC-type-def}.
\end{remark}

For any $g\in\mathcal{G}$ and $h$, define $g_{n,h}$ as follows
\begin{equation*}
    g_{n,h}(x):=\frac{1}{nh_{n}}\sum_{i=1}^{n}g\left(\frac{x-P(Z_{i})}{h_{n}}\right),\quad x\in\mathbb{R}^{d}.
\end{equation*}
Then, we obtain uniform convergence results for $g$.
\begin{corollary}[Corollary of Theorem 4.2 of \textcite{Dony_Einmahl_Mason_2016}]\label{donny}
Assuming (G.i), (G.ii), (F.i), (F.ii), and $f_{P}$ bounded, we have for $c>0$ and $0<h_{0}<1$,
\begin{equation*}\label{uniform-general}
\lim\sup_{n\rightarrow\infty}\sup_{\frac{c\log n}{n}\leq h\leq h_{0}}\sup_{g\in\mathcal{G}}\frac{\sqrt{nh}\|g_{n,h}-Eg_{n,h}\|_{\infty}}{\sqrt{|\log h|\vee\log\log n}}=:G(c)<\infty,\quad a.s. 
\end{equation*}
\end{corollary}

We define $\mathcal{G}$ as 
\begin{align}\label{concrete_g}
    \mathcal{G}:=\{&|(-x)^{s}K^{(\ell)}(-x)|\}\quad \text{for $0\leq s\leq 5$ and $0\leq \ell \leq5$.}
\end{align}
To apply the result of Corollary \ref{donny}, we need to guarantee the existence of $E[g_{n,h}]$ for each $g\in\mathcal{G}$ and check whether $\mathcal{G}$ and $\mathcal{F}_{\mathcal{G}}$ satisfy Assumptions (G.i), (G.ii), (F.i), and (F.ii). We have the following general results.
\begin{lemma}\label{VC-type}
Let $g:\mathbb{R}\mapsto\mathbb{R}$. Assume $g$ is a $K$-Lipschitz continuous function with compact support where $K>0$. Consider the class of functions
\begin{equation*}
    \mathcal{F}_{g}:=\{x\mapsto g(tx-s); t>0, s\in\mathbb{R}\}.
\end{equation*}
Then, $\mathcal{F}_{g}$ is of VC-type with the envelope $\|g\|_{\infty}$ and pointwise measurable.
\end{lemma}

\begin{lemma}\label{union-VC}
Let $\mathcal{F}$ and $\mathcal{G}$ is of VC-type with envelopes $F$ and $G$, respectively. Then, $\mathcal{F}\cup\mathcal{G}$ is also of VC-type class with envelope $H(x):=\max\{F(x),G(x)\}$. 
\end{lemma}

\begin{lemma}\label{union-pointwise-measurable}
Let $\mathcal{F}$ and $\mathcal{G}$  pointwise measurable. Then, $\mathcal{F}\cup\mathcal{G}$ is also pointwise measurable. 
\end{lemma}

Under Assumptions \ref{2}-\ref{covariate}, each function in $\mathcal{G}$ satisfies the conditions of Lemma \ref{VC-type}. Because $\mathcal{G}$ includes finite functions, $\mathcal{G}$ and $\mathcal{F}_{\mathcal{G}}$ satisfy Assumptions (G.i), (G.ii), (F.i), and (F.ii) by Lemma \ref{VC-type}, \ref{union-VC}, and \ref{union-pointwise-measurable}. Moreover, we have the following result.
\begin{lemma}\label{expectation}
 Under Assumptions \ref{2}-\ref{covariate}, $E[|g_{n,h}|]$ does exist for any $g\in\mathcal{G}$.
\end{lemma}
Finally, we obtain the uniform convergence result.
\begin{proposition}\label{kernel-uniform}
 Under Assumptions \ref{2}-\ref{covariate}, it holds that
\begin{equation*}
    \sup_{g\in\mathcal{G}}\sup_{p\in[a_{0},b_{0}]}|g_{n,h}(p)|=O_{p}(1).
\end{equation*}
\end{proposition}
The above result implies that
\begin{equation*}
\sup_{p\in[a_{0},b_{0}]}\left(\frac{1}{n}\sum_{i=1}^{n}\frac{1}{h_{n}}\left|K^{(\ell)}\left(\frac{P(Z_{i})-p}{h_{n}}\right)\left(\frac{P(Z_{i})-p}{h_{n}}\right)^{s}\right|\right)=O_{p}(1), 
\end{equation*}
for $0\leq s\leq 5$ and $0\leq \ell \leq5$.

We close this section by stating some auxiliary results that are useful to prove results in Appendix  \ref{proof_unif_conv}. The proofs are stored in Appendix \ref{aux_aux}.

\begin{proposition}\label{aux-1}
    Under Assumptions \ref{2}-\ref{covariate}, for any integer $0\leq s\leq 5$ and $0\leq \ell$, we have
    \begin{align*}
        \sup_{p\in[a_{0},b_{0}]}\left(\frac{1}{n}\sum_{i=1}^{n}\frac{1}{h_{n}}\left|\left(K\left(\frac{\hat{P}(Z_{i})-p}{h_{n}}\right)-K\left(\frac{P(Z_{i})-p}{h_{n}}\right)\right)\left(\frac{P(Z_{i})-p}{h_{n}}\right)^{s}\right|\right)&=o_{p}(h_{n}), \\
         \sup_{p\in[a_{0},b_{0}]}\left(\frac{1}{n}\sum_{i=1}^{n}\frac{1}{h_{n}}\left|K\left(\frac{\hat{P}(Z_{i})-p}{h_{n}}\right)\left(\frac{P(Z_{i})-p}{h_{n}}\right)^{s}\left(\frac{\hat{P}(Z_{i})-P(Z_{i})}{h_{n}^{3}}\right)^{\ell}\right|\right)&=o_{p}(1).
    \end{align*}
\end{proposition}

\begin{proposition}\label{aux-2}
    Under Assumptions \ref{2}-\ref{covariate}, for any integer $0\leq s\leq 2$, $0\leq k$ and $\ell\in\{1,\cdots,d\}$, we have
    \begin{align*}
   \sup_{p\in[a_{0},b_{0}]}\left| \frac{1}{nh_{n}}\sum_{i=1}^{n}K\left(\frac{\hat{P}(Z_{i})-p}{h_{n}}\right)\left(\frac{\hat{P}(Z_{i})-p}{h_{n}}\right)^{s}X_{i,\ell}\right|=O_{p}(1), 
    \\
     \sup_{p\in[a_{0},b_{0}]}\left|\frac{1}{nh_{n}}\sum_{i=1}^{n}K\left(\frac{\hat{P}(Z_{i})-p}{h_{n}}\right)\left(\frac{\hat{P}(Z_{i})-p}{h_{n}}\right)^{s}X_{i,\ell}P(Z_{i})\right|=O_{p}(1), \\
       \sup_{p\in[a_{0},b_{0}]}\left|\frac{1}{nh_{n}}\sum_{i=1}^{n}K\left(\frac{\hat{P}(Z_{i})-p}{h_{n}}\right)\left(\frac{\hat{P}(Z_{i})-p}{h_{n}}\right)^{s}X_{i,\ell}\left(\frac{\hat{P}(Z_{i})-P(Z_{i})}{h_{n}^{3}}\right)^{k}\right|=o_{p}(1).
\end{align*}
\end{proposition}

\renewcommand{\theequation}{G.\arabic{equation}} 
\setcounter{equation}{0}

\section{Proofs in Appendix  \ref{proof_unif_conv}}\label{proof_proof_unif_conv}

\subsection{The Convergence Rate of Inverse Matrix}
\begin{lemma}\label{inverse-convergence}
Suppose Assumptions \ref{2}-\ref{covariate} holds. Then, uniformly in $p\in[a_{0},b_{0}]$, we have the following result. 
\begin{align*}
    \frac{1}{nh_{n}}\sum_{i=1}^{n}K\left(\frac{\hat{P}(Z_{i})-p}{h_{n}}\right)&=f_{P}(p)+o_{p}(h_{n})+O_{P}\left(\sqrt{\frac{\log n}{nh_{n}}}\right), \\
    \frac{1}{nh_{n}}\sum_{i=1}^{n}K\left(\frac{\hat{P}(Z_{i})-p}{h_{n}}\right)(\hat{P}(Z_{i})-p)&=\mu_{2}(K)h^{2}_{n}(\partial f_{P}(p)/\partial p)+o_{p}(h_{n}^{2})+O_{p}\left(h_{n}\sqrt{\frac{\log n}{n h_{n}}}\right), \\
    \frac{1}{nh_{n}}\sum_{i=1}^{n}K\left(\frac{\hat{P}(Z_{i})-p}{h_{n}}\right)(\hat{P}(Z_{i})-p)^{2}&=\mu_{2}(K)f_{P}(p)h_{n}^{2}+o_{p}(h_{n}^{3})+O_{p}\left(h_{n}^{2}\sqrt{\frac{\log n}{n h_{n}}}\right),    \\
   \frac{1}{nh_{n}}\sum_{i=1}^{n}K\left(\frac{\hat{P}(Z_{i})-p}{h_{n}}\right)(\hat{P}(Z_{i})-p)^{3}&=\mu_{4}(K)(\partial f_{P}(p)/\partial p)h_{n}^{4}+o_{p}(h_{n}^{4})+O_{p}\left(h_{n}^{3}\sqrt{\frac{\log n}{n h_{n}}}\right),    \\
    \frac{1}{nh_{n}}\sum_{i=1}^{n}K\left(\frac{\hat{P}(Z_{i})-p}{h_{n}}\right)(\hat{P}(Z_{i})-p)^{4}&=\mu_{4}(K)f_{P}(p)h_{n}^{4}+o_{p}(h_{n}^{5})+O_{p}\left(h_{n}^{4}\sqrt{\frac{\log n}{n h_{n}}}\right).
\end{align*}
\end{lemma}
By combining the result of Lemma \ref{inverse-convergence}, we obtain
\begin{align*}
&e_{2}^{\prime}\big[\frac{1}{n}\Gamma^{\prime}(\hat{V})\Omega(\hat{V})\Gamma(\hat{V})\big]^{-1} \\
=&\begin{bmatrix}
o_{p}(1), &(f_{P}(p)\mu_{2}(K)h_{n}^{2})^{-1}(1+o_{p}(n^{-c})), & -f^{\prime}_{P}(p)(f^2_{P}(p)\mu_{2}(K)h_{n}^{2})^{-1}(1+o_{p}(1))
\end{bmatrix}.
\end{align*}
\begin{proof}[Proof of  Lemma \ref{inverse-convergence}]
It holds from Corollary \ref{donny}  and the expected value of kernel functions that, uniformly in $p\in[a_{0},b_{0}]$,
\begin{align*}
    \frac{1}{nh_{n}}\sum_{i=1}^{n}K\left(\frac{P(Z_{i})-p}{h_{n}}\right)&=f_{P}(p)+o\left(h_{n}\right)+O_{P}\left(\sqrt{\frac{\log n}{nh_{n}}}\right), \\
    \frac{1}{nh_{n}}\sum_{i=1}^{n}K\left(\frac{P(Z_{i})-p}{h_{n}}\right)(P(Z_{i})-p)&=\mu_{2}(K)h^{2}_{n}(\partial f_{P}(p)/\partial p)+o(h_{n}^{2})+O_{p}\left(h_{n}\sqrt{\frac{\log n}{n h_{n}}}\right), \\
   \frac{1}{nh_{n}}\sum_{i=1}^{n}K\left(\frac{P(Z_{i})-p}{h_{n}}\right)(P(Z_{i})-p)^{2}&=\mu_{2}(K)f_{P}(p)h_{n}^{2}+o(h_{n}^{3})+O_{p}\left(h_{n}^{2}\sqrt{\frac{\log n}{n h_{n}}}\right),    \\
   \frac{1}{nh_{n}}\sum_{i=1}^{n}K\left(\frac{P(Z_{i})-p}{h_{n}}\right)(P(Z_{i})-p)^{3}&=\mu_{4}(K)(\partial f_{P}(p)/\partial p)h_{n}^{4}+o(h_{n}^{4})+O_{p}\left(h_{n}^{3}\sqrt{\frac{\log n}{n h_{n}}}\right),    \\
    \frac{1}{nh_{n}}\sum_{i=1}^{n}K\left(\frac{P(Z_{i})-p}{h_{n}}\right)(P(Z_{i})-p)^{4}&=\mu_{4}(K)f_{P}(p)h_{n}^{4}+o(h_{n}^{5})+O_{p}\left(h_{n}^{4}\sqrt{\frac{\log n}{n h_{n}}}\right).
\end{align*}
Therefore, we need to show the following convergence result for $0\leq s\leq 5$:
\begin{align}
    \sup_{p\in[a_{0},b_{0}]}&\left|\frac{1}{n}\sum_{i=1}^{n}\frac{1}{h_{n}}K\left(\frac{\hat{P}(Z_{i})-p}{h_{n}}\right)\left(\frac{\hat{P}(Z_{i})-p}{h_{n}}\right)^{s}-\frac{1}{n}\sum_{i=1}^{n}\frac{1}{h_{n}}K\left(\frac{P(Z_{i})-p}{h_{n}}\right)\left(\frac{P(Z_{i})-p}{h_{n}}\right)^{s}\right| \notag \\
    =&o_{p}(h_{n}). \label{inverse_core}
\end{align}

For the case $s=0$, \eqref{inverse_core} trivially holds from Proposition \ref{aux-1}. For $1\leq s\leq 5$, a straightforward calculation gives
\begin{align*}
& \sup_{p\in[a_{0},b_{0}]}\left|\frac{1}{n}\sum_{i=1}^{n}\frac{1}{h_{n}}K\left(\frac{\hat{P}(Z_{i})-p}{h_{n}}\right)\left(\frac{\hat{P}(Z_{i})-p}{h_{n}}\right)^{s}-\frac{1}{n}\sum_{i=1}^{n}\frac{1}{h_{n}}K\left(\frac{P(Z_{i})-p}{h_{n}}\right)\left(\frac{P(Z_{i})-p}{h_{n}}\right)^{s}\right| \notag \\
    \leq& \sum_{\ell=0}^{s-1}sC_{\ell}
\sup_{p\in[a_{0},b_{0}]}\left|\frac{1}{n}\sum_{i=1}^{n}\frac{1}{h_{n}}K\left(\frac{\hat{P}(Z_{i})-p}{h_{n}}\right)\left(\frac{P(Z_{i})-p}{h_{n}}\right)^{\ell}\left(\frac{\hat{P}(Z_{i})-P(Z_{i})}{h_{n}}\right)^{s-\ell}\right| \notag \\
    +&\sup_{p\in[a_{0},b_{0}]}\left|\frac{1}{n}\sum_{i=1}^{n}\frac{1}{h_{n}}\left(K\left(\frac{\hat{P}(Z_{i})-p}{h_{n}}\right)-K\left(\frac{P(Z_{i})-p}{h_{n}}\right)\right)\left(\frac{P(Z_{i})-p}{h_{n}}\right)^{s}\right|. 
\end{align*}
where $sC_{\ell}$ refers to a combination.
By Proposition \ref{aux-1},  we have
\begin{align*}
    \sup_{p\in[a_{0},b_{0}]}\left|\frac{1}{n}\sum_{i=1}^{n}\frac{1}{h_{n}}K\left(\frac{\hat{P}(Z_{i})-p}{h_{n}}\right)\left(\frac{P(Z_{i})-p}{h_{n}}\right)^{\ell}\left(\frac{\hat{P}(Z_{i})-P(Z_{i})}{h_{n}}\right)^{s-\ell}\right|&=o_{p}(h^2_{n}), \\
    \sup_{p\in[a_{0},b_{0}]}\left|\frac{1}{n}\sum_{i=1}^{n}\frac{1}{h_{n}}\left(K\left(\frac{\hat{P}(Z_{i})-p}{h_{n}}\right)-K\left(\frac{P(Z_{i})-p}{h_{n}}\right)\right)\left(\frac{P(Z_{i})-p}{h_{n}}\right)^{s}\right|&=o_{p}(h_{n})
\end{align*}
for $1\leq s\leq 5$ and $0\leq s-1$. Therefore, required result \eqref{inverse_core} holds for $0\leq s\leq 5$.

\end{proof}

\subsection{\eqref{CL_S-1} and \eqref{CL_S-2}}

\begin{lemma}\label{unifrom_dif}
Under Assumptions \ref{2}-\ref{covariate}, we have following results:
\begin{align*}
     \frac{1}{nh_{n}}\sum_{i=1}^{n}K\left(\frac{\hat{P}(Z_{i})-p}{h_{n}}\right)(\hat{\tilde{Y}}_{i}-\tilde{Y}_{i})&=O_{p}(h^{2}_{n}), \\
     \frac{1}{nh^{3}_{n}}\sum_{i=1}^{n}K\left(\frac{\hat{P}(Z_{i})-p}{h_{n}}\right)(\hat{P}(Z_{i})-p)(\hat{\tilde{Y}}_{i}-\tilde{Y}_{i})&=O_{p}(h^{2}_{n}), \\
     \frac{1}{nh^{3}_{n}}\sum_{i=1}^{n}K\left(\frac{\hat{P}(Z_{i})-p}{h_{n}}\right)(\hat{P}(Z_{i})-p)^2(\hat{\tilde{Y}}_{i}-\tilde{Y}_{i})&=O_{p}(h^{2}_{n}),
\end{align*}
and 
\begin{align*}
     \frac{1}{nh_{n}}\sum_{i=1}^{n}K\left(\frac{\hat{P}(Z_{i})-p}{h_{n}}\right)P_{i}&=O_{p}(h^{2}_{n}), \\
     \frac{1}{nh^{3}_{n}}\sum_{i=1}^{n}K\left(\frac{\hat{P}(Z_{i})-p}{h_{n}}\right)(\hat{P}(Z_{i})-p)P_{i}&=O_{p}(h^{2}_{n}), \\
     \frac{1}{nh^{3}_{n}}\sum_{i=1}^{n}K\left(\frac{\hat{P}(Z_{i})-p}{h_{n}}\right)(\hat{P}(Z_{i})-p)^2P_{i}&=O_{p}(h^{2}_{n}).
\end{align*}
\end{lemma}
When Lemma \ref{inverse-convergence} and \ref{unifrom_dif}  hold, we have
\begin{align*}
    &e_{2}^{\prime}\left[\frac{1}{n}\Gamma^{\prime}(\hat{V})\Omega(\hat{V})\Gamma(\hat{V})\right]^{-1}\frac{1}{n}\Gamma^{\prime}(\hat{V})\Omega(\hat{V})(\widehat{\tilde{\mathbf{Y}}}-\tilde{\mathbf{Y}}) \\
    =&o_{p}(1) \frac{1}{nh_{n}}\sum_{i=1}^{n}K\left(\frac{\hat{P}(Z_{i})-p}{h_{n}}\right)(\hat{\tilde{Y}}_{i}-\tilde{Y}_{i}) \\
    +&\left(\mu_{2}(K)f_{P}(p)\right)^{-1}(1+o_{p}(n^{-c}))\frac{1}{nh^{3}_{n}}\sum_{i=1}^{n}K\left(\frac{\hat{P}(Z_{i})-p}{h_{n}}\right)(\hat{P}(Z_{i})-p)(\hat{\tilde{Y}}_{i}-\tilde{Y}_{i}) \\
    +&O_{p}(1)\frac{1}{nh^{3}_{n}}\sum_{i=1}^{n}K\left(\frac{\hat{P}(Z_{i})-p}{h_{n}}\right)(\hat{P}(Z_{i})-p)^2(\hat{\tilde{Y}}_{i}-\tilde{Y}_{i})\\
    =&O_{p}(h_{n}^{2}). \tag{\ref{CL_S-1}}
\end{align*}
Similarly, we also obtain
\begin{equation*}
    e_{2}^{\prime}\left[\frac{1}{n}\Gamma^{\prime}(\hat{V})\Omega(\hat{V})\Gamma(\hat{V})\right]^{-1}\frac{1}{n}\Gamma^{\prime}(\hat{V})\Omega(\hat{V})P=O_{p}(h_{n}^{2}). \tag{\ref{CL_S-2}}
\end{equation*}

\begin{proof}[Proof of Lemma \ref{unifrom_dif}]

By definition, for each $i$, we have
\begin{align*}
    \hat{\tilde{Y}}_{i}-\tilde{Y}_{i}=&(Y_{i}-\hat{\beta}_{0}^\prime X_{i}+(\widehat{\beta_{1}-\beta_{0}})^\prime X_{i}\hat{P}(Z_{i}))-(Y_{i}-\beta_{0}^\prime X_{i}+(\beta_{1}-\beta_{0})^\prime X_{i}P(Z_{i})).
\end{align*}
Then, we have
\begin{align*}
    &\frac{1}{nh_{n}}\sum_{i=1}^{n}K\left(\frac{\hat{P}(Z_{i})-p}{h_{n}}\right)(\hat{\tilde{Y}}_{i}-\tilde{Y}_{i}) \\
    =&\sum_{\ell=1}^{d}(\beta_{0}-\hat{\beta}_{0})_{\ell}\frac{1}{nh_{n}}\sum_{i=1}^{n}K\left(\frac{\hat{P}(Z_{i})-p}{h_{n}}\right)X_{i,\ell} \\
    +&\sum_{\ell=1}^{d}\left(\widehat{(\beta_{1}-\beta_{0})}-(\beta_{1}-\beta_{0})\right)_{\ell}\frac{1}{nh_{n}}\sum_{i=1}^{n}K\left(\frac{\hat{P}(Z_{i})-p}{h_{n}}\right)X_{i,\ell}P(Z_{i}) \\
    -&\sum_{\ell=1}^{d}\left(\widehat{\beta_{1}-\beta_{0}}\right)_{\ell}\frac{1}{nh_{n}}\sum_{i=1}^{n}K\left(\frac{\hat{P}(Z_{i})-p}{h_{n}}\right)X_{i,\ell}(P(Z_{i})-\hat{P}(Z_{i}))
\end{align*}
and, for $1\leq s\leq 2$,
\begin{align*}
    &\frac{1}{nh^{3}_{n}}\sum_{i=1}^{n}K\left(\frac{\hat{P}(Z_{i})-p}{h_{n}}\right)(\hat{P}(Z_{i})-p)^s(\hat{\tilde{Y}}_{i}-\tilde{Y}_{i}) \\
    =&\sum_{\ell=1}^{d}(\beta_{0}-\hat{\beta}_{0})_{\ell}\frac{1}{nh^3_{n}}\sum_{i=1}^{n}K\left(\frac{\hat{P}(Z_{i})-p}{h_{n}}\right)(\hat{P}(Z_{i})-p)^sX_{i,\ell} \\
    +&\sum_{\ell=1}^{d}\left(\widehat{(\beta_{1}-\beta_{0})}-(\beta_{1}-\beta_{0})\right)_{\ell}\frac{1}{nh^3_{n}}\sum_{i=1}^{n}K\left(\frac{\hat{P}(Z_{i})-p}{h_{n}}\right)(\hat{P}(Z_{i})-p)^sX_{i,\ell}P(Z_{i}) \\
    -&\sum_{\ell=1}^{d}\left(\widehat{\beta_{1}-\beta_{0}}\right)_{\ell}\frac{1}{nh^3_{n}}\sum_{i=1}^{n}K\left(\frac{\hat{P}(Z_{i})-p}{h_{n}}\right)(\hat{P}(Z_{i})-p)^sX_{i,\ell}(P(Z_{i})-\hat{P}(Z_{i})).
\end{align*}
Hence, the required result holds from Assumptions \ref{2}-\ref{covariate} and Proposition \ref{aux-2}.

For the latter part of the lemma, by definition of $P_{i}$, we have
\begin{align*}
    &\frac{1}{nh_{n}}\sum_{i=1}^{n}K\left(\frac{\hat{P}(Z_{i})-p}{h_{n}}\right)P_{i} \\
    =&\frac{\partial E[\tilde{Y}_{i}|P(Z)=p]}{\partial P(Z)}\frac{1}{nh_{n}}\sum_{i=1}^{n}K\left(\frac{\hat{P}(Z_{i})-p}{h_{n}}\right)
    (P(Z_{i})-\hat{P}(Z_{i})) \\
    +&\frac{\partial^2 E[\tilde{Y}_{i}|P(Z)=p]}{\partial (P(Z))^2}\frac{1}{nh_{n}}\sum_{i=1}^{n}K\left(\frac{\hat{P}(Z_{i})-p}{h_{n}}\right)\left(2(\hat{P}(Z_{i})-p)(P(Z_{i})-\hat{P}(Z_{i}))+(P(Z_{i})-\hat{P}(Z_{i}))^2\right)
\end{align*}
and, for $1\leq s\leq 2$,
\begin{align*}
        &\frac{1}{nh^3_{n}}\sum_{i=1}^{n}K\left(\frac{\hat{P}(Z_{i})-p}{h_{n}}\right)(\hat{P}(Z_{i})-p)^sP_{i} \\
    =&\frac{\partial E[\tilde{Y}_{i}|P(Z)=p]}{\partial P(Z)}\frac{1}{nh^3_{n}}\sum_{i=1}^{n}K\left(\frac{\hat{P}(Z_{i})-p}{h_{n}}\right)(\hat{P}(Z_{i})-p)^s
    (P(Z_{i})-\hat{P}(Z_{i})) \\
    +&\frac{\partial^2 E[\tilde{Y}_{i}|P(Z)=p]}{\partial (P(Z))^2}\frac{1}{nh^{3}_{n}}\sum_{i=1}^{n}K\left(\frac{\hat{P}(Z_{i})-p}{h_{n}}\right)2(\hat{P}(Z_{i})-p)^{s+1}(P(Z_{i})-\hat{P}(Z_{i})) \\
    +&\frac{\partial^2 E[\tilde{Y}_{i}|P(Z)=p]}{\partial (P(Z))^2}\frac{1}{nh^3_{n}}\sum_{i=1}^{n}K\left(\frac{\hat{P}(Z_{i})-p}{h_{n}}\right)(P(Z_{i})-\hat{P}(Z_{i}))^2(\hat{P}(Z_{i})-p)^{s}.
\end{align*}
Because  first and second derivatives of $m(p)$ are continuous on $[a_{0},b_{0}]$, those functions are bounded. Therefore, required results hold from Proposition \ref{aux-1} and the binomial theorem as in the proof of Lemma \ref{inverse-convergence}.
\end{proof}

\subsection{\eqref{CL_2}}
In this section, we consider the convergence rate of the following terms:
\begin{equation*}
    \frac{1}{n}\Gamma^{\prime}(\hat{V})\Omega(\hat{V})r=\begin{bmatrix}
 \frac{1}{n}\sum_{i=1}^{n}h_n^{-1}K\left(\frac{\hat{P}(Z_{i})-p}{h_n}\right)\frac{\partial^3 E[Y_{i}|P(Z)=p+\lambda_{i}(P(Z_{i})-p)]}{\partial (P(Z))^3}(P(Z_{i})-p)^{3} \\
\frac{1}{n}\sum_{i=1}^{n}h_n^{-1}K\big(\frac{\hat{P}(Z_{i})-p}{h_n}\big)\frac{\partial^3 E[Y_{i}|P(Z)=p+\lambda_{i}(P(Z_{i})-p)]}{\partial (P(Z))^{3}}(P(Z_{i})-p)^{3}(\hat{P}(Z_{i})-p)  \\
\frac{1}{n}\sum_{i=1}^{n}h_n^{-1}K\big(\frac{\hat{P}(Z_{i})-p}{h_n}\big)\frac{\partial^3 E[Y_{i}|P(Z)=p+\lambda_{i}(P(Z_{i})-p)]}{\partial (P(Z))^{3}}(P(Z_{i})-p)^{3}(\hat{P}(Z_{i})-p)^2  
\end{bmatrix}
\end{equation*}
where $\lambda_{i}\in(0,1)$.
\begin{lemma}\label{r-convergence}
Suppose Assumptions \ref{2}-\ref{covariate} hold. Then, uniformly in $p\in[a_{0},b_{0}]$, we have the following result.
\begin{align*}
    &\frac{1}{nh_{n}}\sum_{i=1}^{n}K\left(\frac{\hat{P}(Z_{i})-p}{h_n}\right)m_{3}(P(Z_{i}))(P(Z_{i})-p)^{3}=o_{p}(h_{n}^{2}), \\
    &\frac{1}{nh_{n}^{3}}\sum_{i=1}^{n}K\left(\frac{\hat{P}(Z_{i})-p}{h_n}\right)m_{3}(P(Z_{i}))(P(Z_{i})-p)^{3}(\hat{P}(Z_{i})-p)=O_{p}(h_{n}^{2}), \\
    &\frac{1}{nh_{n}^{3}}\sum_{i=1}^{n}K\left(\frac{\hat{P}(Z_{i})-p}{h_n}\right)m_{3}(P(Z_{i}))(P(Z_{i})-p)^{3}(\hat{P}(Z_{i})-p)^2=o_{p}(h_{n}^{3}),
\end{align*}
where $m_{3}(P(Z_{i}))=\partial^{3} E[Y_{i}|P(Z)=p+\lambda_{i}(P(Z_{i})-p)]/\partial (P(Z))^3$ and $\lambda_{i}\in(0,1)$.
\end{lemma}

When Lemma \ref{inverse-convergence} and \ref{r-convergence}  hold, we have
\begin{align*}
    &e_{2}^{\prime}\left[\frac{1}{n}\Gamma^{\prime}(\hat{V})\Omega(\hat{V})\Gamma(\hat{V})\right]^{-1}\frac{1}{n}\Gamma^{\prime}(\hat{V})\Omega(\hat{V})r \\
    =&o_{p}(1)\frac{1}{nh_{n}}\sum_{i=1}^{n}K\left(\frac{\hat{P}(Z_{i})-p}{h_n}\right)m_{3}(P(Z_{i}))(P(Z_{i})-p)^{3} \\
    +&\left(\mu_{2}(K)f_{P}(p)\right)^{-1}(1+o_{p}(n^{-c}))\frac{1}{nh_{n}^{3}}\sum_{i=1}^{n}K\left(\frac{\hat{P}(Z_{i})-p}{h_n}\right)m_{3}(P(Z_{i}))(P(Z_{i})-p)^{3}(\hat{P}(Z_{i})-p) \\
    +&O_{p}(1)\frac{1}{nh_{n}^{3}}\sum_{i=1}^{n}K\left(\frac{\hat{P}(Z_{i})-p}{h_n}\right)m_{3}(P(Z_{i}))(P(Z_{i})-p)^{3}(\hat{P}(Z_{i})-p)^2\\
    =&O_{p}(h_{n}^{2}). \tag{\ref{CL_2}}
\end{align*}

\begin{proof}[Proof of Lemma \ref{r-convergence}]

Because $m_{3}(P(Z_{i}))$ is uniformly bounded in $(0,1)$ by assumption, it holds from Corollary \ref{donny} and the expected value of kernel functions that, uniformly in $p\in[a_{0},b_{0}]$, 
\begin{align*}
    \frac{1}{nh_{n}}\sum_{i=1}^{n}K\left(\frac{P(Z_{i})-p}{h_{n}}\right)m_{3}(P(Z_{i}))(P(Z_{i})-p)^{3}&=o_{p}(h_{n}^{3}), \\
    \frac{1}{nh_{n}^{3}}\sum_{i=1}^{n}K\left(\frac{P(Z_{i})-p}{h_n}\right)m_{3}(P(Z_{i}))(P(Z_{i})-p)^{4}&=O_{p}(h_{n}^{2}), \\
    \frac{1}{nh_{n}^{3}}\sum_{i=1}^{n}K\left(\frac{P(Z_{i})-p}{h_n}\right)m_{3}(P(Z_{i}))(P(Z_{i})-p)^{5}&=o_{p}(h_{n}^{3}).
\end{align*}
Moreover,  for $3\leq s\leq 5$, and $1\leq\ell\leq 2$, it holds from  Proposition \ref{aux-1} and uniform boundedness of $m_{3}(P(Z_{i}))$ that
\begin{align*}
    \sup_{p\in[a_{0},b_{0}]}&\left|\frac{1}{n}\sum_{i=1}^{n}\frac{1}{h_{n}}m_{3}(P(Z_{i}))\left(K\left(\frac{\hat{P}(Z_{i})-p}{h_{n}}\right)-K\left(\frac{P(Z_{i})-p}{h_{n}}\right)\right)\left(\frac{P(Z_{i})-p}{h_{n}}\right)^{s}\right| \notag \\
    =&o_{p}(h_{n})
\end{align*}
and 
\begin{align*}
    \sup_{p\in[a_{0},b_{0}]}\left(\frac{1}{n}\sum_{i=1}^{n}\frac{1}{h_{n}}m_{3}(P(Z_{i}))\left|K\left(\frac{\hat{P}(Z_{i})-p}{h_{n}}\right)\left(\frac{P(Z_{i})-p}{h_{n}}\right)^{s}\left(\frac{\hat{P}(Z_{i})-P(Z_{i})}{h_{n}^{3}}\right)^{\ell}\right|\right)&=o_{p}(1).
\end{align*}
Therefore, the required results hold.
\end{proof}

\subsection{\eqref{CL_1}}
Before considering the uniform convergence rate of $(1/n)\Gamma^{\prime}(\hat{V})\Omega(\hat{V})\varepsilon$, we investigate the convergence rate with the true propensity score. 
\begin{lemma}\label{unifrom_varepsilon_kernel}
Under Assumptions \ref{2}-\ref{covariate}, we have the following result:
\begin{align*}
        \sup_{p\in[a_{0},b_{0}]}\left|\frac{1}{nh_{n}}\sum_{i=1}^{n}\varepsilon_{i}K\left(\frac{P(Z_{i})-p}{h_{n}}\right)\right|&=O_{p}\left(\sqrt{\frac{\log n}{nh_{n}}}\right), \\
        \sup_{p\in[a_{0},b_{0}]}\left|\frac{1}{nh_{n}^{3}}\sum_{i=1}^{n}\varepsilon_{i}K\left(\frac{P(Z_{i})-p}{h_{n}}\right)(P(Z_{i})-p)\right|&=O_{p}\left(\sqrt{\frac{\log n}{nh_{n}^{3}}}\right), \\
         \sup_{p\in[a_{0},b_{0}]}\left|\frac{1} {nh_{n}^{5}}\sum_{i=1}^{n}\varepsilon_{i}K\left(\frac{P(Z_{i})-p}{h_{n}}\right)(P(Z_{i})-p)^2\right|&=O_{p}\left(\sqrt{\frac{\log n}{nh_{n}^{5}}}\right).
\end{align*}
\end{lemma}

\begin{proof}[Proof of Lemma \ref{unifrom_varepsilon_kernel}]
First, we formalize the uniform convergence problem for an empirical process. For the class of functions $\mathcal{F}$, define $\mathbb{G}_{n}$ and $\|\mathbb{G}_{n}\|_{\mathcal{F}}$ as follows:
\begin{equation*}
    \mathbb{G}_{n}f:=\frac{1}{\sqrt{n}}\sum_{i=1}^{n}(f(X_{i})-E[f(X_{i})]), \quad \|\mathbb{G}_{n}\|_{\mathcal{F}}:=\sup_{f\in\mathcal{F}}|\mathbb{G}_{n}f|.
\end{equation*}
Set $\mathcal{K}^{s}$ and $\mathcal{K}^{s}_{0}$ as 
\begin{align*} 
    \mathcal{K}^{s}&:=\left\{(u,t)\mapsto uK\left(\frac{t-x}{h_{n}}\right)\left(\frac{t-x}{h_{n}}\right)^{s}: x\in [a_{0},b_{0}]\right\}\quad \text{for $s=0,1$ and $2$,} \\
    \mathcal{K}^{s}_{0}&:=\left\{t \mapsto K\left(\frac{t-x}{h_{n}}\right)\left(\frac{t-x}{h_{n}}\right)^{s}: x\in [a_{0},b_{0}]\right\}\quad \text{for $s=0,1$and $2$}.
\end{align*}
Therefore, we have
\begin{align*}
        &\sup_{p\in[a_{0},b_{0}]}\left|\frac{1}{nh_{n}}\sum_{i=1}^{n}\varepsilon_{i}K\left(\frac{P(Z_{i})-p}{h_{n}}\right)\left(\frac{(P(Z_{i})-p)}{h_{n}}\right)^{s}\right| \\
        =&\frac{\sqrt{n}}{nh_{n}}\sup_{k\in\mathcal{K}^{s}}\left|\frac{1}{\sqrt{n}}\sum_{i=1}^{n}k\left(\varepsilon_{i},\frac{P(Z_{i})-p}{h_{n}}\right)\right| \\
        =&\frac{\sqrt{n}}{nh_{n}}\|\mathbb{G}_{n}\|_{\mathcal{K}^{s}}
\end{align*}
where the last equality holds $E[\varepsilon_{i}|P(Z_{i})=p]=0$ for any $p\in(0,1)$. By Corollary A.1 of \textcite{chernozhukov2014gaussian}, if $\mathcal{K}^{s}_{0}$ is a VC-type class,  $\mathcal{K}^{s}$ is a VC-type class. Let $ \mathcal{F}_{\mathcal{G}}$ denote the class of functions  $\{g(z-s\lambda):\lambda\geq1, z\in\mathbb{R}\ \text{and}\ g\in\mathcal{G}\}$ where $\mathcal{G}$ is defined in \eqref{concrete_g}. Because $\mathcal{K}^{s}_{0}\subset \mathcal{F}_{\mathcal{G}}$ holds, $\mathcal{K}^{s}_{0}$ is a VC-type class, i.e. there exist $A\geq e$ and $\nu\geq 1$ such that 
\begin{equation*}
    \sup_{Q}N(\mathcal{K}^{s},e_{Q},\varepsilon\|F_{s}\|_{Q,2})\leq (A/\varepsilon)^{\nu}\quad 0<\forall\varepsilon\leq1
\end{equation*}
where $F_{s}$ is an envelope function of $\mathcal{K}^{s}$. Therefore, from Corollary 5.1 of \textcite{chernozhukov2014gaussian}, for any $\sigma>0$ such that $\sup_{k\in\mathcal{K}^{s}}Pk^{2}\leq \sigma^{2}\leq \|F_{s}\|^{2}_{Q,2}$, we have
\begin{equation*}
    E[\|\mathbb{G}_{n}\|_{\mathcal{K}^{s}}]\lesssim \sqrt{\nu\sigma^{2}\log\left(\frac{A\|F_{s}\|_{Q,2}}{\sigma}\right)}+\frac{\nu\|M\|_{Q,2}}{\sqrt{n}}\log\left(\frac{A\|F_{s}\|_{Q,2}}{\sigma}\right)
\end{equation*}
where we define $M:=\max_{1\leq i\leq n}F_{s}(\varepsilon_{i},P(Z_{i}))$. Under Assumptions \ref{2}-\ref{covariate}, we have
\begin{align*}
    Pk^{2}=&E\left[\varepsilon^{2}_{i}K^{2}\left(\frac{P(Z_{i})-p}{h_{n}}\right)\left(\frac{P(Z_{i})-p}{h_{n}}\right)^{2s}\right] \\
    \leq&\sup_{p\in(0,1)}E\left[\varepsilon^{2}_{i}|P(Z_{i})=p\right]E\left[K^2\left(\frac{P(Z_{i})-p}{h_{n}}\right)\left(\frac{P(Z_{i})-p}{h_{n}}\right)^{2s}\right] \\
    \lesssim&h_{n}.
\end{align*}
Hence, for sufficiently large $n$, we can set $\sigma^{2}=Ch_{n}$ for some positive constant $C>0$. Because $\|F_{s}\|_{Q,2}$ is constant, we have
\begin{align*}
    E[\|\mathbb{G}_{n}\|_{\mathcal{K}^{s}}]&\lesssim \sqrt{\nu\sigma^{2}\log\left(\frac{A\|F_{s}\|_{Q,2}}{\sigma}\right)}+\frac{\nu\|M\|_{Q,2}}{\sqrt{n}}\log\left(\frac{A\|F_{s}\|_{Q,2}}{\sigma}\right) \\
    &=O\left(\sqrt{h_{n}\log(n)}\right)+O\left(\sqrt{\log{n}n^{-1/2}}\right) \\
    &=O\left(\sqrt{h_{n}\log(n)}\right).
\end{align*}
Hence, by Markov inequality, we have
\begin{equation*}
    \frac{1}{\sqrt{n}h_{n}}\sup_{k\in\mathcal{K}^{s}}\left|\frac{1}{\sqrt{n}}\sum_{i=1}^{n}k\left(\varepsilon_{i},\frac{P(Z_{i})-p}{h_{n}}\right)\right|=O_{p}\left(\sqrt{\frac{\log n}{nh_{n}}}\right).\quad \text{for $s=0,1$ and $2$}.
\end{equation*}
Therefore, the requirement result holds.
\end{proof}

Then, we consider the uniform convergence rate of the following term:

\begin{align*}
\frac{1}{n}\Gamma^{\prime}(\hat{V})\Omega(\hat{V})\varepsilon=
\begin{bmatrix}
\frac{1}{nh_{n}}\sum_{i=1}^{n}\varepsilon_{i}K\left(\frac{\hat{P}(Z_{i})-p}{h_{n}}\right) \\
\frac{1}{nh_{n}}\sum_{i=1}^{n}\varepsilon_{i}K\left(\frac{\hat{P}(Z_{i})-p}{h_{n}}\right)(\hat{P}(Z_{i})-p)
\\
\frac{1}{nh_{n}}\sum_{i=1}^{n}\varepsilon_{i}K\left(\frac{\hat{P}(Z_{i})-p}{h_{n}}\right)(\hat{P}(Z_{i})-p)^2
\end{bmatrix}.
\end{align*}
\begin{lemma}\label{epsilon-convergence}
Suppose Assumptions \ref{2}-\ref{covariate} hold. Then, uniformly in $p\in[a_{0},b_{0}]$, we obtain
\begin{align*}
    \frac{1}{nh_{n}}\sum_{i=1}^{n}\varepsilon_{i}K\left(\frac{\hat{P}(Z_{i})-p}{h_{n}}\right)&=o_{p}\left(n^{-c}\sqrt{\frac{\log n}{nh_{n}^{3}}}\right), \\
    \frac{1}{nh_{n}^{3}}\sum_{i=1}^{n}\varepsilon_{i}K\left(\frac{\hat{P}(Z_{i})-p}{h_{n}}\right)(\hat{P}(Z_{i})-p)&=\frac{1}{nh_{n}^{3}}\sum_{i=1}^{n}\varepsilon_{i}K\left(\frac{P(Z_{i})-p}{h_{n}}\right)(P(Z_{i})-p) \\
    &+o_{p}\left(n^{-c}\sqrt{\frac{\log n}{nh^{3}_{n}}}\right),  \\
    \frac{1}{nh_{n}^{5}}\sum_{i=1}^{n}\varepsilon_{i}K\left(\frac{\hat{P}(Z_{i})-p}{h_{n}}\right)(\hat{P}(Z_{i})-p)^2&=O_{p}\left(\sqrt{\frac{\log n}{nh^{5}_{n}}}\right),
\end{align*}
for some positive constant $c>0$.
\end{lemma}

When Lemma \ref{inverse-convergence}, \ref{unifrom_varepsilon_kernel} and \ref{epsilon-convergence} hold, we have
\begin{align*}
    &e_{2}^{\prime}\left[\frac{1}{n}\Gamma^{\prime}(\hat{V})\Omega(\hat{V})\Gamma(\hat{V})\right]^{-1}\frac{1}{n}\Gamma^{\prime}(\hat{V})\Omega(\hat{V})\varepsilon \\
    =&o_{p}(1)\frac{1}{nh_{n}}\sum_{i=1}^{n}\varepsilon_{i}K\left(\frac{\hat{P}(Z_{i})-p}{h_n}\right) \\
    +&\left(\mu_{2}(K)f_{P}(p)\right)^{-1}(1+o_{p}(n^{-c}))\frac{1}{nh_{n}^{3}}\sum_{i=1}^{n}\varepsilon_{i}K\left(\frac{\hat{P}(Z_{i})-p}{h_{n}}\right)(\hat{P}(Z_{i})-p) \\
     +&O_{p}(1)\frac{1}{nh_{n}^{3}}\sum_{i=1}^{n}\varepsilon_{i}K\left(\frac{\hat{P}(Z_{i})-p}{h_{n}}\right)(\hat{P}(Z_{i})-p)^2 \\
    =&\frac{1}{nh_{n}^{3}f_{P}(p)\mu_{2}(K)}\sum_{i=1}^{n}\varepsilon_{i}K\left(\frac{P(Z_{i})-p}{h_{n}}\right)(P(Z_{i})-p)+o_{p}\left(n^{-c}\sqrt{\frac{\log n}{nh_{n}^{3}}}\right) \tag{\ref{CL_1}}
\end{align*}
where the last inequality holds from the result of Lemma \ref{unifrom_varepsilon_kernel}.

\begin{proof}[Proof of Lemma \ref{epsilon-convergence}]
Set events $A_{\varepsilon},B_{n}$ and $C^{s}_{\varepsilon,n}$ as follows:
\begin{align*}
  A^{s}_{\varepsilon,n}:=&\left\{\sup_{p\in[a_{0},b_{0}]}\left|\frac{1}{nh_{n}}\sum_{i=1}^{n}\varepsilon_{i} k^{s}_{P,h_{n}}(\hat{P}(Z_{i}),p)\right|>\varepsilon\right\}  \quad \text{for $s=0,1$ and $2$}, \\
            B_{n}:=&\left\{\hat{P}\in\mathcal{M}_{n}^{*}\right\}, \\
            C^{s}_{\varepsilon,n}:=&\left\{\sup_{p\in[a_{0},b_{0}], r\in\mathcal{M}_{n}^{*}}\left|\frac{1}{nh_{n}}\sum_{i=1}^{n}\varepsilon_{i} k^{s}_{P,h_{n}}(r(Z_{i}),p)\right|>\varepsilon\right\}  \quad \text{for $s=0,1$ and $2$}
\end{align*}
where 
\begin{align*}
    \mathcal{M}_{n}^{*}&:=\left\{r\in\mathcal{M}|\|r-P\|_{\infty}\leq h_{n}^{2}\right\}, \\ k^{s}_{P,h_{n}}(r(z),x)&:=K\left(\frac{r(z)-x}{h_{n}}\right)\left(\frac{r(z)-x}{h_{n}}\right)^{s}-K\left(\frac{P(z)-x}{h_{n}}\right)\left(\frac{P(z)-x}{h_{n}}\right)^{s} \quad \text{for $s=0,1$ and $2$}.
\end{align*}
A straightforward calculation gives
\begin{align*}
\Pr\left(A^{s}_{\varepsilon,n}\right)&=\Pr\left(A^{s}_{\varepsilon,n}\cap B_{n}\right)+\Pr\left(A^{s}_{\varepsilon,n}\cap B_{n}^{c}\right) \\
&\leq \Pr\left(C^{s}_{\varepsilon,n}\right)+\Pr\left(B_{n}^{c}\right).
\end{align*}
Therefore, to obtain the convergence rate result, we consider the event $C^{s}_{\varepsilon,n}$ under the condition that $\Pr\left(B_{n}^{c}\right)\rightarrow 0$ as $n$ goes to infinity.

First of all, we have to show the convergence of $\Pr\left(B_{n}^{c}\right)$. By the construction of this event, we have
\begin{align*}
    \Pr\left(B_{n}^{c}\right)&=1-\Pr\left(B_{n}\right) \\
    &=1-\Pr\left(\{\hat{P}\in\mathcal{M}\}\cap \{\|\hat{P}-P\|_{\infty}\leq h_{n}^{2}\}\right) \\
    &=\Pr\left(\{\hat{P}\notin\mathcal{M}\}\cup \{\|\hat{P}-P\|_{\infty}> h_{n}^{2}\}\right)  \\
    &\leq\Pr\left(\hat{P}\notin\mathcal{M}\right)+\Pr\left( \|\hat{P}-P\|_{\infty}> h_{n}^{2}\right) \\
    &=2-\left(\Pr\left(\hat{P}\in\mathcal{M}\right)+\Pr\left( \|\hat{P}-P\|_{\infty}\leq h_{n}^{2}\right)\right).
\end{align*}
By Assumptions \ref{7-1} and \ref{7-2}, we have
\begin{equation*}
    \Pr\left(\hat{P}\in\mathcal{M}\right)+\Pr\left( \|\hat{P}-P\|_{\infty}\leq h_{n}^{2}\right)\rightarrow 2.
\end{equation*}
Therefore, $\Pr\left(B_{n}^{c}\right)\rightarrow 0$ as $n$ goes to infinity.

Second, we show that classes of functions satisfy preferable properties for the uniform convergence result. Set $\overline{\mathcal{K}^{s}_{n}}, \mathcal{M}_{\varepsilon,n}^{*}$ and $\mathcal{M}_{\varepsilon}$ as 
\begin{align*} 
        \overline{\mathcal{K}^{s}_{n}}:=&\left\{(u,z)\mapsto uk^{s}_{P,h_{n}}(r(z),x): x\in [a_{0},b_{0}], r\in\mathcal{M}_{n}^{*}\right\}\quad \text{for $s=0,1$ and $2$,} \\
    \mathcal{M}_{\varepsilon,n}^{*}:=&\left\{(u,z)\mapsto ur(z): r\in\mathcal{M}_{n}^{*}\right\}, \\
     \mathcal{M}_{\varepsilon}:=&\left\{(u,z)\mapsto ur(z): r\in\mathcal{M}\right\}.
\end{align*}
We assume $\mathcal{M}$ is a VC-type class with an envelope function 2. Because $\mathcal{M}_{n}^{*}$  is a subset of $\mathcal{M}$, it is of VC-type. Hence, by  Corollary A.1 of \textcite{chernozhukov2014gaussian},  we have
\begin{align*}
     &\sup_{Q}N(\mathcal{M}_{\varepsilon,n}^{*},e_{Q},\eta\|M^{*}\|_{Q,2}) \\
     \leq &\sup_{Q}N(\mathcal{M}_{\varepsilon},e_{Q},(\eta/2)\|M^{*}\|_{Q,2}) \\
    \leq&\left(\frac{2A}{\eta}\right)^{\nu}\quad 0<\forall\eta\leq 1
\end{align*}
where $A\geq e$ and $\nu\geq 1$ and $M^{*}(u,z)=2|u|$. Note that $A,\nu$ does not depend on $n$ due to the lack of a subscript on $\mathcal{M}$.

For any function $r$ which maps from $\mathcal{Z}$ to $(0,1)$, we define $\mathcal{K}^{s}_{n}(r)$,   $\mathcal{K}^{s}_{n,0}(r)$ and $\mathcal{K}^{s}_{n,0}$  as the following:
\begin{align*}
    \mathcal{K}^{s}_{n}(r):=&\left\{(u,z)\mapsto uK\left(\frac{r(z)-x}{h_{n}}\right)\left(\frac{r(z)-x}{h_{n}}\right)^{s}: x\in [a_{0},b_{0}]\right\}\quad \text{for $s=0,1$ and $2$}, \\
      \mathcal{K}^{s}_{n,0}(r):=&\left\{z\mapsto K\left(\frac{r(z)-x}{h_{n}}\right)\left(\frac{r(z)-x}{h_{n}}\right)^{s}: x\in [a_{0},b_{0}]\right\}\quad \text{for $s=0,1$ and $2$}, \\
      \mathcal{K}^{s}_{n,0}&:=\left\{t \mapsto K\left(\frac{t-x}{h_{n}}\right)\left(\frac{t-x}{h_{n}}\right)^{s}: x\in [a_{0},b_{0}]\right\}\quad \text{for $s=0,1$ and $2$}.
\end{align*}
By definition, $\mathcal{K}^{s}_{n,0}\circ r=\mathcal{K}^{s}_{n,0}(r)$ apparently holds. Take any finitely discrete probability measure $Q^{\prime}$ on the probability space on $\mathcal{Z}$. Set $R$ as $Q^{\prime}\circ r^{-1}$ where $r^{-1}$ the inverse image of $r$. Then, for any $k_{x},k_{x^\prime
}\in \mathcal{K}^{s}_{n,0}(r)$, there exist $k_{x}$ and $k_{x^\prime
}$ in $\mathcal{K}^{s}_{n,0}$ such that $e_{Q^{\prime}}(k_{x},k_{x^{\prime}})=e_{R}(k_{x},k_{x^{\prime}})$
holds. Therefore, for any finitely discrete probability measure $Q^{\prime}$, there exists finitely discrete probability measure $R$ such that
\begin{equation*}
    N(\mathcal{K}_{n,0}^{s}(r),e_{Q^{\prime}},\eta\|F_{s}\|_{Q^{\prime},2})=N(\mathcal{K}_{n,0}^{s},e_{R},\eta\|F_{s}\|_{R,2})\quad 0<\forall\eta\leq1
\end{equation*}
Because $\mathcal{K}^{s}_{n,0}\subset \mathcal{F}_{\mathcal{G}}$ and $\mathcal{K}^{s}_{n,0}$ is of VC-type, $\mathcal{K}_{n,0}^{s}(r)$ is also of VC-type class. In addition, by Corollary A.1 of \textcite{chernozhukov2014gaussian}, $\mathcal{K}_{n}^{s}(r)$ is of VC-type  for any $r$, i.e. there exist $B\geq e$ and $\mu\geq 1$ such that 
\begin{equation*}
    \sup_{Q}N(\mathcal{K}_{n}^{s}(r),e_{R},\eta\|F_{s}\|_{R,2})\leq (B/\eta)^{\mu}\quad 0<\forall\eta\leq1
\end{equation*}
where $F_{s}(u,z):=|u|\max_{x\in\mathbb{R}}|K(x)x^{s}|$. Therefore, for any $r$, we have
\begin{equation}\label{r-P_complex}
    \sup_{Q}N(\mathcal{K}_{n}^{s}(r)-\mathcal{K}_{n}^{s}(P),e_{Q},2\eta\|F_{s}\|_{Q,2})\leq (B/\eta)^{2\mu}\quad 0<\forall\eta\leq1
\end{equation}
where we define
\begin{align*}
    \mathcal{K}_{n}^{s}(r)-\mathcal{K}_{n}^{s}(P):=&\left\{(u,z)\mapsto uk^{s}_{P,h_{n}}(r(z),x): x\in [a_{0},b_{0}]\right\}\quad \text{for $s=0,1$ and $2$}.
\end{align*}

We need to show $\overline{\mathcal{K}^{s}_{n}}$ is of VC-type with an envelope function $\overline{F}_{s}$. Define $\overline{F}_{s}, G_{s}$ as follows:
\begin{align*}
    \overline{F}_{s}(u,z)&:=2|u|\max_{x\in\mathbb{R}}|K(x)x^{s}|, \\ 
    G_{0}&:=\max_{x\in\mathbb{R}}|K^{\prime}(x)|, \quad G_{1}:=\max_{x\in\mathbb{R}}|K^{\prime}(x)x+K(x)|, \quad G_{2}:=\max_{x\in\mathbb{R}}|K^{\prime}(x)x^2+2K(x)x|.
\end{align*}
Because $h_{n}$ converges to zero, for sufficiently large $n$, we have
\begin{equation*}
    \frac{h_{n}\max_{s\in\{0,1,2\}}\|\overline{F}_{s}\|_{Q,2}}{\|M^{*}\|_{Q,2}\min_{s\in\{0,1,2\}}G_{s}}\leq 2. \tag{*}
\end{equation*}
In the following discussion, we consider the case (*) holds.
For any $0<\eta\leq 1$, define $\eta^{*}_{s}$ as 
\begin{equation*}
    \eta^{*}_{s}= \frac{h_{n}\|\overline{F}_{s}\|_{Q,2}}{\|M^{*}\|_{Q,2}G_{s}}\frac{\eta}{2} .
\end{equation*}
 By assumption, $0<\eta_{s}^{*}\leq 1$. For any $\varepsilon_{i}k^{s}_{P,h_{n}}(r(Z_{i}),x)\in\overline{\mathcal{K}^{s}_{n}}$, there exists $(x^{*},r^{*})$ such that $\varepsilon_{i}r^{*}(Z_{i})\in \mathcal{M}_{\varepsilon,n}^{*}$ and $\varepsilon_{i}k^{s}_{P,h_{n}}(r^{*}(Z_{i}),x^{*})\in \mathcal{K}_{n}^{s}(r^{*})-\mathcal{K}_{n}^{s}(P)$ such that
\begin{align*}
&\left\|\varepsilon_{i}\left(K\left(\frac{r(Z_{i})-x}{h_{n}}\right)\left(\frac{r(Z_{i})-x}{h_{n}}\right)^{s}-K\left(\frac{P(Z_{i})-x}{h_{n}}\right)\left(\frac{P(Z_{i})-x}{h_{n}}\right)^{s}\right)\right. \\ &\left.-\varepsilon_{i}\left(K\left(\frac{r^{*}(Z_{i})-x^{*}}{h_{n}}\right)\left(\frac{r^{*}(Z_{i})-x^{*}}{h_{n}}\right)^{s}-K\left(\frac{P(Z_{i})-x^{*}}{h_{n}}\right)\left(\frac{P(Z_{i})-x^{*}}{h_{n}}\right)^{s}\right)\right\|_{Q,2} \\
\leq    &\left\|\varepsilon_{i}\left(K\left(\frac{r(Z_{i})-x}{h_{n}}\right)\left(\frac{r(Z_{i})-x}{h_{n}}\right)^{s}-K\left(\frac{P(Z_{i})-x}{h_{n}}\right)\left(\frac{P(Z_{i})-x}{h_{n}}\right)^{s}\right)\right. \\
    &\left.-\varepsilon_{i}\left(K\left(\frac{r^{*}(Z_{i})-x}{h_{n}}\right)\left(\frac{r^{*}(Z_{i})-x}{h_{n}}\right)^{s}-K\left(\frac{P(Z_{i})-x}{h_{n}}\right)\left(\frac{P(Z_{i})-x}{h_{n}}\right)^{s}\right)\right\|_{Q,2} \\
    +&\left\|\varepsilon_{i}\left(K\left(\frac{r^{*}(Z_{i})-x}{h_{n}}\right)\left(\frac{r^{*}(Z_{i})-x}{h_{n}}\right)^{s}-K\left(\frac{P(Z_{i})-x}{h_{n}}\right)\left(\frac{P(Z_{i})-x}{h_{n}}\right)^{s}\right)\right. \\
    &\left.-\varepsilon_{i}\left(K\left(\frac{r^{*}(Z_{i})-x^{*}}{h_{n}}\right)\left(\frac{r^{*}(Z_{i})-x^{*}}{h_{n}}\right)^{s}-K\left(\frac{P(Z_{i})-x^{*}}{h_{n}}\right)\left(\frac{P(Z_{i})-x^{*}}{h_{n}}\right)^{s}\right)\right\|_{Q,2} \\
    \leq&\left\|\varepsilon_{i}\left(K\left(\frac{r(Z_{i})-x}{h_{n}}\right)\left(\frac{r(Z_{i})-x}{h_{n}}\right)^{s}-K\left(\frac{r^{*}(Z_{i})-x}{h_{n}}\right)\left(\frac{r^{*}(Z_{i})-x}{h_{n}}\right)^{s}\right)\right\|_{Q,2} \\
    +&\eta\|F_{s}\|_{Q,2} \\
    \leq &\frac{G_{s}}{h_{n}}\left\|\varepsilon_{i}(r(Z_{i})-r^{*}(Z_{i}))\right\|_{Q,2}+\frac{\eta}{2}\|\overline{F}_{s}\|_{Q,2}\\
    \leq &\frac{G_{s}}{h_{n}}\eta^{*}_{s}\|M^{*}\|_{Q,2}+\frac{\eta}{2}\|\overline{F}_{s}\|_{Q,2} \\
    \leq &\eta\|\overline{F}_{s}\|_{Q,2}.
\end{align*}

Note that the third last inequality follows from the definition of $\overline{F_{s}}$. Therefore, we have
\begin{align*}
    &\sup_{Q}N(\overline{\mathcal{K}_{n}^{s}},e_{Q},\eta\|\overline{F_{s}}\|_{Q,2}) \\
    \leq &\sup_{Q}N(\mathcal{M}_{\varepsilon,n}^{*},e_{Q},\eta^{*}_{s}\|M^{*}\|_{Q,2}) \sup_{Q}N(\mathcal{K}_{n}^{s}(r)-\mathcal{K}_{n}^{s}(P),e_{Q},\eta\|F_{s}\|_{Q,2})\\ 
    \leq &\left(\frac{2A}{\eta^{*}_{s}}\right)^{\nu}\left(\frac{2B}{\eta}\right)^{2\mu} \\
    \leq &\left(\frac{1}{h_{n}}\frac{A^{*}}{\eta}\right)^{\nu+2\mu} \quad 0<\forall\eta\leq1
\end{align*}
 where we define $A^{*}$ as $(2A (\|M^{*}\|_{Q,2}\max_{s\in\{0,1,2\}}G_{s})/(\min_{s\in\{0,1,2\}}\|\overline{F}_{s}\|_{Q,2})+2B)$.

Finally, we formalize the uniform convergence problem in terms of an empirical process. For a certain type of function class $\mathcal{F}$, define $\mathbb{G}_{n}$ and $\|\mathbb{G}_{n}\|_{\mathcal{F}}$ as follows:
\begin{equation*}
    \mathbb{G}_{n}f=\frac{1}{\sqrt{n}}\sum_{i=1}^{n}(f(X_{i})-E[f(X_{i})]), \quad \|\mathbb{G}_{n}\|_{\mathcal{F}}=\sup_{f\in\mathcal{F}}|\mathbb{G}_{n}f|.
\end{equation*}
Then, we obtain
\begin{align*}
        &\sup_{p\in[a_{0},b_{0}], r\in\mathcal{M}_{n}}\left|\frac{1}{nh_{n}}\sum_{i=1}^{n}\varepsilon_{i}\left(K\left(\frac{r(Z_{i})-x}{h_{n}}\right)\left(\frac{r(Z_{i})-x}{h_{n}}\right)^{s}-K\left(\frac{P(Z_{i})-x}{h_{n}}\right)\left(\frac{P(Z_{i})-x}{h_{n}}\right)^{s}\right)\right| \\
        =&\frac{\sqrt{n}}{nh_{n}}\sup_{f\in\overline{\mathcal{K}_{n}^{s}}}\left|\frac{1}{\sqrt{n}}\sum_{i=1}^{n}f\left(\varepsilon_{i},Z_{i}\right)\right| \\
        =&\frac{\sqrt{n}}{nh_{n}}\|\mathbb{G}_{n}\|_{\overline{\mathcal{K}_{n}^{s}}}
\end{align*}
where the last equality holds from  $E[\varepsilon_{i}|Z_{i}]=0$. 
Therefore, from Corollary 5.1 of \textcite{chernozhukov2014gaussian}, for any $\sigma>0$ such that $\sup_{k\in\overline{\mathcal{K}^{s}_{n}}}Pk^{2}\leq \sigma^{2}\leq \|\overline{F}_{s}\|^{2}_{Q,2}$, we have
\begin{equation*}
    E[\|\mathbb{G}_{n}\|_{\overline{\mathcal{K}_{n}^{s}}}]\lesssim \sqrt{(\nu+2\mu)\sigma^{2}\log\left(\frac{A^{*}\|F_{s}\|_{Q,2}}{h_{n}\sigma}\right)}+\frac{(\nu+2\mu)\|M\|_{Q,2}}{\sqrt{n}}\log\left(\frac{A^{*}\|\overline{F}_{s}\|_{Q,2}}{h_{n}\sigma}\right)
\end{equation*}
where we define $M=\max_{1\leq i\leq n}\overline{F}_{s}(\varepsilon_{i},Z_{i})$. Under Assumptions \ref{2}-\ref{covariate}, we have
\begin{align*}
    Pk^{2}=&E\left[\varepsilon^{2}_{i}\left(K\left(\frac{r(Z_{i})-p}{h_{n}}\right)\left(\frac{r(Z_{i})-p}{h_{n}}\right)^{s}-K\left(\frac{P(Z_{i})-p}{h_{n}}\right)\left(\frac{P(Z_{i})-p}{h_{n}}\right)^{s}\right)^{2}\right] \\
     \leq&G_{s}^{2}\frac{\|r-P\|_{\infty}^{2}}{h_{n}^{2}}E\left[\varepsilon^{2}_{i}\right] \\
    \lesssim&h_{n}^{2}.
\end{align*}
Hence, for sufficiently large $n$, we can set $\sigma^{2}=Ch_{n}^{2}$ for some positive constant $C>0$. Because $\|\overline{F}_{s}\|_{Q,2}$ is constant, we have
\begin{align*}
E[\|\mathbb{G}_{n}\|_{\overline{\mathcal{K}_{n}^{s}}}]&\lesssim \sqrt{(\nu+\mu)\sigma^{2}\log\left(\frac{A^{*}\|\overline{F}_{s}\|_{Q,2}}{h_{n}\sigma}\right)}+\frac{(\nu+\mu)\|M\|_{Q,2}}{\sqrt{n}}\log\left(\frac{A^{*}\|\overline{F}_{s}\|_{Q,2}}{h_{n}\sigma}\right) \\
    &=O\left(\sqrt{h_{n}^{2}\log(n)}\right)+O\left(\sqrt{\log{n}n^{-1/2}}\right) \\
    &=O\left(\sqrt{h_{n}^{2}\log(n)}\right).
\end{align*}
Hence, by Markov inequality, for any $\varepsilon>0$, we have
\begin{align*}
&\Pr\left(\frac{1}{\sqrt{n}h_{n}}\sup_{k\in\overline{\mathcal{K}_{n}^{s}}}\left|\frac{1}{\sqrt{n}}\sum_{i=1}^{n}k\left(\varepsilon_{i},Z_{i}\right)\right|>\varepsilon\sqrt{\frac{\log n}{n }}\right) \\
=&\Pr\left(\frac{1}{\sqrt{h^{2}_{n}\log (n) }}\sup_{k\in\overline{\mathcal{K}_{n}^{s}}}\left|\frac{1}{\sqrt{n}}\sum_{i=1}^{n}k\left(\varepsilon_{i},Z_{i}\right)\right|>\varepsilon\right) \\
    \leq&\frac{1}{\sqrt{h^{2}_{n}\log (n)}}\frac{E[\|\mathbb{G}_{n}\|_{\overline{\mathcal{K}_{n}^{s}}}]}{\varepsilon}\quad \text{for $s=0,1$ and $2$.}
\end{align*}
Therefore, the required result holds, combined with the result of Lemma \ref{unifrom_varepsilon_kernel}.
\end{proof}

\renewcommand{\theequation}{H.\arabic{equation}} 
\setcounter{equation}{0}

\section{Proofs in Appendix  \ref{proof_gauss_stationary}}\label{proof_proof_gauss_stationary}

\begin{proof}[Proof of Lemma \ref{lemma_3_aux}]
First, we consider the $L_{2}$ diameter of $\delta_n(p)^{2}$. A straightforward calculation provides
\begin{align}
&E[\delta_n(p)^{2}] \notag \\
    =&E[(\phi_{n,p}(\varepsilon,P(Z))-\psi_{n,p}(\varepsilon,P(Z)))^{2}] \notag \\
    =&\int \left(\frac{1}{\sqrt{\int \sigma^{2}(p+h_{n}u) K^{2}\left(u\right)u^{2}f_{P}(p+h_{n}u)du}}-\frac{1}{\sqrt{\sigma^{2}(P(Z_{i}))f_{P}(P(Z_{i}))\int K^{2}(u)u^{2}du}}\right)^{2} \notag \\
    \times&\sigma^{2}(P(Z_{i}))f_{P}(P(Z_{i}))K^{2}\left(\frac{P(Z_{i})-p}{h_{n}}\right)\frac{(P(Z)-p)^{2}}{h^{3}_{n}}dP(Z_{i}) \notag \\
    =&\int \left(\frac{1}{\sqrt{\int \sigma^{2}(p+h_{n}u) K^{2}\left(u\right)u^{2}f_{P}(p+h_{n}u)du}}-\frac{1}{\sqrt{\sigma^{2}(p+h_{n}t)f_{P}(p+h_{n}t)\int K^{2}(u)u^{2}du}}\right)^{2} \notag \\
    \times&\sigma^{2}(p+h_{n}t)f_{P}(p+h_{n}t)K^{2}\left(t\right)t^{2}dt.\label{g-1}
\end{align}
Define $A(u):=u^{2}K(u)$. Then, we obtain
\begin{align*}
    &\left(\frac{1}{\sqrt{\int \sigma^{2}(p+h_{n}u) A(u)f_{P}(p+h_{n}u)du}}-\frac{1}{\sqrt{\sigma^{2}(p+h_{n}t)f_{P}(p+h_{n}t)\int A(u)du}}\right)^{2} \\
    =&\left(\frac{1}{\sqrt{\int \sigma^{2}(p+h_{n}u) A(u)f_{P}(p+h_{n}u)du\sigma^{2}(p+h_{n}t)f_{P}(p+h_{n}t)\int A(u)du}}\right)^{2} \\
    \times &\left(\sqrt{\sigma^{2}(p+h_{n}t)f_{P}(p+h_{n}t)\int A(u)du}-\sqrt{\int \sigma^{2}(p+h_{n}u) A(u)f_{P}(p+h_{n}u)du}\right)^{2}.
\end{align*}

For the numerator, we have
\begin{align*}
  &\left|\sqrt{\sigma^{2}(p+h_{n}t)f_{P}(p+h_{n}t)\int A(u)du}-\sqrt{\int \sigma^{2}(p+h_{n}u) A(u)f_{P}(p+h_{n}u)du}\right|^{2} \\
\leq &\left(\sqrt{\left|\int (\sigma^{2}(p+h_{n}t)f_{P}(p+h_{n}t)- \sigma^{2}(p+h_{n}u) f_{P}(p+h_{n}u))A(u)du\right|}\right)^{2}.
\end{align*}
Because $p\mapsto \sigma^{2}(p)f_{P}(p)$ is Lipschitz continuous, we obtain
\begin{align*}
    |\sigma^{2}(p+h_{n}t)f_{P}(p+h_{n}t)- \sigma^{2}(p+h_{n}u) f_{P}(p+h_{n}u)| 
    \leq h_{n}(|t|+|u|).
\end{align*}
Therefore, we have
\begin{align*}
    &\left(\frac{1}{\sqrt{\int \sigma^{2}(p+h_{n}u) f_{P}(p+h_{n}u)A(u)du}}-\frac{1}{\sqrt{\sigma^{2}(p+h_{n}t)f_{P}(p+h_{n}t)\int K^{2}(u)du}}\right)^{2} \\
    \leq&\frac{h_{n}\int(|t|+|u|)A(u)du}{\int \sigma^{2}(p+h_{n}u)A(u)f_{P}(p+h_{n}u)du\sigma^{2}(p+h_{n}t)f_{P}(p+h_{n}t)\int A(u)du} .
\end{align*}
It holds from \eqref{g-1} that 
\begin{equation*}
E[\delta_n(p)^{2}]
    \leq \frac{h_{n}\int\int(|t|+|u|)u^{2}K^{2}(u)t^{2}K^{2}(t)dtdu}{\int \sigma^{2}(p+h_{n}u)u^{2}K^{2}\left(u\right)f_{P}(p+hu)du\int u^{2}K^{2}(u)du} .
\end{equation*}
Therefore, $\sup_{p\in\mathcal{P}}E[(\delta_{n}(p)^{2})]=O(h_{n})$ holds.

Next, we derive the upper bound of $e_{Q}(\delta_{n}(p),\delta_{n}(p^{\prime }))$ for any  $p,p^{\prime}\in \mathcal{P}$. A straightforward calculation gives
\begin{align}
    &E\big[(\delta_n(p)-\delta_n(p^{\prime})^2\big] \notag \\
    \lesssim &E[(\phi_{n,p}(\varepsilon,P(Z))-\phi_{n,p^{\prime}}(\varepsilon,P(Z)))^{2}]+E[(\psi_{n,p}(\varepsilon,P(Z))-\psi_{n,p^{\prime}}(\varepsilon,P(Z)))^{2}].\label{ell2_distance_main}
\end{align}
For the first term of \eqref{ell2_distance_main}, by definition , we have
\begin{align*}
    &E[(\psi_{n,p}(\varepsilon,P(Z))-\psi_{n,p^{\prime}}(\varepsilon,P(Z)))^{2}]\\
    =&\int^{1}_{0}\frac{\left((P(Z_{i})-p)K\left(\frac{P(Z_{i})-p}{h_{n}}\right)-(P(Z_{i})-p^{\prime})K\left(\frac{P(Z_{i})-p^{\prime}}{h_{n}}\right)\right)^{2}}{h^{3}_{n}\sigma^{2}(P(Z_{i}))f_{P}(P(Z_{i}))\int K^{2}(u)u^{2}du}\sigma^{2}(P(Z_{i}))f_{P}(P(Z_{i}))dP(Z_{i}) \\
    =&\frac{1}{\int K^{2}(u)u^{2}du}\int^{(1-p)/h_{n}}_{-p/h_{n}}\frac{1}{h_{n}^{3}}h^{2}_{n}\left(uK\left(u\right)-\left(u+\frac{p-p^{\prime}}{h_{n}}\right)K\left(u+\frac{p-p^{\prime}}{h_{n}}\right)\right)^{2}h_{n}du.
\end{align*}
By Assumptions \ref{2}-\ref{covariate}, $uK(u)$ is Lipschitz continuous. Hence, we obtain
\begin{align}\label{tildeB2}
    &E[(\psi_{n,p}(\varepsilon,P(Z))-\psi_{n,p^{\prime}}(\varepsilon,P(Z)))^{2}] \notag \\
    \lesssim &\int^{(1-p)/h_{n}}_{-p/h_{n}}\left|\frac{p-p^{\prime}}{h_{n}}\right|^{2}du \notag  \\
    \lesssim &\frac{|p-p^{\prime}|^{2}}{h_{n}^{3}}.
\end{align}

For the second term of \eqref{ell2_distance_main}, we have
\begin{align}
    &E[(\phi_{n,p}(\varepsilon,P(Z))-\phi_{n,p^{\prime}}(\varepsilon,P(Z)))^{2}] \notag \\
    =&\int\left(\frac{(P(Z_{i})-p)K\left(\frac{P(Z_{i})-p}{h_{n}}\right)}{\sqrt{E\left[\varepsilon^{2}_{i}(P(Z_{i})-p)^{2}K^{2}\left(\frac{P(Z_{i})-p}{h_{n}}\right)\right]}}-\frac{(P(Z_{i})-p^{\prime})K\left(\frac{P(Z_{i})-p^{\prime}}{h_{n}}\right)}{\sqrt{E\left[\varepsilon_{i}^{2}(P(Z_{i})-p^{\prime})^{2}K^{2}\left(\frac{P(Z_{i})-p^{\prime}}{h_{n}}\right)\right]}}\right)^{2} \notag \\
    \times &\sigma^{2}(P(Z_{i}))f_{P}(P(Z_{i}))dP(Z_{i}) \notag \\
    \lesssim&\int\left(\frac{1}{\sqrt{E\left[\varepsilon^{2}_{i}(P(Z_{i})-p)^{2}K^{2}\left(\frac{P(Z_{i})-p}{h_{n}}\right)\right]}}-\frac{1}{\sqrt{E\left[\varepsilon_{i}^{2}(P(Z_{i})-p^{\prime})^{2}K^{2}\left(\frac{P(Z_{i})-p^{\prime}}{h_{n}}\right)\right]}}\right)^{2} \notag\\
    \times &K^{2}\left(\frac{P(Z_{i})-p}{h_{n}}\right)(P(Z_{i})-p)^{2}\sigma^{2}(P(Z_{i}))f_{P}(P(Z_{i}))dP(Z_{i})  \label{second-1-G} \\
    +&\int\left(K\left(\frac{P(Z_{i})-p}{h_{n}}\right)(P(Z_{i})-p)-K\left(\frac{P(Z_{i})-p^{\prime}}{h_{n}}\right)(P(Z_{i})-p^{\prime})\right)^{2} \notag \\
    \times&\frac{\sigma^{2}(P(Z_{i}))f_{P}(P(Z_{i}))}{E\left[\varepsilon_{i}^{2}(P(Z_{i})-p^{\prime})^{2}K^{2}\left(\frac{P(Z_{i})-p^{\prime}}{h_{n}}\right)\right]}dP(Z_{i}). \label{second-2-G}
\end{align}
First, consider the upper bound of \eqref{second-1-G}. A straightforward calculation yields
\begin{equation}\label{ell2_distance_first}
  \begin{split}
     &\left(\frac{1}{\sqrt{E\left[\varepsilon^{2}_{i}(P(Z_{i})-p)^{2}K^{2}\left(\frac{P(Z_{i})-p}{h_{n}}\right)\right]}}-\frac{1}{\sqrt{E\left[\varepsilon_{i}^{2}(P(Z_{i})-p^{\prime})^{2}K^{2}\left(\frac{P(Z_{i})-p^{\prime}}{h_{n}}\right)\right]}}\right)^{2} \\
    \leq&\left(\frac{\sqrt{E\left[\varepsilon_{i}^{2}(P(Z_{i})-p^{\prime})^{2}K^{2}\left(\frac{P(Z_{i})-p^{\prime}}{h_{n}}\right)\right]}-\sqrt{E\left[\varepsilon_{i}^{2}(P(Z_{i})-p)^{2}K^{2}\left(\frac{P(Z_{i})-p}{h_{n}}\right)\right]}}{\sqrt{E\left[\varepsilon_{i}^{2}(P(Z_{i})-p)^{2}K^{2}\left(\frac{P(Z_{i})-p}{h_{n}}\right)\right]E\left[\varepsilon_{i}^{2}(P(Z_{i})-p^{\prime})^{2}K^{2}\left(\frac{P(Z_{i})-p^{\prime}}{h_{n}}\right)\right]}}\right)^{2} \\
    =&\frac{\left(\sqrt{E\left[\varepsilon_{i}^{2}(P(Z_{i})-p^{\prime})^{2}K^{2}\left(\frac{P(Z_{i})-p^{\prime}}{h_{n}}\right)\right]}-\sqrt{E\left[\varepsilon_{i}^{2}(P(Z_{i})-p)^{2}K^{2}\left(\frac{P(Z_{i})-p}{h_{n}}\right)\right]}\right)^{2}}{h_{n}^{3}s_{n}(p)s_{n}(p^{\prime})\kappa_{2}^{2}(K)f_{P}(p)f_{P}(p^{\prime})}.
\end{split}
\end{equation}
Because $|\sqrt{a}-\sqrt{b}|\leq \sqrt{|a-b|}$ holds for any $a,b\geq 0$, we have
\begin{align*}
    &\left|\sqrt{E\left[\varepsilon_{i}^{2}(P(Z_{i})-p^{\prime})^{2}K^{2}\left(\frac{P(Z_{i})-p^{\prime}}{h_{n}}\right)\right]}-\sqrt{E\left[\varepsilon_{i}^{2}(P(Z_{i})-p)^{2}K^{2}\left(\frac{P(Z_{i})-p}{h_{n}}\right)\right]}\right|^{2} \\
    \leq &\left|E\left[\varepsilon_{i}^{2}(P(Z_{i})-p^{\prime})^{2}K^{2}\left(\frac{P(Z_{i})-p^{\prime}}{h_{n}}\right)\right]-E\left[\varepsilon_{i}^{2}(P(Z_{i})-p)^{2}K^{2}\left(\frac{P(Z_{i})-p}{h_{n}}\right)\right]\right| \\
    \leq &E\left[\varepsilon_{i}^{2}\left|(P(Z_{i})-p^{\prime})^{2}K^{2}\left(\frac{P(Z_{i})-p^{\prime}}{h_{n}}\right)-(P(Z_{i})-p)^{2}K^{2}\left(\frac{P(Z_{i})-p}{h_{n}}\right)\right|\right] .
\end{align*}
Due to Lipschitz continuity of $u^{2}K(u)$, we have
\begin{align*}
    &E\left[\varepsilon_{i}^{2}\left|(P(Z_{i})-p^{\prime})^{2}K^{2}\left(\frac{P(Z_{i})-p^{\prime}}{h_{n}}\right)-(P(Z_{i})-p)^{2}K^{2}\left(\frac{P(Z_{i})-p}{h_{n}}\right)\right|\right] \\
    \leq &h_{n}E[\varepsilon^{2}_{i}]|p-p^{\prime}|.
\end{align*}
Hence, it holds from \eqref{ell2_distance_first} and Assumptions \ref{2}-\ref{covariate} that
\begin{align}
    &\int\left(\frac{1}{\sqrt{E\left[\varepsilon^{2}_{i}(P(Z_{i})-p)^{2}K^{2}\left(\frac{P(Z_{i})-p}{h_{n}}\right)\right]}}-\frac{1}{\sqrt{E\left[\varepsilon_{i}^{2}(P(Z_{i})-p^{\prime})^{2}K^{2}\left(\frac{P(Z_{i})-p^{\prime}}{h_{n}}\right)\right]}}\right)^{2} \notag \\
    \times&K^{2}\left(\frac{P(Z_{i})-p}{h_{n}}\right)(P(Z_{i})-p)^{2}\sigma^{2}(P(Z_{i}))f_{P}(P(Z_{i}))dP(Z_{i})\notag \\
    \lesssim &\left|p-p^{\prime}\right|. \label{result_g_final_1}
\end{align}

For the second term \eqref{second-2-G}, by previous discussion,  we obtain
\begin{align}
   &\int\left(K^{2}\left(\frac{P(Z_{i})-p}{h_{n}}\right)(P(Z_{i})-p)-K^{2}\left(\frac{P(Z_{i})-p^{\prime}}{h_{n}}\right)(P(Z_{i})-p^{\prime})\right)^{2} \notag \\
    \times&\frac{\sigma^{2}(P(Z_{i})f_{P}(P(Z_{i}))}{E\left[\varepsilon_{i}^{2}(P(Z_{i})-p^{\prime})^{2}K^{2}\left(\frac{P(Z_{i})-p^{\prime}}{h_{n}}\right)\right]}dP(Z_{i}) \notag \\
    &\leq \frac{1}{h_{n}^{3}f_{P}^{2}(p^{\prime})\kappa_{2}^{2}(K)s_{n}^{2}(p^{\prime})}h_{n}^{2}\int\left|\frac{p-p^{\prime}}{h_{n}}\right|^{2} \sigma^{2}(P(Z_{i})f_{P}(P(Z_{i})))dP(Z_{i}) \notag
    \\
    &\lesssim \frac{|p-p^{\prime}|^{2}}{h_{n}^{3}}. \label{result_g_final_2}
\end{align}
Therefore, it holds from \eqref{ell2_distance_main}, \eqref{tildeB2}, \eqref{result_g_final_1} and \eqref{result_g_final_2} that
\begin{align*}
 E\big[(\delta_n(p)-\delta_n(p^{\prime}))^2\big]&\lesssim \frac{|p-p^{\prime}|^{2}}{h_{n}^{3}}+|p-p^{\prime}| \\
 &\lesssim \frac{|p-p^{\prime}|}{h_{n}^{3}} .
\end{align*}

\end{proof}

\renewcommand{\theequation}{I.\arabic{equation}} 
\setcounter{equation}{0}

\section{Proofs in Appendix  \ref{proof_main}}\label{proof_proof_main}

\begin{proof}[Proof of Lemma \ref{theorem4_aux}]

    For the first result,
  from Lemma \ref{lemma1} to Lemma \ref{lemma3} with symmetry, there exist some positive constant $d$ such that
\begin{equation*}
    \lim_{n\rightarrow\infty}\Pr\left(\ell_{n} \left|\sup_{(p,x)\in \mathcal{P}\times\mathcal{X}}\sqrt{nh_n^{3}}   \left|\frac{\widehat{MTE}(v)-MTE(v)}{\hat{s}(p)}\right|-\sup_{p\in\mathcal{P}}|\tilde{B}_{n,2}(h^{-1}_{n}p)|\right|\geq n^{-d}\right)=0. \label{prob-mean-exchange}
\end{equation*}
holds.

For the second part, by trivial calculation, we have
\begin{align*}
    &\max\left\{\Pr\left(\ell_{n}\left[\sup_{p\in\mathcal{P}}|\tilde{B}_{n,2}(h^{-1}_{n}p)|-\ell_{n}\right]<t+\varepsilon_{n}\right)-\Pr\left(\ell_{n} \left[\sup_{p\in\mathcal{P}}|\tilde{B}_{n,2}(h^{-1}_{n}p)|-\ell_{n}\right]<t\right),\right. \\
    &\left. \Pr\left(\ell_{n} \left[\sup_{p\in\mathcal{P}}|\tilde{B}_{n,2}(h^{-1}_{n}p)|-\ell_{n}\right]<t\right)-\Pr\left(\ell_{n}\left[\sup_{p\in\mathcal{P}}|\tilde{B}_{n,2}(h^{-1}_{n}p)|-\ell_{n}\right]<t-\varepsilon_{n}\right)\right\} \\
    \leq & \Pr\left(\left|\ell_{n} \left[\sup_{p\in\mathcal{P}}|\tilde{B}_{n,2}(h^{-1}_{n}p)|-\ell_{n}\right]-t\right|\leq \varepsilon_{n}\right) \\
    =&\Pr\left(\left|\sup_{p\in\mathcal{P}}|\tilde{B}_{n,2}(h^{-1}_{n}p)|-\left(\ell_{n}+\frac{t}{\ell_{n}}\right)\right|\leq \frac{\varepsilon_{n}}{\ell_{n}}\right).
\end{align*}
It holds from Corollary 2.1 of \textcite{chernozhukov2014anti} and Corollary 2.2.8 of \textcite{vanderVaartWellner96} that
\begin{align*}
&\Pr\left(\left|\sup_{p\in\mathcal{P}}|\tilde{B}_{n,2}(h^{-1}_{n}p)|-\left(\ell_{n}+\frac{t}{\ell_{n}}\right)\right|\leq \frac{\varepsilon_{n}}{\ell_{n}}\right) \\
\leq& \sup_{x\in\mathbb{R}}\Pr\left(\left|\sup_{p\in\mathcal{P}}|\tilde{B}_{n,2}(h^{-1}_{n}p)|-x\right|\leq \frac{\varepsilon_{n}}{\ell_{n}}\right)  \\
    \lesssim &\frac{\varepsilon_{n}}{\ell_{n}}\left\{E\left[\sup_{p\in\mathcal{P}}|\tilde{B}_{n,2}(h_{n}^{-1}p)|\right]+1\right\} \\
        \lesssim &\frac{\varepsilon_{n}}{\ell_{n}}\left\{E\left[\left|\tilde{B}_{n,2}(h_{n}^{-1}a_{0})\right|\right]+\int^{\infty}_{0}\sqrt{\log  N(\{\tilde{B}_{n,2}(h^{-1}_{n}p):p\in\mathcal{P}\}, e_{Q},\varepsilon)} d\varepsilon+1\right\} .
\end{align*}
By trivial calculation, $E\left[\left|\tilde{B}_{n,2}(h_{n}^{-1}a_{0})\right|\right]$ is bounded. From \eqref{tildeB2}, we have
\begin{align*}
    & \int^{\infty}_{0}\sqrt{\log  N(\{\tilde{B}_{n,2}(h^{-1}_{n}p):p\in\mathcal{P}\}, e_{Q},\varepsilon)} d\varepsilon \\
    \lesssim &\int^{2}_{0}\sqrt{\left|\log  \frac{1}{h_{n}^{3/2}\varepsilon}\right|} d\varepsilon \\
    \lesssim&O(\sqrt{\log h_{n}^{-3/2}}).
\end{align*}
Therefore, there exists $\kappa^{\prime}>0$ such that
\begin{align*}
    &\max\left\{\Pr\left(\ell_{n}\left[\sup_{p\in\mathcal{P}}|\tilde{B}_{n,2}(h^{-1}_{n}p)|-\ell_{n}\right]<t+\varepsilon_{n}\right)-\Pr\left(\ell_{n} \left[\sup_{p\in\mathcal{P}}|\tilde{B}_{n,2}(h^{-1}_{n}p)|-\ell_{n}\right]<t\right),\right. \notag \\
    &\left. \Pr\left(\ell_{n} \left[\sup_{p\in\mathcal{P}}|\tilde{B}_{n,2}(h^{-1}_{n}p)|-\ell_{n}\right]<t\right)-\Pr\left(\ell_{n}\left[\sup_{p\in\mathcal{P}}|\tilde{B}_{n,2}(h^{-1}_{n}p)|-\ell_{n}\right]<t-\varepsilon_{n}\right)\right\} \notag \\
    =&o(1). \label{anti-concentration}
\end{align*}

\end{proof}

\renewcommand{\theequation}{J.\arabic{equation}} 
\setcounter{equation}{0}

\section{Proofs in Section  \ref{main-proposition}}\label{aux_aux}

\subsection{Auxiliary Result}

\begin{proposition}\label{donny-aux}
    Let $\mathcal{G}$ be a pointwise measurable class of bounded functions such that for some constants, $\beta, \nu, C>1$, $\sigma\leq1/(8C)$ and the envelope function $G$, the following four conditions hold:
\begin{align*}
        &E[G^{2}(X)]\leq \beta^{2}, \\
        &\text{$\mathcal{G}$ is of VC-type with the envelope $G$,} \\
        &\sigma^{2}_{0}:=\sup_{g\in\mathcal{G}}E[g^{2}(X)]\leq \sigma^{2}, \\
        &\sup_{g\in\mathcal{G}}\|g\|_{\infty} \leq \frac{1}{2\sqrt{\nu+1}}\sqrt{\frac{n\sigma^{2}}{\log(\beta\vee 1/\sigma)}}.
\end{align*}
Then, for a universal constant $A_{3}$, it holds
\begin{equation*}
    E\left\|\sum_{i=1}^{n}\varepsilon_{i}g(X_{i})\right\|_{\mathcal{G}}\leq A_{3}\sqrt{\nu n\sigma^{2}\log(\beta\vee 1/\sigma)}.
\end{equation*}
\begin{proof}[Proof of Proposition \ref{donny-aux}]
The proof is similar to that of Proposition A.1 in \textcite{einmahl2000empirical}. They also use the definition of VC-type in Remark \ref{def-VC}. However, in the proof of Proposition A.1 in \textcite{einmahl2000empirical}, they only use finitely discrete probability measures to take the supremum in terms of a probability measure. Therefore, even though we replace the definition of VC-type with Definition \ref{VC-type-def}, we can ensure that the result of Proposition A.1 in \textcite{einmahl2000empirical} still holds.
\end{proof}
\end{proposition}

Before providing helpful lemmas, we introduce a definition of bounded variation.
\begin{dfn}[bounded variation]
A function $f:\mathbb{R}\mapsto\mathbb{R}$ is of bounded variation if the quantity,
\begin{equation*}
    v(f):=\sup\left\{\sum_{i=1}^{n}|f(x_{i})-f(x_{i-1})|: -\infty<x_{0}<\cdots<x_{n}<\infty, n\in\mathbb{N}\right\}
\end{equation*}
is finite.
    
\end{dfn}

\begin{lemma}\label{Lemma3.6.11}
    Let f be a $K$-Lipschitz continuous function with compact support where $K$>0. Then, there exists a non-decreasing and continuous $h$ such that $0 \leq h(x)\leq v(f)$ and $g$ that is $1$-Lipschitz continuous on the interval $[0, v(f)]$. The composite function $g$ of $h$ is equal to $f$, namely $f=g\circ h$.
\end{lemma}

\begin{proof}[Proof of Lemma \ref{Lemma3.6.11}]
    Define $h(x)$ as
    \begin{equation*}
        h(x):=v_{x}(f)
    \end{equation*}
    where we define
    \begin{equation*}
        v_{x}(f):=\sup\left\{\sum_{i=1}^{n}|f(x_{i})-f(x_{i-1})|: -\infty<x_{0}<\cdots<x_{n}\leq x, n\in\mathbb{N}\right\}.
    \end{equation*}
     The function $h$ reflects the total variation of the function $f$ till point $x$. By construction, $h$ is non-decreasing, and the value of $h$ is in the closed interval $[0,v(f)]$. Moreover, $h$ is continuous. Indeed, by definition of $h$, we obtain the following inequality for $x,y\in\mathbb{R}$:
     \begin{align*}
         |h(y)-h(x)|&\leq \sup\left\{\sum_{i=1}^{n}|f(x_{i})-f(x_{i-1})|: \min\{x,y\}\leq x_{0}<\cdots<x_{n}\leq \max\{x,y\}, n\in\mathbb{N}\right\} \\
         &\leq L|x-y|,
     \end{align*}
     where the last inequality follows the assumption that $f$ is $K$-Lipschitz continuous. We immediately obtain the continuity of $h$ from the above discussion.
     
     By the definition of bounded variation, for any $x<y$, we have
    \begin{align*}\label{h-f_relation}
        &h(x)+|f(y)-f(x)|\leq h(y) \notag \\
        \Leftrightarrow& |f(y)-f(x)|\leq h(y)-h(x). 
    \end{align*}
    If $h(y)$ is equal to $h(x)$ even when $x$ is not equal to $y$, $f(y)=f(x)$. Hence, we can define $g(u)$ as the value of $f$ on any of points $h^{-1}\{u\}$ where $u$ is in the range of $h$. Furthermore, because $h$ is non-decreasing and continuous, it holds from the compact support of $f$ that the range of $h$ is equal to $[0,v(f)]$. Therefore, $g$ is defined on $[0,v(f)]$.

    By construction of $g$, for any $u,v \in[0,v(f)]$, we have the following inequality:
    \begin{align*}
        |g(u)-g(v)|&=|g(h(x))-g(h(y))|\\
        &=|f(x)-f(y)| \\
        &\leq |h(x)-h(y)|  \\
        &= |u-v| 
    \end{align*}
    where we define $x,y$ as $h(x):=u$ and  $h(y):=v$, respectively. Therefore, we obtain the required result.
\end{proof}

\begin{proposition}\label{proposition3.6.12}
    Let f be a $K$-Lipschitz continuous function with a compact support where $K>0$. Consider the collection
    \begin{equation*}
        \mathcal{F}=\{x\mapsto f(tx-s):t>0, s\in\mathbb{R}\}.
    \end{equation*}
    Then, $\mathcal{F}$ is of VC-type with envelope $\|f\|_{\infty}$ and pointwise measurable.
\end{proposition}

\begin{proof}[Proof of Proposition \ref{proposition3.6.12}]
Consider the following class $\mathcal{F}_{c}$ defined as 
\begin{equation*}
    \mathcal{F}_{c}=\{x\mapsto f(tx-s):t>0, t,s\in\mathbb{Q}\}.
\end{equation*}
Apparently, $\mathcal{F}_{c}$ is countable and, for every $f\in\mathcal{F}$, there exists a sequence $f_{m}\in\mathcal{F}_{c}$ with $f_{m}(x)\rightarrow f(x)$ for every $x$. Therefore, $\mathcal{F}$ is pointwise measurable.

We prove that $\mathcal{F}$ defined in the statement is of VC-type. First, we consider the composition function of $f$ obtained in Lemma \ref{Lemma3.6.11}. We can write $f=g\circ h$, where $h$ is non-decreasing and $g$ is $1$-Lipschitz continuous. For any function in $\mathcal{F}$ we have $f(tx-s)=g(h(tx-s))$. Hence, we need to consider the following class of functions:
\begin{equation*}
    \mathcal{H}=\{x\mapsto h(tx-s):t>0, s\in\mathbb{R}\}.
\end{equation*}
Because of the non-decreasing property of $h$, we can define its generalized inverse $h^{-1}(u)$ for any value $u\in [0,v(f)]$ as below:
\begin{equation*}
    h^{-1}(u):=\inf\{x\in\mathbb{R} | h(x)\geq u\}.
\end{equation*}
The subgraph of a particular element indexed by $t,s$ in $\mathcal{H}$ will be
\begin{align*}
    G_{t,s} &=\{(x,u)\in\mathbb{R}\times[0,v(f)]:u\leq h(tx-s)\} \\
    &=\{(x,u)\in\mathbb{R}\times[0,v(f)]:h^{-1}(u)\leq tx-s\} \\
&= \{(x,u)\in\mathbb{R}\times[0,v(f)]:h^{-1}(u)-tx+s\leq 0\}.
\end{align*}
Thus, all possible subgraphs in $\mathcal{H}$ is the set $\mathcal{G}=\{G_{t,s}:t>0, s\in\mathbb{R}\}$. So the VC dimension of $\mathcal{H}$ is the VC
dimension of the set $\mathcal{G}$.

Because each element $G_{t,s}$ is determined by function $h^{-1}(u)-tx+s\leq 0$, one can easily see that
\begin{align*}
    \mathcal{G}&\subset\mathcal{V} \\
\mathcal{V}&=\{\mathbb{V}_{a,b,c}:a,b,c\in\mathbb{R}\} \\
\mathbb{V}_{a,b,c}&=\{(x,u) :ah^{-1}(u)+bx+c\leq 0\}.
\end{align*}
Because $\mathcal{V}$ is formed by the vector space of 3 functions ,$(x,u)\mapsto h^{-1}(u), (x,u)\mapsto x, (x,u)\mapsto1$, it holds from Lemma 2.6.15, Lemma 2.6.17 (i) and Lemma 2.6.18 (iii) of \textcite{vanderVaartWellner96} that its VC dimension is at most 5.
So the VC dimension of $\mathcal{G}\subset\mathcal{V}$ will be at most 5, which implies that $\mathcal{H}$ is a VC-type class with VC dimension at most five, and we can pick the envelope function of $\mathcal{H}$ to be constant $v(f)$. By Theorem 2.6.7 of \textcite{vanderVaartWellner96}, for $0<\varepsilon<1$, there exist positive numbers $A_{0}$ such that
\begin{equation*}
    N(\mathcal{H},e_{Q},\varepsilon v(f))\leq \left(\frac{A_{0}}{\varepsilon}\right)^{8}
\end{equation*}
for any probability measure $Q$. 

Let $f_{t,s}, h_{t,s}$ denote $f(tx-s)$ and $ h(tx-s)$, respectively. Using the fact that $g$ is $1$-Lipschitz continuous, for any $h_{t,s},h_{t^{\prime},s^{\prime}}$ with $e_{Q}(h_{t,s},h_{t^{\prime},s^{\prime}})\leq\varepsilon v(f)$, we have
\begin{align*}
    e_{Q}(f_{t,s},f_{t^{\prime},s^{\prime}})&=e_{Q}(g(h_{t,s}),g(h_{t^{\prime},s^{\prime}})) \\
    &\leq e_{Q}(h_{t,s},h_{t^{\prime},s^{\prime}})\\
    &\leq \varepsilon v(f).
\end{align*}
Hence, any $\varepsilon v(f)$-cover of $\mathcal{H}$ induces an $\varepsilon v(f)$-cover of $\mathcal{F}=g\circ\mathcal{H}$. Let $C$ denote $\|f\|_{\infty}/v(f)$. By definition, we have $ C \ leq 1$. Hence, it holds that
\begin{equation*}
    N(\mathcal{F} ,e_{Q},\varepsilon \|f\|_{\infty})=N(\mathcal{F} ,e_{Q},\varepsilon C v(f)).
\end{equation*}
Then, for $0<\varepsilon<1$, we have
\begin{align*}
N(\mathcal{F} ,e_{Q},\varepsilon \|f\|_{\infty})=&N(\mathcal{F} ,e_{Q},\varepsilon C v(f)) \\
\leq &N(\mathcal{H},e_{Q},\varepsilon Cv(f)) \\
\leq &\left(\frac{B_{0}}{\varepsilon}\right)^{8}
\end{align*}
where $B_{0}=A_{0}/C$. Therefore, $\mathcal{F}$ is of VC-type with the envelope $\|f\|_{\infty}$.

\end{proof}

\subsection{Proof of Corollary \ref{donny}}
\begin{proof}[Proof of Corollary \ref{donny}]
    Because \textcite{Dony_Einmahl_Mason_2016} show Corollary \ref{donny} holds under the same assumptions except for the definition of VC-type, we need to guarantee  Theorem 2 of \textcite{Dony_Einmahl_Mason_2016} still holds with Definition \ref{VC-type-def}. 

In the proof of Theorem 2 of \textcite{Dony_Einmahl_Mason_2016}, the definition of VC-type is only used to verify Proposition A.1 in \textcite{einmahl2000empirical}. As we proved in Proposition \ref{donny-aux}, the result of  Proposition A.1 in \textcite{einmahl2000empirical} still holds with the weaker definition of VC-type. Therefore, Theorem 2 of \textcite{Dony_Einmahl_Mason_2016} still holds with Definition \ref{VC-type-def}. 
\end{proof}

\subsection{Proof of Lemma \ref{VC-type}}

\begin{proof}[Proof of Lemma \ref{VC-type}]
From Lemma \ref{Lemma3.6.11} and Proposition \ref{proposition3.6.12}, this result immediately holds.
\end{proof}

\subsection{Proof of Lemma \ref{union-VC}}
\begin{proof}[Proof of Lemma \ref{union-VC}]
Let $H(x)$ denote $\max\{F(x),G(x)\}$. By definition of the covering number, for any finite discrete probability measure $Q$, we have
\begin{equation*}
    N( \mathcal{F}\cup\mathcal{G},e_{Q},\varepsilon\|H\|_{Q,2})\leq N( \mathcal{F},e_{Q},\varepsilon\|F\|_{Q,2})+N( \mathcal{G},e_{Q},\varepsilon\|G\|_{Q,2}).
\end{equation*}
Because $\mathcal{F}$ and $\mathcal{G}$ are of VC-type with envelopes $F$ and $G$ respectively, there exists $A,B,\nu$ and $\eta$ satisfying
\begin{equation*}
   \sup_{Q}N( \mathcal{F},e_{Q},\varepsilon\|F\|_{Q,2})\leq \left(\frac{A}{\varepsilon}\right)^{\nu}\quad \text{and}\quad \sup_{Q}N(\mathcal{G},e_{Q},\varepsilon\|G\|_{Q,2})\leq \left(\frac{B}{\varepsilon}\right)^{\eta}.
\end{equation*}
where the supremum is taken over all finitely discrete probability measures. Hence, from the above inequality, we obtain
\begin{align*}
     \sup_{Q}N(\varepsilon\|H\|_{Q,2}, \mathcal{F}\cup\mathcal{G},e_{Q})&\leq \sup_{Q}N(\varepsilon\|F\|_{Q,2}, \mathcal{F},e_{Q})+\sup_{Q}N(\varepsilon\|G\|_{Q,2}, \mathcal{G},e_{Q}) \\
     &\leq \left(\frac{A}{\varepsilon}\right)^{\nu}+\left(\frac{B}{\varepsilon}\right)^{\eta} \\
     &\leq \left(\frac{C_{0}}{\varepsilon}\right)^{\mu}
\end{align*}
where we define $\mu:=\max\{\nu,\eta\}$ and $C_{0}=(A^{\mu}+B^{\mu})^{1/\mu}$. Therefore, the required statement does hold.
\end{proof}

\subsection{Proof of Lemma \ref{union-pointwise-measurable}}

\begin{proof}[Proof of Lemma \ref{union-pointwise-measurable}]
There exist countable subsets $\mathcal{F}_{c}$ and $\mathcal{G}_{c}$ for $\mathcal{F}$ and $\mathcal{G}$, respectively.  For any $h\in\mathcal{F}\cup\mathcal{G}$, there exists $h_{m}\in \mathcal{F}_{c}\cup\mathcal{G}_{c}$ such that $h_{m}\rightarrow h$ for any $x$ because $\mathcal{F}$ and $\mathcal{G}$ are pointwise measurable. Because $\mathcal{F}_{c}\cup\mathcal{G}_{c}$ are countable, $\mathcal{F}\cup\mathcal{G}$ are pointwise measurable.
\end{proof}

\subsection{Proof of Lemma \ref{expectation}}

 \begin{proof}[Proof of Lemma \ref{expectation}]
Consider the existence of $E[g_{n,h}]$ for each element in $\mathcal{G}$. For $0\leq s\leq 5$ and $0\leq \ell\leq 5$, we have
\begin{align*}
   E\left[\frac{1}{nh_{n}}\sum_{i=1}^{n}\left|\left(\frac{P(Z_{i})-p}{h_{n}}\right)^{s}K^{(\ell)}\left(\frac{P(Z_{i})-p)}{h_{n}}\right)\right|\right]&=E\left[\frac{1}{h_{n}}\left|\left(\frac{P(Z_{i})-p}{h_{n}}\right)^{s}K^{(\ell)}\left(\frac{P(Z_{i})-p}{h_{n}}\right)\right|\right] \\
   &\leq \int |u|^{s}|K^{(\ell)}\left(u\right)|f_{P}(p+uh_{n})du .
\end{align*}
By Assumptions \ref{2}-\ref{covariate}, there exists $M$ such that $|f_{P}(u)|\leq M$ for any $u\in(0,1)$. Therefore, we have
\begin{equation*}
    E\left[\frac{1}{nh_{n}}\sum_{i=1}^{n}\left|\left(\frac{P(Z_{i})-p}{h_{n}}\right)^{s}K^{(\ell)}\left(\frac{P(Z_{i})-p}{h_{n}}\right)\right|\right]  \leq M\int |u|^{s}|K^{(\ell)}\left(u\right)|du .
\end{equation*}
By definition of the kernel function, $ |u|^{s}|K^{(\ell)}\left(u\right)|$ is bounded. Therefore, we obtain
\begin{equation*}
    E\left[\frac{1}{nh_{n}}\sum_{i=1}^{n}\left|\left(\frac{P(Z_{i})-p}{h_{n}}\right)^{s}K^{(\ell)}\left(\frac{P(Z_{i})-p}{h_{n}}\right)\right|\right] < \infty.
\end{equation*}
\end{proof}

\subsection{Proof of Proposition \ref{kernel-uniform}}
\begin{proof}[Proof of Proposition \ref{kernel-uniform}]
Because we define $\mathcal{G}$ as follows
\begin{equation*}
    \mathcal{G}:=\{|(-x)^{s}K^{(\ell)}(-x)|\}\quad \text{for $0\leq s\leq 3$ and $0\leq \ell \leq 5$},
\end{equation*}
from the result of Corollary \ref{donny} and Lemma \ref{expectation}, all we need to show is that $\mathcal{G}$ satisfies (G.i), (G.ii), (F.i), and (F.ii). Under Assumptions \ref{2}-\ref{covariate}, (G.i) and (G.ii) immediately hold.

Now we verify (F.i) and (F.ii). $ \mathcal{F}_{\mathcal{G}}$ is defined as follows:
\begin{equation*}
    \mathcal{F}_{\mathcal{G}}=\{s\mapsto g(z-s\lambda):\lambda\geq1, z\in\mathbb{R}\ \text{and}\ g\in\mathcal{G}\}.
\end{equation*}
By definition of VC class, if the class of functions  $\{s\mapsto g(z+s\lambda):\lambda\geq1, z\in\mathbb{R}\}$ for each $g\in\mathcal{G}$ is of VC class with the envelope $\|g\|_{\infty}$, we can show $\{s\mapsto g(z-s\lambda):\lambda\geq1, z\in\mathbb{R}\}$ is also of VC class with the envelope $\|g\|_{\infty}$ through the argument of change of variables. The class of functions \(\{s \mapsto g(z + s\lambda) : \lambda \geq 1, z \in \mathbb{R}\}\) is a subset of the larger class \(\{s \mapsto g(s\lambda - z) : \lambda > 0, z \in \mathbb{R}\}\). By Assumptions \ref{2}-\ref{covariate}, each \(g \in \mathcal{G}\) is \(K\)-Lipschitz continuous and has compact support. Hence, applying Lemma \ref{VC-type}, the class \(\{s \mapsto g(z + s\lambda) : \lambda \geq 1, z \in \mathbb{R}\}\) is of VC-type for each \(g\), with an envelope function given by \(\|g\|_{\infty}\). Therefore, by Lemma \ref{union-VC}, the class of functions  $ \mathcal{F}_{\mathcal{G}}$ is of VC-type with the envelope $\sup_{g\in\mathcal{G}}\|g\|_{\infty}$. Furthermore, apparently, the class of functions  $\{s\mapsto g(z+s\lambda):\lambda\geq1, z\in\mathbb{R}\}$ is pointwise measurable. By Lemma \ref{union-pointwise-measurable}, the class of functions  $ \mathcal{F}_{\mathcal{G}}$ is pointwise measurable.

\end{proof}

\subsection{Proof of Proposition \ref{aux-1}}
\begin{proof}[Proof of Proposition \ref{aux-1}]
    Under Assumptions \ref{2}-\ref{covariate}, because $K(\cdot)$ is six times differentiable, we have
    \begin{align*}
        K\left(\frac{\hat{P}(Z_{i})-p}{h_{n}}\right)-K\left(\frac{P(Z_{i})-p}{h_{n}}\right)
   =&\sum_{j=1}^{5}\frac{1}{j!}K^{(j)}\left(\frac{P(Z_{i})-p}{h_{n}}\right)\left(\frac{\hat{P}(Z_{i})-P(Z_{i})}{h_{n}}\right)^{j} \\
   +&K^{(6)}\left(u_{i}\right)\frac{1}{6!}\left(\frac{\hat{P}(Z_{i})-P(Z_{i})}{h_{n}}\right)^{6}
    \end{align*}
where $u_{i}=\lambda_{i}((\hat{P}(Z_{i})-p)/h_{n})+(1-\lambda_{i})((P(Z_{i})-p)/h_{n})$. Hence, we have,
\begin{align*}
   &  \sup_{p\in[a_{0},b_{0}]}\left(\frac{1}{n}\sum_{i=1}^{n}\frac{1}{h_{n}}\left|\left(K\left(\frac{\hat{P}(Z_{i})-p}{h_{n}}\right)-K\left(\frac{P(Z_{i})-p}{h_{n}}\right)\right)\left(\frac{P(Z_{i})-p}{h_{n}}\right)^{s}\right|\right) \\
   \leq &\sum_{j=1}^{5}\frac{1}{j!}\left(\frac{\max_{1\leq i\leq n}|\hat{P}(Z_{i})-P(Z_{i})|}{h_{n}}\right)^{j}\sup_{p\in[a_{0},b_{0}]}\left(\frac{1}{n}\sum_{i=1}^{n}\frac{1}{h_{n}}\left|K^{(j)}\left(\frac{P(Z_{i})-p}{h_{n}}\right)\left(\frac{P(Z_{i})-p}{h_{n}}\right)^{s}\right|\right) \\
   +&\frac{\sup |K^{(6)}(u)|}{6!}\left(\frac{\max_{1\leq i\leq n}|\hat{P}(Z_{i})-P(Z_{i})|}{h_{n}}\right)^{6}\frac{1}{h_{n}^{s}}.
\end{align*}
By Proposition \ref{kernel-uniform}, for any integer $0\leq s\leq 5$ and $0\leq j\leq 5$, it holds that
\begin{equation*}
    \sup_{p\in[a_{0},b_{0}]}\left(\frac{1}{n}\sum_{i=1}^{n}\frac{1}{h_{n}}\left|K^{(j)}\left(\frac{P(Z_{i})-p}{h_{n}}\right)\left(\frac{P(Z_{i})-p}{h_{n}}\right)^{s}\right|\right)=O_{p}(1).
\end{equation*}
Therefore, for any integer $0\leq s\leq 5$, we have 
\begin{equation*}
    \sup_{p\in[a_{0},b_{0}]}\left(\frac{1}{n}\sum_{i=1}^{n}\frac{1}{h_{n}}\left|\left(K\left(\frac{\hat{P}(Z_{i})-p}{h_{n}}\right)-K\left(\frac{P(Z_{i})-p}{h_{n}}\right)\right)\left(\frac{P(Z_{i})-p}{h_{n}}\right)^{s}\right|\right)=o_{p}(h_{n}).
\end{equation*}

For the second statement, by the triangular inequality, for any integer $0\leq s\leq 5$ and $0\leq \ell$, we have
\begin{align*}
    &\sup_{p\in[a_{0},b_{0}]}\left(\frac{1}{n}\sum_{i=1}^{n}\frac{1}{h_{n}}\left|K\left(\frac{\hat{P}(Z_{i})-p}{h_{n}}\right)\left(\frac{P(Z_{i})-p}{h_{n}}\right)^{s}\left(\frac{\hat{P}(Z_{i})-P(Z_{i})}{h_{n}^{3}}\right)^{\ell}\right|\right) \\
    \leq & \sup_{p\in[a_{0},b_{0}]}\left(\frac{1}{n}\sum_{i=1}^{n}\frac{1}{h_{n}}\left|\left(K\left(\frac{\hat{P}(Z_{i})-p}{h_{n}}\right)-K\left(\frac{P(Z_{i})-p}{h_{n}}\right)\right)\left(\frac{P(Z_{i})-p}{h_{n}}\right)^{s}\right|\right) \\
    &\left(\frac{\max_{1\leq i\leq n}|\hat{P}(Z_{i})-P(Z_{i})|}{h_{n}^{3}}\right)^{\ell} \\
    +&\sup_{p\in[a_{0},b_{0}]}\left(\frac{1}{n}\sum_{i=1}^{n}\frac{1}{h_{n}}\left|K\left(\frac{P(Z_{i})-p}{h_{n}}\right)\left(\frac{P(Z_{i})-p}{h_{n}}\right)^{s}\right|\right)\left(\frac{\max_{1\leq i\leq n}|\hat{P}(Z_{i})-P(Z_{i})|}{h_{n}^{3}}\right)^{\ell} .
\end{align*}
From the above discussion, we have
\begin{equation*}
    \sup_{p\in[a_{0},b_{0}]}\left(\frac{1}{n}\sum_{i=1}^{n}\frac{1}{h_{n}}\left|\left(K\left(\frac{\hat{P}(Z_{i})-p}{h_{n}}\right)-K\left(\frac{P(Z_{i})-p}{h_{n}}\right)\right)\left(\frac{P(Z_{i})-p}{h_{n}}\right)^{s}\right|\right)=o_{p}(h_{n}).
\end{equation*}
Therefore, it holds from Proposition \ref{kernel-uniform} and Assumption \ref{2}-\ref{covariate} that
\begin{equation*}
    \sup_{p\in[a_{0},b_{0}]}\left(\frac{1}{n}\sum_{i=1}^{n}\frac{1}{h_{n}}\left|K\left(\frac{\hat{P}(Z_{i})-p}{h_{n}}\right)\left(\frac{P(Z_{i})-p}{h_{n}}\right)^{s}\left(\frac{\hat{P}(Z_{i})-P(Z_{i})}{h_{n}^{3}}\right)^{\ell}\right|\right)=o_{p}(1).
\end{equation*}

\end{proof}

\subsection{Proof of Proposition \ref{aux-2}}
\begin{proof}[Proof of Proposition \ref{aux-2}]

For the first two arguments, because we have
\begin{align*}
    \max&\left\{\sup_{p\in[a_{0},b_{0}]}\left| \frac{1}{nh_{n}}\sum_{i=1}^{n}K\left(\frac{\hat{P}(Z_{i})-p}{h_{n}}\right)\left(\frac{\hat{P}(Z_{i})-p}{h_{n}}\right)^{s}X_{i,\ell}\right|,\right. \\
&\left.\sup_{p\in[a_{0},b_{o}]}\left|\frac{1}{nh_{n}}\sum_{i=1}^{n}K\left(\frac{\hat{P}(Z_{i})-p}{h_{n}}\right)\left(\frac{\hat{P}(Z_{i})-p}{h_{n}}\right)^{s}X_{i,\ell}P(Z_{i})\right|\right\} \\
    \leq &\sup_{p\in[a_{0},b_{o}]}\left(\frac{1}{nh_{n}}\sum_{i=1}^{n}K\left(\frac{\hat{P}(Z_{i})-p}{h_{n}}\right)\left|\frac{\hat{P}(Z_{i})-p}{h_{n}}\right|^{s}|X|_{i,\ell}\right),
\end{align*}
 we only need to show 
\begin{equation}
    \sup_{p\in[a_{0},b_{o}]}\left(\frac{1}{nh_{n}}\sum_{i=1}^{n}K\left(\frac{\hat{P}(Z_{i})-p}{h_{n}}\right)\left|\frac{\hat{P}(Z_{i})-p}{h_{n}}\right|^{s}|X|_{i,\ell}\right)=O_{p}(1) \label{aux-2-first-two}
\end{equation}
for any $0\leq s\leq 2$.

A straightforward calculation gives
\begin{align}
&\sup_{p\in[a_{0},b_{0}]}\left|\left(\frac{1}{nh_{n}}\sum_{i=1}^{n}K\left(\frac{\hat{P}(Z_{i})-p}{h_{n}}\right)\left|\frac{\hat{P}(Z_{i})-p}{h_{n}}\right|^s|X_{i,\ell}|\right)-f_{P}(p)E[|X_{i\ell}||P(Z_{i})=p]E[K(u)|u|^s]\right| \notag \\
    \leq & \sup_{p\in[a_{0},b_{0}]}\left|\frac{1}{nh_{n}}\sum_{i=1}^{n}\left(K\left(\frac{\hat{P}(Z_{i})-p}{h_{n}}\right)\left|\frac{\hat{P}(Z_{i})-p}{h_{n}}\right|^s-K\left(\frac{P(Z_{i})-p}{h_{n}}\right)\left|\frac{P(Z_{i})-p}{h_{n}}\right|^s\right)|X_{i,\ell}|\right| \label{aux-2-1}\\
    +& \sup_{p\in[a_{0},b_{0}]}\left|\frac{1}{nh_{n}}\sum_{i=1}^{n}\left(K\left(\frac{P(Z_{i})-p}{h_{n}}\right)\left|\frac{P(Z_{i})-p}{h_{n}}\right|^s|X_{i,\ell}|\right.\right.  \notag \\
        -&\left.\left.E\left[K\left(\frac{P(Z_{i})-p}{h_{n}}\right)\left|\frac{P(Z_{i})-p}{h_{n}}\right|^s|X_{i,\ell}|\right]\right)\right| \label{aux-2-2}\\
    +&\sup_{p\in[a_{0},b_{0}]}\left|E\left[\frac{1}{h_{n}}K\left(\frac{P(Z_{i})-p}{h_{n}}\right)\left|\frac{P(Z_{i})-p}{h_{n}}\right|^s|X_{i,\ell}|\right]-f_{P}(p)E[|X_{i\ell}||P(Z_{i})=p]E[K(u)|u|^s]\right|. \label{aux-2-3}
\end{align}

For \eqref{aux-2-1}, it holds from Assumptions \ref{2}-\ref{covariate} that $K(u)|u|^s$ is a Lipschitz function. Hence, we have
\begin{align*}
    &\sup_{p\in[a_{0},b_{0}]}\left|\frac{1}{nh_{n}}\sum_{i=1}^{n}\left(K\left(\frac{\hat{P}(Z_{i})-p}{h_{n}}\right)\left|\frac{\hat{P}(Z_{i})-p}{h_{n}}\right|^s-K\left(\frac{P(Z_{i})-p}{h_{n}}\right)\left|\frac{P(Z_{i})-p}{h_{n}}\right|^s\right)|X_{i,\ell}|\right| \\
    \lesssim  &\sup_{p\in[a_{0},b_{0}]}\left|\left(\frac{1}{nh_{n}}\sum_{i=1}^{n}|X_{i,\ell}|\left(\frac{\hat{P}(Z_{i})-P(Z_{i})}{h_{n}}\right)\right)\right| \\
    \leq &\frac{\|\hat{P}(Z_{i})-P(Z_{i})\|_{\infty}}{h^{2}_{n}}\frac{1}{n}\sum_{i=1}^{n}|X_{i,\ell}|.
\end{align*}
Therefore, a rate of convergence for \eqref{aux-2-1} is $o_{p}(h_{n})$. We can establish the uniform CLT result for \eqref{aux-2-2}, as in the proof of Lemma \ref{unifrom_varepsilon_kernel}. Therefore, we have
\begin{align*}
    &\sup_{p\in[a_{0},b_{0}]}\left|\frac{1}{nh_{n}}\sum_{i=1}^{n}\left(K\left(\frac{P(Z_{i})-p}{h_{n}}\right)\left|\frac{P(Z_{i})-p}{h_{n}}\right||X_{i,\ell}|-E\left[K\left(\frac{P(Z_{i})-p}{h_{n}}\right)\left|\frac{P(Z_{i})-p}{h_{n}}\right||X_{i,\ell}|\right]\right)\right| \\
        =&O_{p}\left(\sqrt{\frac{\log n}{nh_{n}}}\right).
\end{align*}
For \eqref{aux-2-3}, by traditional arguments for the expected value of the kernel function, we achieve 
\begin{align*}
&\sup_{p\in[a_{0},b_{0}]}\left|E\left[\frac{1}{h_{n}}K\left(\frac{P(Z_{i})-p}{h_{n}}\right)\left|\frac{P(Z_{i})-p}{h_{n}}\right|^s|X_{i,\ell}|\right]-f_{P}(p)E[|X_{i\ell}||P(Z_{i})=p]E[K(u)|u|^s]\right|\\
      =&O_{p}(1).
\end{align*}
Combining all the above results, the required result \eqref{aux-2-first-two} holds.

For the last statement, we have
\begin{align*}
    &\sup_{p\in[a_{0},b_{0}]}\left|\frac{1}{nh_{n}}\sum_{i=1}^{n}K\left(\frac{\hat{P}(Z_{i})-p}{h_{n}}\right)\left(\frac{\hat{P}(Z_{i})-p}{h_{n}}\right)^sX_{i,\ell}\left(\frac{P(Z_{i})-\hat{P}(Z_{i})}{h_{n}^{3}}\right)^{k}\right| \\
    \leq &\sup_{p\in[a_{0},b_{0}]}\left|\frac{1}{nh_{n}}\sum_{i=1}^{n}\left(K\left(\frac{\hat{P}(Z_{i})-p}{h_{n}}\right)\left(\frac{\hat{P}(Z_{i})-p}{h_{n}}\right)^s-K\left(\frac{P(Z_{i})-p}{h_{n}}\right)\left(\frac{P(Z_{i})-p}{h_{n}}\right)^s\right)\right. \\
    \times &\left.X_{i,\ell}\left(\frac{P(Z_{i})-\hat{P}(Z_{i})}{h_{n}^{3}}\right)^{k}\right| \\
    + &\sup_{p\in[a_{0},b_{0}]}\left|\frac{1}{nh_{n}}\sum_{i=1}^{n}K\left(\frac{P(Z_{i})-p}{h_{n}}\right)\left(\frac{P(Z_{i})-p}{h_{n}}\right)^{s}X_{i,\ell}\left(\frac{P(Z_{i})-\hat{P}(Z_{i})}{h^{3}_{n}}\right)^{k}  \right|.
\end{align*}
For the first term, due to the fact that $K(u)|u|^s$ is a Lipschitz function, we have
\begin{align*}
    &\sup_{p\in[a_{0},b_{0}]}\left|\frac{1}{nh_{n}}\sum_{i=1}^{n}\left(K\left(\frac{\hat{P}(Z_{i})-p}{h_{n}}\right)\left(\frac{\hat{P}(Z_{i})-p}{h_{n}}\right)^s-K\left(\frac{P(Z_{i})-p}{h_{n}}\right)\left(\frac{P(Z_{i})-p}{h_{n}}\right)^s\right)\right. \\
    \times &\left.X_{i,\ell}\left(\frac{P(Z_{i})-\hat{P}(Z_{i})}{h_{n}^{3}}\right)^{k}\right| \\
     \lesssim &\frac{\|P(Z_{i})-\hat{P}(Z_{i})\|^{k+1}_{\infty}}{h^{3k+2}_{n}}\frac{1}{n}\sum_{i=1}^{n}|X_{i,\ell}|.
\end{align*}
Because $E[|X_{i,\ell}|]$ is finite, by the law of large numbers, the first term converges to zero in probability. For the second term, we obtain
\begin{align*}
    &\sup_{p\in[a_{0},b_{0}]}\left|\frac{1}{nh_{n}}\sum_{i=1}^{n}K\left(\frac{P(Z_{i})-p}{h_{n}}\right)\left(\frac{P(Z_{i})-p}{h_{n}}\right)^{s}X_{i,\ell}\left(\frac{P(Z_{i})-\hat{P}(Z_{i})}{h^{3}_{n}}\right)^{k}  \right| \\
    \leq &\left(\frac{\|P(Z_{i})-\hat{P}(Z_{i})\|_{\infty}}{h^{3}_{n}}\right)^{k}\sup_{p\in[a_{0},b_{0}]}\left(\frac{1}{nh_{n}}\sum_{i=1}^{n}|X_{i,\ell}|K\left(\frac{P(Z_{i})-p}{h_{n}}\right)\left|\frac{P(Z_{i})-p}{h_{n}}\right|^s\right).
\end{align*}
It follows from the previous discussion and Assumptions \\ref{2}-\ref{covariate} that the second term also converges to zero in probability. Therefore, the required result holds.

\end{proof}

\section{Additional Monte Carlo Simulations}

\subsection{Computation Time}
\label{sec_monte_time}

Table \ref{monte_result_time} compares the computation time required for constructing confidence bands across the three designs. While the statistical performance of our method and the bootstrap procedure is comparable, our method has a substantial computational advantage. The bootstrap method requires approximately 89 seconds per Monte Carlo iteration in all designs, whereas our procedure requires only 8-10 seconds on average. Thus, our method is roughly ten times faster than the bootstrap approach. Moreover, the computation time of our method is more stable across replications, as reflected in the smaller standard deviations of computation time.

\begin{table}[thb]
\centering
\caption{Computation time comparison}
\label{monte_result_time}
\begin{tabular}{cccc}
\hline
 & Design $\Sigma_1$ & Design $\Sigma_2$ & Design $\Sigma_3$ \\
\hline
Mean time (Bootstrap) & 89.99 & 89.87 & 89.68 \\
SD time (Bootstrap) & 2.26 & 2.23 & 2.31 \\
Mean time (Ours) & 8.38 & 8.57 & 9.75 \\
SD time (Ours) & 0.81 & 0.87 & 1.16 \\
Relative speedup & 10.74 & 10.49 & 9.20 \\
\hline
\end{tabular}

\begin{minipage}{\linewidth}
\smallskip
\footnotesize
Note: The reported computation time for our method includes the elapsed time spent on estimating the MTE function and calculating the analytical critical value. The reported computation time for the bootstrap procedure corresponds to the elapsed time spent on the multiplier-bootstrap resampling step, which repeatedly recomputes the local-polynomial estimator, without accounting for the time of calculating the MTE function itself. Thus, the comparison highlights the additional computational burden induced by bootstrap resampling.

\end{minipage}
\end{table}

This computational advantage is particularly important in applications involving large samples, repeated simulations, or empirical procedures requiring repeated construction of confidence bands. While bootstrap confidence bands can also achieve valid uniform inference, they require repeated nonparametric estimation across bootstrap samples, resulting in substantially higher computational costs. In contrast, our approach avoids iterative resampling procedures and therefore provides a computationally efficient alternative for uniform inference on the MTE function.

\subsection{Bandwidth Selection}
\label{sec:sensitivity_bw}

In this subsection, we investigate the sensitivity of our procedure to the choice of bandwidth. Throughout the above simulation, the bandwidth is selected according to the rate $h_n = Cn^{-2/13}$. Since finite-sample performance may depend on the choice of bandwidth, we examine the robustness of our procedure to alternative bandwidth rates.

Specifically, we consider the bandwidth sequence
$h_n = C n^{-\eta},
$ where $C$ is the constant term used in the previous simulation design. The exponent $\eta$ varies across $
\eta \in
\{
\frac{1}{6},
\frac{1}{6.25},
\frac{1}{6.5},
\frac{1}{6.75},
\frac{1}{7}
\}
$.
Note that Assumption \ref{8} requires $1/6 < \eta < 1/7$.
This exercise allows us to assess whether the finite-sample performance of the proposed confidence band is sensitive to moderate deviations.

\begin{table}[thb]
\centering
\caption{Bandwidth selection: Results of Monte Carlo simulations}
\label{monte_result_bandwidth}
\begin{tabular}{ccccc}
\hline
& CP (90\%) & Mean Crit. Val. (90\%) & CP (95\%) & Mean Crit. Val. (95\%) \\
\hline

\multicolumn{5}{c}{Design B-1: $\eta=1/6$. The mean bandwidth is 0.049} \\
Pointwise & 0.201 & 1.65 & 0.434 & 1.96 \\
Gumbel & 0.985 & 3.50 & 0.998 & 4.01 \\
Bootstrap & 0.913 & 2.89 & 0.958 & 3.15 \\
Ours & 0.903 & 2.82 & 0.950 & 3.06 \\
\hline

\multicolumn{5}{c}{Design B-2: $\eta=1/6.25$. The mean bandwidth is 0.052} \\
Pointwise & 0.214 & 1.65 & 0.445 & 1.96 \\
Gumbel & 0.987 & 3.52 & 0.996 & 4.04 \\
Bootstrap & 0.918 & 2.89 & 0.959 & 3.15 \\
Ours & 0.900 & 2.80 & 0.950 & 3.05 \\
\hline

\multicolumn{5}{c}{Design B-3: $\eta=1/6.5$. The mean bandwidth is 0.055} \\
Pointwise & 0.235 & 1.65 & 0.450 & 1.96 \\
Gumbel & 0.981 & 3.54 & 0.997 & 4.08 \\
Bootstrap & 0.916 & 2.89 & 0.962 & 3.16 \\
Ours & 0.895 & 2.78 & 0.943 & 3.03 \\
\hline

\multicolumn{5}{c}{Design B-4: $\eta=1/6.75$. The mean bandwidth is 0.057} \\
Pointwise & 0.224 & 1.65 & 0.462 & 1.96 \\
Gumbel & 0.986 & 3.57 & 0.998 & 4.11 \\
Bootstrap & 0.902 & 2.90 & 0.953 & 3.17 \\
Ours & 0.874 & 2.77 & 0.930 & 3.01 \\
\hline

\multicolumn{5}{c}{Design B-5: $\eta=1/7$. The mean bandwidth is 0.060} \\
Pointwise & 0.234 & 1.65 & 0.449 & 1.96 \\
Gumbel & 0.987 & 3.59 & 0.998 & 4.16 \\
Bootstrap & 0.908 & 2.90 & 0.962 & 3.19 \\
Ours & 0.875 & 2.75 & 0.936 & 3.00 \\
\hline

\end{tabular}

\begin{minipage}{\linewidth}
\smallskip
\footnotesize
Note: This table shows the empirical coverage probabilities and critical values for the pointwise, Gumbel, bootstrap, and proposed confidence bands under alternative bandwidth rates. The bandwidth is given by $h_n=Cn^{-\eta}$. The reported critical values and mean bandwidths are rounded to two and three decimal places, respectively.
\end{minipage}
\end{table}

\begin{table}[thb]
\centering
\caption{Bandwidth selection: Computation time comparison}
\label{monte_result_time_bandwidth}
\begin{tabular}{cccccc}
\hline
& Design B-1 & Design B-2 & Design B-3 & Design B-4 & Design B-5 \\
\hline
Mean time (Bootstrap) & 90.34 & 90.04 & 89.67 & 89.35 & 89.29 \\
SD time (Bootstrap) & 2.19 & 2.27 & 2.38 & 2.37 & 2.41 \\
Mean time (Ours) & 8.23 & 8.22 & 8.29 & 8.31 & 8.32 \\
SD time (Ours) & 0.78 & 0.81 & 0.83 & 0.79 & 0.81 \\
Relative speedup & 10.98 & 10.96 & 10.81 & 10.75 & 10.74 \\
\hline
\end{tabular}

\begin{minipage}{\linewidth}
\smallskip
\footnotesize
Note: This table reports the computation time, in seconds, for a single iteration of the Monte Carlo experiment. All values are rounded to two decimal places. See the note for Table \ref{monte_result_time} for more details.
\end{minipage}
\end{table}

Table \ref{monte_result_bandwidth} reports the results under alternative bandwidth rates. Overall, the findings are remarkably stable across all bandwidth choices. Our method continues to attain empirical coverage probabilities close to the nominal levels, while the pointwise confidence band substantially undercovers and the Gumbel-based confidence band remains conservative throughout.

As the bandwidth increases, the empirical coverage probabilities of our method remain reasonably close to the nominal level across all specifications, while they decline modestly. For example, the 95\% coverage probability decreases from $0.950$ in Design B-1 to $0.936$ in Design B-5. The bootstrap method also exhibits very stable performance, with 95\% coverage probabilities ranging between $0.953$ and $0.962$. These findings suggest that both procedures are fairly robust to moderate changes in the bandwidth rate.

The critical values exhibit a systematic but quantitatively small dependence on the bandwidth choice. As the bandwidth increases, the critical values associated with our method decrease slightly, from $3.06$ in Design B-1 to $3.00$ in Design B-5 for the 95\% confidence band. By contrast, the corresponding Gumbel critical values increase from $4.01$ to $4.16$. Despite these changes, the qualitative conclusions remain unchanged across all bandwidth specifications.

Table \ref{monte_result_time_bandwidth} reports the computation times. The computational burden is nearly identical across all bandwidth choices. The bootstrap procedure requires approximately 90 seconds per Monte Carlo iteration, whereas our method requires approximately 8.3 seconds. Consequently, our procedure remains roughly ten times faster than the bootstrap approach while delivering comparable statistical performance.

Overall, the results indicate that the finite-sample performance of the proposed confidence band is not particularly sensitive to moderate perturbations of the bandwidth rate around the theoretically motivated choice.

\subsection{Heteroskedasticity}
\label{sec:sensitivity_hetero}

In this subsection, we investigate the robustness of our procedure to deviations from homoskedasticity of the unobservables. To assess the impact of heteroskedasticity, we modify the structure of the unobservables while preserving the analytical tractability of the MTE. Specifically, we consider the following decomposition:
\begin{equation*}
U_1 = a_1 U_D + \sigma_1(X)\varepsilon_1, \qquad
U_0 = a_0 U_D + \sigma_0(X)\varepsilon_0,
\end{equation*}
where $U_D, \varepsilon_1, \varepsilon_0 \sim \mathcal{N}(0,1)$ are mutually independent. The functions $\sigma_1(X)$ and $\sigma_0(X)$ introduce heteroskedasticity through dependence on the covariate $X$. In our implementation, we set
\begin{equation*}
\sigma_d(X) = \sigma_d \exp\left(\frac{\xi_d X}{2}\right), \quad d \in \{0,1\},
\end{equation*}
where $\xi_d$ controls the degree of heteroskedasticity.

 \begin{table}[thb]
        \centering
    \caption{Heteroskedasticity: Results of Monte Carlo simulations}
           \label{monte_result_hetero}
        \begin{tabular}{ccccc}
        \hline
             & CP (90\%) & Mean Crit. Val. (90\%) & CP (95\%) &  Mean Crit. Val. (95\%) \\
             \hline 
    \multicolumn{5}{c}{Design H-1: $(\xi_0,\xi_1) = (0, 0)$. The mean bandwidth is 0.054} \\
            Pointwise & 0.228 & 1.65 & 0.457  & 1.96 \\
            Gumbel &  0.991  & 3.54 & 0.998 & 4.07 \\
            Bootstrap & 0.918  & 2.88 & 0.959 & 3.15 \\ 
            Ours  & 0.900 & 2.78 & 0.951 & 3.03 \\ 
            \hline
    \multicolumn{5}{c}{Design H-2: $(\xi_0,\xi_1) = (0, 0.5)$. The mean bandwidth is 0.055} \\
            Pointwise &  0.220 & 1.65  & 0.427 & 1.96 \\
            Gumbel & 0.988  &  3.54  & 0.999 & 4.08 \\
            Bootstrap & 0.918 &   2.89 &  0.956 & 3.16 \\ 
            Ours  & 0.892 &  2.78  & 0.948& 3.03  \\ 
            \hline
    \multicolumn{5}{c}{Design H-3: $(\xi_0,\xi_1) = (0.5, 0)$.  The mean bandwidth is 0.054} \\
            Pointwise & 0.227 & 1.65  & 0.451 & 1.96 \\
            Gumbel &  0.983 &  3.53 & 0.995 & 4.06  \\
            Bootstrap & 0.903  & 2.88  & 0.951  & 3.15 \\ 
            Ours  & 0.886 & 2.79  & 0.942 & 3.03  \\ 
            \hline

      \multicolumn{5}{c}{Design H-4: $(\xi_0,\xi_1) = (0.5, 0.5)$. The mean bandwidth is 0.054} \\
            Pointwise &  0.219 & 1.65  & 0.440 & 1.96 \\
            Gumbel & 0.983  &  3.54 &  0.997 & 4.07 \\
            Bootstrap & 0.902 &  2.89  & 0.948 & 3.15 \\ 
            Ours  & 0.884 & 2.78
            & 0.930 &  3.03 \\ 
            \hline
        \end{tabular}

\begin{minipage}{\linewidth}

    Note: This table shows the empirical coverage and critical values for pointwise, Gumbel, bootstrap, and our method. The reported critical values and mean bandwidths are rounded to two and three decimal places, respectively.
\end{minipage}

\end{table}

\begin{table}[thb]
\centering
\caption{Heteroskedasticity: Computation time comparison}
\label{monte_result_time_hetero}
\begin{tabular}{ccccc}
\hline
 & Design H-1 & Design H-2 & Design H-3 & Design H-4 \\
\hline
Mean time (Bootstrap) & 89.81 & 90.18 & 90.20 & 90.17 \\
SD time (Bootstrap) & 2.13 & 2.23 & 2.13 &2.15 \\
Mean time (Ours) & 8.44 & 8.43 & 8.47 & 8.42 \\
SD time (Ours) & 0.82 & 0.86 & 0.84 & 0.85 \\
Relative speedup & 10.64 & 10.70 & 10.65 & 10.71 \\
\hline
\end{tabular}

\begin{minipage}{\linewidth}
\smallskip
\footnotesize
Note: This table reports the computation time in seconds for a single iteration of the Monte Carlo experiment.  All values are rounded to two decimal places. See the note for Table \ref{monte_result_time} for more details.
\end{minipage}
\end{table}

This specification implies that
\begin{equation*}
E[U_d^2 \mid X] = a_d^2 + \sigma_d^2(X),
\end{equation*}
which depends on $X$. At the same time, the selection component of the MTE remains unchanged:
\begin{equation*}
E[U_1 - U_0 \mid V = p, X] = (a_1 - a_0)\Phi^{-1}(p),
\end{equation*}
since the idiosyncratic shocks $\varepsilon_1$ and $\varepsilon_0$ are independent of $U_D$. Consequently, the MTE retains the closed-form expression
\begin{equation*}
MTE(p,x_1) = 0.2+0.3 x_1+ (a_1 - a_0)\Phi^{-1}(p),
\end{equation*}
which is identical to the baseline specification with $\Sigma_1$ when $a_1-a_0=0.3$. In summary, this design allows us to isolate the effect of heteroskedasticity on estimation and inference, without altering the true MTE.

In the heteroskedastic design, we vary the parameter $\xi_d$ to control the strength of heteroskedasticity while keeping the resulting MTE equal to that in the baseline specification with $\Sigma=\Sigma_1$. Specifically, we consider four configurations of $(\xi_0,\xi_1)$ given by $(0,0)$, $(0.5,0)$, $(0,0.5)$, and $(0.5,0.5)$, corresponding to the homoskedastic benchmark, heteroskedasticity only in the untreated outcome, heteroskedasticity only in the treated outcome, and heteroskedasticity in both potential outcomes, respectively.

Table \ref{monte_result_hetero} reports the results under heteroskedasticity. Overall, the results are very similar to those reported in Table \ref{monte_result_1} under the baseline specification with $\Sigma=\Sigma_1$. Our method continues to attain empirical coverage probabilities close to the nominal levels, while the pointwise confidence band substantially undercovers, and the Gumbel-based confidence band remains conservative across all designs. 

As the degree of heteroskedasticity increases, the empirical coverage probabilities for both our method and the bootstrap method decline slightly. In particular, the coverage probabilities under Designs H-3 and H-4 are marginally lower than those under Designs H-1 and H-2. Nevertheless, the overall differences are relatively small, suggesting that the proposed procedure is reasonably robust to heteroskedasticity.

Table \ref{monte_result_time_hetero} compares the computation times under the heteroskedastic designs. The computational patterns are nearly identical across all four cases. The bootstrap procedure takes approximately 90 seconds per Monte Carlo iteration, whereas our method takes approximately 8.5 seconds. Thus, our procedure remains roughly ten times faster than the bootstrap approach while achieving comparable statistical performance.

\section{Additional empirical results}
\label{sec:add_emp}

We assess the extent to which the propensity score retains meaningful variation after conditioning on observed covariates. Because the covariate vector is high-dimensional, it is not feasible to evaluate overlap at exact covariate values. Instead, we conduct a diagnostic analysis by restricting attention to communities whose covariate values fall within the central region of their respective distributions. Specifically, we consider the distance to the nearest electricity substation, baseline household density, distance to the nearest road, distance to the nearest town, and the baseline poverty rate. These variables are particularly relevant because they capture infrastructure access, remoteness, and local economic conditions, and are therefore likely to be correlated with the land-gradient instrument. Examining the support of the estimated propensity score within these relatively homogeneous subsets provides a practical assessment of whether the instrument continues to generate sufficient variation in the propensity score after conditioning on observed covariates.

Table \ref{tab:overlap_diagnostic} reports the results of this diagnostic exercise. For each covariate, we restrict attention to communities whose covariate values lie between the 37.5th and 62.5th percentiles of the corresponding distribution. The columns \emph{Overlap L} and \emph{Overlap U} report the lower and upper bounds of the central overlap region of the estimated propensity score distributions for the treated and untreated communities. Specifically, the lower bound is defined as the larger of the two 2.5th percentiles, while the upper bound is defined as the smaller of the two 97.5th percentiles. The column \emph{Overlap Length} reports the width of this common support region.

The results indicate substantial overlap across all covariates considered. Even after conditioning on relatively homogeneous subsets of communities, the overlap lengths range from 0.439 to 0.534. These findings suggest that the land-gradient instrument continues to generate meaningful variation in the propensity score after conditioning on observed covariates, supporting the practical relevance of the overlap condition required for MTE estimation.

\begin{table}[thb]
\centering
\caption{Diagnostic analysis of propensity-score overlap}
\label{tab:overlap_diagnostic}
\begin{tabular}{lcccc}
\hline
Covariate & $N$ & Overlap L & Overlap U & Overlap Length \\
\hline
Distance to substation & 454 & 0.061 & 0.530 & 0.469  \\
Household density & 454 & 0.045 & 0.506 & 0.460 \\
Distance to road & 454 & 0.028 & 0.528 & 0.500 \\
Distance to town & 454 & 0.040 & 0.478 & 0.439 \\
Poverty rate & 454 & 0.034 & 0.568 & 0.534 \\
\hline
\end{tabular}

\begin{minipage}{\linewidth}
\smallskip
\footnotesize
Note: For each covariate, we restrict attention to communities whose covariate values lie between the 37.5th and 62.5th percentiles of the corresponding distribution. Overlap L and Overlap U denote the lower and upper bounds of the central overlap region of the estimated propensity score distributions for treated and untreated communities. Overlap Length is defined as the difference between Overlap U and Overlap L.
\end{minipage}
\end{table}

\printbibliography[heading=subbibliography]

\end{refsection}

\end{document}